\documentclass[11pt,a4paper]{article}
\pdfoutput=1
\usepackage{jheppub}
\usepackage{amsthm,bbm,graphicx,slashed}

\newtheorem*{lem}{Lemma}

\DeclareMathOperator\ee{e}
\DeclareMathOperator\tr{tr}
\DeclareMathOperator\re{Re}
\DeclareMathOperator\im{Im}
\DeclareMathOperator\Pf{Pf}
\DeclareMathOperator\ind{ind}
\DeclareMathOperator\sgn{sgn}
\DeclareMathOperator\diag{diag}

\newcommand\fh{\tilde f}
\newcommand\vh{\tilde v}
\newcommand\dd{\mathcal D}
\newcommand\del{\partial}
\newcommand{\bep}{\begin{pmatrix}} 
\newcommand{\eep}{\end{pmatrix}}
\newcommand{\Sp}{\text{Sp}}
\newcommand{\SU}{\text{SU}}
\newcommand{\SO}{\text{SO}}
\newcommand{\su}{\text{su}}
\newcommand{\U}{\text{U}}
\newcommand{\connect}{\text{conn}}
\renewcommand{\L}{\textbf{L}}
\newcommand{\I}{\textbf{I}}
\renewcommand{\H}{\textbf{H}}
\newcommand\mL{\mathcal L}
\newcommand\rsv{\rho_\text{sv}}
\newcommand\rsvm{\hat\rho_\text{sv}}
\newcommand\eps{\varepsilon}
\renewcommand\epsilon{\varepsilon}
\renewcommand\phi\varphi
\newcommand{\1}{\mathbbm{1}}
\newcommand\Z{\mathbbm{Z}}
\newcommand\N{\mathbbm{N}}
\newcommand\R{\mathbbm{R}}
\newcommand\C{\mathbbm{C}}
\newcommand\ev[1]{\left\langle#1\right\rangle}
\renewcommand{\bar}[1]{\overline{#1}} 

\title{Singular values of the Dirac operator in dense QCD-like theories}
\author[a]{Takuya Kanazawa,}
\author[a]{Tilo Wettig,}
\author[b]{and Naoki Yamamoto}
\affiliation[a]{Department of Physics, University of Regensburg, 93040
    Regensburg, Germany}
\affiliation[b]{Institute for Nuclear Theory, University of Washington,
    Seattle, Washington 98195-1550, USA}
\emailAdd{takuya.kanazawa@ur.de}
\emailAdd{tilo.wettig@ur.de}
\emailAdd{nyama@u.washington.edu}

\subheader{\hfill\normalsize October 25, 2011}

\abstract{We study the singular values of the Dirac operator in dense
  QCD-like theories at zero temperature. The Dirac singular values are
  real and nonnegative at any nonzero quark density.  The scale of
  their spectrum is set by the diquark condensate, in contrast to the
  complex Dirac eigenvalues whose scale is set by the chiral
  condensate at low density and by the BCS gap at high density.  We
  identify three different low-energy effective theories with diquark
  sources applicable at low, intermediate, and high density, together
  with their overlapping domains of validity.  We derive a number of
  exact formulas for the Dirac singular values, including
  Banks-Casher-type relations for the diquark condensate,
  Smilga-Stern-type relations for the slope of the singular value
  density, and Leutwyler-Smilga-type sum rules for the inverse
  singular values.  We construct random matrix theories and determine
  the form of the microscopic spectral correlation functions of the
  singular values for all nonzero quark densities.  We also derive a
  rigorous index theorem for non-Hermitian Dirac operators.  Our
  results can in principle be tested in lattice simulations.}

\keywords{Spontaneous symmetry breaking, chiral Lagrangians, 
random matrix theory, sum rules, index theorem}

\begin{document}

\maketitle

\section{Introduction}

A prominent feature of quantum chromodynamics (QCD) in the vacuum is
the dynamical breaking of chiral symmetry through the formation of a
chiral condensate $\langle \bar \psi \psi \rangle$.  While an
analytical demonstration of this phenomenon from the underlying theory
is still lacking, a number of studies have deepened our understanding
of its nature.  For example, spontaneous chiral symmetry breaking is
reflected in the infrared limit of the Dirac eigenvalue spectrum
through the Banks-Casher relation \cite{Banks:1979yr} and the
Smilga-Stern relation \cite{Smilga:1993in}.  Also, a universal
finite-volume domain (called the $\epsilon$-regime) has been
discovered \cite{Gasser:1987ah} in which the theory becomes
zero-dimensional and is governed entirely by the global symmetries of
the system.  It was shown to lead to an infinite number of constraints
on the low-lying Dirac eigenvalues, known as the Leutwyler-Smilga
spectral sum rules \cite{Leutwyler:1992yt}.

Thanks to the universality, rigorous results for the spectral
properties of the Dirac operator can be derived from a much simpler
chiral random matrix theory (RMT) which has the same global
symmetry-breaking pattern as the QCD vacuum (see
\cite{Verbaarschot:2000dy,Akemann:2007rf} for reviews).  This has
advanced our knowledge of the low-lying Dirac eigenvalues to the level
of microscopic spectral correlation functions, which eventually made
it possible to determine the value of the chiral condensate to high
precision by first-principle lattice QCD simulations
\cite{Fukaya:2009fh}, thus numerically verifying that chiral symmetry
is spontaneously broken in the QCD vacuum.  There are also other
QCD-like theories that can be described by RMT.  They can be
classified by the Dyson index $\beta$.  If the Dirac operator commutes
with an anti-unitary operator $T$, then $\beta=1$ if $T^2=\1$ and
$\beta=4$ if $T^2=-\1$.  Otherwise $\beta=2$.  Equivalently,
$\beta=1$, $2$, or $4$ if the representation of the gauge group in
which the fermions transform is pseudoreal, complex, or real,
respectively.\footnote{This can be shown following the considerations
  in sections 21.2--21.4 of \cite{Georgi:1999}.}

At large quark chemical potential $\mu$,\footnote{Throughout most of
  the paper we will use the terms ``chemical potential'' and
  ``density'' interchangeably, unless the distinction is important.
  Nonzero density always implies nonzero $\mu$.  However, the reverse
  is not always true, e.g., in two-color QCD at small $\mu$ the
  density is zero for $\mu<m_\pi/2$, where $m_\pi$ is the pion mass.}
QCD is believed to exhibit another nonperturbative phenomenon: The
attractive interaction between quarks near the Fermi surface leads to
color superconductivity characterized by a diquark condensate $\langle
\psi \psi \rangle$ (see \cite{Alford:2007xm, Fukushima:2010bq} for
recent reviews).  In the color-superconducting phases, gluons acquire
masses through the Anderson-Higgs mechanism, and weakly interacting
quarks acquire a Bardeen-Cooper-Schrieffer (BCS) gap $\Delta$.  In
particular, in the color-flavor-locked (CFL) phase
\cite{Alford:1998mk} of $N_f=3$ QCD at high density, chiral symmetry
is spontaneously broken in much the same way as in the QCD vacuum, but
the order parameter is the four-quark condensate rather than the
conventional chiral condensate.  Using effective-theory techniques, it
was shown \cite{Yamamoto:2009ey} that the distribution of the smallest
Dirac eigenvalues in the CFL phase is entirely governed by global
symmetries and that the relevant scale of the low-lying spectrum is
set by the BCS gap $\Delta$.

Unfortunately, at $\mu\ne 0$ the complex phase of the fermion
determinant in the QCD partition function has so far hampered lattice
computations of physical quantities such as the BCS gap and the
diquark condensate. This is an example of the so-called sign problem,
which is encountered in many areas of physics.  However, several
QCD-like theories are devoid of the sign problem even at nonzero $\mu$
\cite{Alford:1998sd} and develop a nonzero BCS gap and a diquark
condensate (or pion condensate) signaling superfluidity.  Examples
include two-color QCD ($\beta=1$), $\Sp(2N_c)$ gauge theories with
fundamental fermions and an arbitrary number $N_c$ of colors
($\beta=1$), $\SU(N_c)$ gauge theories with adjoint fermions
($\beta=4$), and $\SO(N_c)$ gauge theories with fundamental fermions
($\beta=4$).  Another example is three-color QCD at nonzero isospin
chemical potential ($\beta=2$).%
\footnote{To evade the sign problem, one has to assume an even number
  of flavors with degenerate quark masses in the $\beta=1$ and
  $\beta=2$ cases.  On the other hand, we can take an arbitrary number
  of flavors with nondegenerate quark masses in the $\beta=4$ cases.}
In spite of the clear differences between these theories and
three-color QCD at nonzero quark chemical potential, they share some
common features with \emph{bona fide} QCD, e.g., the same mechanism of
quark-quark pairing that leads to superfluidity or color
superconductivity at nonzero $\mu$, and the universality of their
phase diagrams \cite{Hanada:2011ju}.  This observation leads us to
expect that we can obtain some insights into the properties of
realistic QCD matter through a deeper understanding of these QCD-like
theories.
 
Indeed, the aforementioned work on the Dirac eigenvalues in the CFL
phase has been successfully generalized to two-color QCD at high
density ($\beta=1$) \cite{Kanazawa:2009ks}, where the BCS pairing was
shown to have sizable consequences on the behavior of the small Dirac
eigenvalues, as represented through new Leutwyler-Smilga-type sum
rules.  (This analysis permits a straightforward extension to theories
with $\beta=2$ and $4$.)  Moreover, the chiral random matrix theory
that describes the Dirac eigenvalue distribution on the scale
characterized by the BCS gap was constructed \cite{Kanazawa:2009en}
and solved analytically \cite{Akemann:2009fc,Akemann:2010tv}.  These
results make it possible in principle to compute the BCS gap from
Dirac spectra on the lattice, which would directly verify the BCS-type
superfluid phase of QCD-like theories.

Although the BCS gap $\Delta$ and the diquark condensate
$\ev{\psi\psi}$ are closely intertwined at high density, they are
\emph{a priori} different objects.  While it is unclear to what lower
densities the BCS-type pairing with well-defined $\Delta$ will
persist, the diquark condensate is still nonvanishing and breaks
global symmetries even at lower densities in those QCD-like theories.
Some of the Nambu-Goldstone (NG) modes are diquarks charged under
baryon number (with mass $m_{\pi}$) so that they exhibit Bose-Einstein
condensation (BEC) for $\mu>m_{\pi}/2$
\cite{Kogut:1999iv,Kogut:2000ek,Son:2000xc,Son:2000by,Zhang:2010kn,Cherman:2010jj,Cherman:2011mh,Hanada:2011ju}.%
\footnote{In QCD at nonzero isospin chemical potential $\mu_I$, the
  usual pions, which do not have baryon charge but isospin charge,
  show BEC, i.e., $\langle \bar d\gamma_5 u \rangle \neq 0$, for
  $\mu_I>m_{\pi}$. In this case, however, it is still possible to view
  this as a BEC of ``diquarks" if we switch $d \to d'\equiv C\bar
  d^T$.}  The Bose-Einstein condensate at small $\mu$ consists of
strongly coupled bound diquarks, in contrast to the weakly coupled
BCS-type diquarks at large $\mu$.  Despite this qualitative
difference, however, the quantum numbers of the condensate are the
same at small and large $\mu$, so it is natural to expect that there
is no phase transition between the two limits.  This phenomenon is
called (relativistic) BEC-BCS crossover.  It passes a number of
nontrivial tests \cite{Son:2000xc,Son:2000by,Splittorff:2000mm}, and
its possible realization in dense quark matter in QCD has been
investigated in various model calculations
\cite{Abuki:2001be,Nawa:2005sb,Abuki:2006dv,Sun:2007fc,He:2007yj,Kitazawa:2007zs,Brauner:2008td,Baym:2008me,Abuki:2008tj,Brauner:2009gu,Matsuzaki:2009kb,Andersen:2010vu,Abuki:2010jq,He:2010nb,Zhang:2010ct,Wang:2010uj}
(see also
\cite{Nishida:2005ds,Deng:2006ed,He:2007kd,Deng:2008ah,Chatterjee:2008dr,Guo:2008ti}
for related models).  In view of the known relation between $\Delta$
and the Dirac eigenvalues \cite{Yamamoto:2009ey,Kanazawa:2009ks}, a
natural question arises: How is the diquark condensate $\ev{\psi\psi}$
reflected in the Dirac spectrum of QCD at nonzero density?

In this paper, we point out that it is the spectrum of the
\emph{singular values} of the Dirac operator $D(\mu)$, i.e., the
square roots of the eigenvalues of $D(\mu)^\dagger D(\mu)$, that
carries the information on the diquark condensate $\langle \psi \psi
\rangle$ at any nonzero chemical potential in all QCD-like theories we
consider here.  By inserting a diquark source (instead of a quark
mass)\footnote{The diquark source can also be called Majorana mass
  while the quark mass is called Dirac mass.}  we derive a number of
rigorous relationships covering the entire BEC-BCS crossover region,
such as Banks-Casher-type relations, Smilga-Stern-type relations, and
Leutwyler-Smilga-type spectral sum rules for the Dirac singular
values.  This significantly extends previous related work
\cite{Fukushima:2008su} in many directions.  We also construct chiral
random matrix theories and determine the form of the microscopic
spectral correlation functions for singular values at all nonzero
densities.  In most of this paper we will work in the chiral limit.

Note that, although the Dirac singular values and the Dirac
eigenvalues coincide at $\mu=0$, they are essentially different
objects at nonzero $\mu$.%
\footnote{As an aside we mention that at $\mu=0$ the spectrum of the
  so-called Hermitian Wilson-Dirac operator in lattice QCD with Wilson
  fermions is nothing but the singular value spectrum of the
  Wilson-Dirac operator owing to the $\gamma_5$-Hermiticity of the
  operator. In this case, the density of the singular values near the
  origin is proportional to the parity-breaking condensate
  \cite{Bitar:1997ic}.  } First, while the Dirac eigenvalues are
complex, the Dirac singular values are always real and
nonnegative. Second, while the scale of the small Dirac eigenvalues is
governed by the chiral condensate at small $\mu$ and by the BCS gap at
large $\mu$, the scale of the small Dirac singular values is governed
by the diquark condensate at any $\mu$.  From the viewpoint of
symmetries, the diquark source breaks $\U(1)_B$ (baryon number), and
therefore the spectrum of the Dirac singular values characterizes the
$\U(1)_B$ symmetry breaking.  This is in contrast to the spectrum of
the Dirac eigenvalues, which is unrelated to $\U(1)_B$ since neither
the quark mass (the source for the chiral condensate at low density)
nor the quark mass squared (the source for $\Delta^2$ at high density)
breaks $\U(1)_B$.

Our results for the Dirac singular values will make it possible to
measure the magnitude of the diquark condensate with high precision by
lattice QCD simulations at any density.  Together with earlier results
for the Dirac eigenvalues they lead to a more detailed understanding
of superfluidity and the BEC-BCS crossover in QCD-like theories, and
hopefully also of the physics of color superconductivity in
three-color QCD.

The structure of this paper is as follows.  In
section~\ref{sec:micro}, we introduce the microscopic theories
considered in this paper and write down the general expressions for
the partition functions.  In section~\ref{sec:singular}, we review
basic properties of the eigenvalues and singular values of the Dirac
operator and clarify their relation at nonzero chiral chemical
potential.  Special emphasis is given to the zero modes.  In
section~\ref{sec:bc}, we derive Banks-Casher-type relations for
two-color QCD ($\beta=1$), QCD at nonzero isospin chemical potential
($\beta=2$), and QCD with adjoint fermions ($\beta=4$).  In
sections~\ref{sec:eff}, \ref{sec:ss}, and \ref{sec:ls}, we concentrate
on two-color QCD ($\beta=1$) for simplicity and brevity.  The results
of these sections admit straightforward generalization to the other
QCD-like theories with $\beta=2$ and $4$.  In section~\ref{sec:eff} we
introduce three different low-energy effective Lagrangians with
diquark sources applicable at low, intermediate, and high density.
In sections~\ref{sec:ss} and \ref{sec:ls}, we derive Smilga-Stern-type
relations and Leutwyler-Smilga-type sum rules, respectively, using the
effective theories from section~\ref{sec:eff}.  In
section~\ref{sec:ls} we also construct chiral random matrix theories
that describe the spectrum of the Dirac singular values in the
$\epsilon$-regime and determine the form of the microscopic spectral
correlation functions.  Section~\ref{sec:conclusion} contains the
conclusions and an outlook.

In appendix~\ref{app:conv}, we summarize our definitions and
conventions.  Appendix~\ref{app:Z-pfaffian} describes the derivations
of the singular value representations of the partition functions for
the theories with $\beta=1$, $2$, and $4$.  In
appendix~\ref{app:complex}, we comment on the importance of the
positivity of the measure, which could be spoiled by the diquark
sources.  We also discuss QCD inequalities and derive constraints on
the symmetry-breaking pattern for positive definite measure.  In
appendix \ref{app:index}, we present careful derivations of the
anomaly equation and the extension of the index theorem to $\mu\ne0$,
paying special attention to the non-Hermitian nature of the Dirac
operator.  In appendix \ref{app:rho_pert} the derivation of
eq.~\eqref{eq:rho_pert} in the main text is outlined.  In appendix
\ref{app:isospin} we comment on the random matrix theory for QCD at
nonzero isospin chemical potential ($\beta=2$).  Finally, appendix
\ref{app:sumlow0} is devoted to the derivation of
eq.~\eqref{eq:sumlow0} in the main text.

\section{Microscopic theories}
\label{sec:micro}

In this section, we introduce QCD-like theories with the Dyson indices
$\beta=1$, $2$, and $4$, emphasizing the anti-unitary symmetries of
the Dirac operator and the global symmetries of the theories.  Since
the analysis of all three cases is similar, we first examine the
$\beta=1$ case in detail and then discuss $\beta=2$ and $\beta=4$
briefly without redundancy.  We take two-color QCD and QCD with
adjoint fermions as examples for the $\beta=1$ and $\beta=4$ cases,
respectively.  The same arguments are readily applicable to
$\Sp(2N_c)$ gauge theory ($\beta=1$) and $\SO(N_c)$ gauge theory
($\beta=4$).

Unless stated otherwise we always work in Euclidean space and at zero
temperature.  Our definitions and conventions are summarized in
appendix~\ref{app:conv}.

\subsection[Two-color QCD ($\beta=1$)]{\boldmath Two-color QCD
  ($\beta=1$)}
\label{sec:micro1}

The Dirac operator of two-color QCD in the presence of a chemical
potential is given by
\begin{align}
  \label{eq:dirac}
  D(\mu)=\gamma_\nu D_\nu+\mu\gamma_4
  \quad\text{with}\quad D_\nu=\partial_\nu+iA_\nu\,,
\end{align}
where $A_\nu=A_\nu^a\tau_a/2$ is the gauge field and the $\tau_a$ are
the generators of $\SU(2)$ color, i.e., the Pauli matrices.  For
$\mu=0$ the Dirac operator is anti-Hermitian, but for $\mu\ne0$ it no
longer has definite Hermiticity properties since $\mu\gamma_4$ is
Hermitian.  However, for two colors there is an anti-unitary
symmetry that can be expressed in two equivalent ways
\cite{Leutwyler:1992yt,Verbaarschot:1994qf,Halasz:1997fc,Kogut:1999iv},
\begin{subequations}
  \begin{align}
    \label{eq:symm_i}
    [C\tau_2K,iD(\mu)]&=0\quad\text{with}\quad(C\tau_2K)^2=\1\,,\\
    \label{eq:symm}
    [\gamma_5C\tau_2K,D(\mu)]&=0\quad\text{with}\quad(\gamma_5C\tau_2K)^2=\1\,,
  \end{align}
\end{subequations}
where $C=i\gamma_4\gamma_2$ is the charge conjugation matrix (see
appendix~\ref{app:conv}) and $K$ is the operator of complex
conjugation.  Hence the Dyson index is $\beta=1$ in this case and we
can choose a basis in which the Dirac operator is real.  This implies
that $\det D(\mu)$ is real, which also follows from $C\tau_2\gamma_5
D(\mu)C\tau_2 \gamma_5=D(\mu)^*$ as a consequence of \eqref{eq:symm}.

We introduce $N_f$ quark flavors\footnote{We assume $N_f<11$ to ensure
  asymptotic freedom.} described by Dirac spinors $\psi_f$
($f=1,\ldots,N_f$), which we collect in
$\psi=(\psi_1,\ldots,\psi_{N_f})^T$.  Each Dirac spinor can be split
into two Weyl spinors, which we collect in
$\psi_R=(\psi_{1R},\ldots,\psi_{N_fR})^T$ and
$\psi_L=(\psi_{1L},\ldots,\psi_{N_fL})^T$.  The fermionic part of the
Lagrangian, including mass term and diquark sources, is given by
\begin{align}
  \label{eq:Lf}
  \mL_f
  &=\bar\psi [D(\mu)+MP_L+M^\dagger P_R]\psi
  +\frac{1}{2}\psi^TC\tau_2(J_RP_R+J_LP_L)\psi
  +\frac{1}{2}\psi^\dagger C\tau_2(J_R^\dagger P_R+J_L^\dagger P_L)\psi^*\!,
\end{align}
where $P_{R/L}=(\1\pm\gamma_5)/2$ and $M$, $J_R$, and $J_L$ are
$N_f$-dimensional complex matrices in flavor space.  As a consequence
of the Pauli principle, it suffices to take $J_R$ and $J_L$ to be
antisymmetric since their symmetric parts drop out of the combinations
in \eqref{eq:Lf}.  For greater generality we have allowed for two
independent matrices $J_R$ and $J_L$.

We briefly summarize the symmetries of \eqref{eq:Lf} obtained in
\cite{Kogut:1999iv,Kogut:2000ek}, assuming $M=m\1$ and $J_R=-J_L=jI$,
where $m$ and $j$ are real and $I$ is defined in \eqref{eq:I}.  For
$\mu=0$, the Lagrangian in the absence of mass term and diquark source
is symmetric under $\SU(2N_f)$.\footnote{$\U(1)_B$ is contained in
  $\SU(2N_f)$. There is an additional $\U(1)_A$ symmetry which is
  anomalous and therefore not considered here, but see
  section~\ref{sec:eff_high}.} The mass term breaks this symmetry to
$\Sp(2N_f)$.  The diquark source transforms into the mass term under
an $\SU(2N_f)$ rotation, and therefore it brings nothing new at
$\mu=0$.  For $\mu\ne0$ and no sources, the chemical potential breaks
the $\SU(2N_f)$ symmetry explicitly to
$\SU(N_f)_R\times\SU(N_f)_L\times\U(1)_B$.  The mass term breaks this
symmetry to $\SU(N_f)_{L+R}\times\U(1)_B$, while the diquark source
breaks it to $\Sp(N_f)_R\times\Sp(N_f)_L$.  In the presence of both
mass term and diquark source the symmetry is broken to
$\Sp(N_f)_{L+R}$.

We now express $\mL_f$ in the so-called Nambu-Gor'kov formalism.
Defining
\begin{align}
  \label{eq:Psi}
  \Psi\equiv\begin{pmatrix}\psi\\\bar\psi^T\end{pmatrix},  
\end{align}
\eqref{eq:Lf} can be rewritten as 
\begin{align}
  \mL_f=\frac{1}{2}\Psi^TW\Psi
\end{align}
with
\begin{align}
  W =
  \begin{pmatrix}
    C\tau_2(J_RP_R+J_LP_L) & -D(\mu)^T-M^TP_L-M^*P_R\\
    D(\mu)+MP_L+M^\dagger P_R & -C\tau_2(J_R^\dagger P_L+J_L^\dagger P_R)
  \end{pmatrix}.
  \label{eq:W}
\end{align}
Since $J_R$, $J_L$, $C$, and $\tau_2$ are all antisymmetric, $W$ is
antisymmetric as well, and thus
\begin{align}
  \int \dd \Psi\,\exp\left(-\frac{1}{2}\int d^4x\,\Psi^TW\Psi\right)=\Pf (W)\,,
\end{align}
where $\Pf$ denotes the Pfaffian, which for an antisymmetric matrix of
even dimension $N$ is defined as%
\footnote{Later we will use the following basic properties of the
  Pfaffian: $\Pf(X)^2=\det(X)$, $\Pf(X^T)=(-1)^{N/2}\Pf(X)$,
  $\Pf(cX)=c^{N/2}\Pf(X)$ for $c\in\C$, and
  $\Pf(UXU^T)=\det(U)\Pf(X)$.  }
\begin{align}
  \Pf (X) \equiv \frac{1}{2^{N/2}(N/2)!}\sum_{\sigma}
  \sgn(\sigma)X_{\sigma(1)\sigma(2)}\dots X_{\sigma(N-1)\sigma(N)}\,.
\end{align}
The partition function is therefore
\begin{align}
  \label{eq:Z}
  Z=\Big\langle \! \Pf (W)\Big\rangle _\text{YM}\,,
\end{align}
where the subscript YM means that the average is over gauge fields
weighted by the pure gauge (Yang-Mills) action.  Since in this paper
we will almost always work in the chiral limit we set $M$ to zero in
\eqref{eq:W}.  As shown in appendix~\ref{app:Z1}, we then obtain from
\eqref{eq:Z}
\begin{align}
  \label{eq:ZJ}
  Z(J_L,J_R)&=\bigg\langle
    \big[\Pf(J_R) \Pf(J_L^\dagger)\big]^{n_R}
    \big[\Pf(J_R^\dagger) \Pf(J_L)\big]^{n_L}
    \notag\\
    &\qquad\times
    \sqrt{{\det}'(D(\mu)^\dagger D(\mu) +J_R^\dagger J_RP_R 
    + J_L^\dagger J_LP_L )}\,\bigg\rangle_\text{YM}\,,
\end{align}
where $n_R$ ($n_L$) denotes the number of right- (left-) handed zero
modes of $D^\dagger D$ and the prime on the determinant means that the
zero modes of $D^\dagger D$ are omitted.\footnote{The numbers $n_R$
  and $n_L$ are also equal to the number of right- and left-handed
  zero modes of $D$, see the discussion in
  section~\ref{sec:singular}.}  This expression is invariant under
$\SU(N_f)_R\times \SU(N_f)_L$, as it should be, since $\Pf(J_R) \to
\Pf(U_R^TJ_RU_R)=\det(U_R)\Pf(J_R)=\Pf(J_R)$ and likewise for
$\Pf(J_L)$.

We can also add a P- and CP-violating term $i\theta F\tilde F/32\pi^2$
to the Lagrangian, where $\tilde F_{\alpha \beta}^a =
\frac{1}{2}\epsilon_{\alpha \beta \gamma \delta}F_{\gamma \delta}^a$.
The topological charge $\nu$ of a gauge field configuration is given
by
\begin{align}
  \label{eq:nu_definition}
  \nu \equiv \frac1{32\pi^2}\int d^4x\,
  F_{\alpha\beta}^a \tilde F_{\alpha\beta}^a
\end{align}
and related to the number of zero modes of $D$ by
$\nu=n_R-n_L$.\footnote{Note that for $\mu\ne0$ the equality
  $\nu=n_R-n_L$ is violated on a gauge field set of measure zero, see
  appendix~\ref{app:index}.  Here (and also in the next two
  subsections) we exclude this possibility.}  The $\theta$-term
corresponds, for a fixed gauge field, to a term $i\nu\theta$ in the
action.  The partition function is then given by
\begin{align}
  \label{eq:Z_theta}
  Z(\theta)=\sum_{\nu=-\infty}^{\infty} \ee^{i\nu\theta}Z_\nu\,,
\end{align}
where $Z_\nu$ is computed by integrating only over gauge fields with
fixed topology $\nu$.  

For zero diquark sources and nonzero quark masses it is well known
that an axial rotation of the quark fields changes the fermionic
measure in such a way that the $\theta$-term can be traded for a
redefinition of the quark mass matrix, $M\to M\ee^{-i\theta/N_f}$
\cite{Leutwyler:1992yt}.\footnote{In \cite{Leutwyler:1992yt} $M\to
  M\ee^{i\theta/N_f}$ is used, which is due to different conventions,
  see also footnotes \ref{ft:conv_diff} and \ref{ft:conv_diff2}.} For
nonzero diquark sources and in the chiral limit, it follows from
\eqref{eq:ZJ} that the $\theta$-term can be traded for a redefinition
of the diquark sources, $J_R\to J_R\ee^{i\theta_R/N_f}$and $J_L\to
J_L\ee^{i\theta_L/N_f}$ with $\theta=(\theta_R-\theta_L)/2$.

Note that for some choices of the diquark sources and/or for a nonzero
value of $\theta$ the fermionic measure in the partition function is
not positive definite, which causes some subtleties that are discussed
in detail in appendix~\ref{app:complex}.

\subsection[QCD with isospin chemical potential ($\beta=2$)]{\boldmath
  QCD with isospin chemical potential ($\beta=2$)}
\label{sec:isospin}

We now consider $N_f=2$ QCD at nonzero isospin chemical potential
$\mu_I = 2\mu$ for an arbitrary number of colors $N_c\ge3$.  The Dirac
operator is given as in \eqref{eq:dirac}, except that the $\tau_a$ are
now the generators of $\SU(N_c)$ in the fundamental representation.
The chemical potential $\mu$ is assigned to the $u$-quark, while
$-\mu$ is assigned to the $d$-quark.  For $\mu\ne0$, the Dirac
operator is no longer anti-Hermitian, but because of
$D(\mu)^\dagger=-D(-\mu)$ the fermion determinant is real and
nonnegative: $\det D(\mu)D(-\mu)=\det
D(\mu)D(\mu)^\dagger\ge0$.\footnote{Note that $\det(-D^\dagger)=\det
  D^\dagger$ if $D^\dagger$ has no zero modes.  If it has zero modes
  we simply have $0=0$.}  Since there is no anti-unitary symmetry we
have $\beta=2$.

The fermionic part of the Lagrangian is given by
\begin{align}
\label{eq:Lf2}
  \mL&=\bar u(\gamma_\nu D_\nu+\mu\gamma_4)u
  +\bar d(\gamma_\nu D_\nu-\mu\gamma_4)d + (m u^{\dag}_R u_L + m
  d^{\dag}_R d_L + \text{h.c.}) 
  \notag \\
  & \quad + (\lambda^* u_L^\dagger d_R + \rho d_L^\dagger u_R  +
  \text{h.c.}) \notag\\
  &= \begin{pmatrix}\bar u & \;\bar d\end{pmatrix}
  \begin{pmatrix}D(\mu) + m P_L + m^* P_R & \lambda^* P_R+\rho^*P_L 
  \\
  \rho P_R + \lambda P_L & D(-\mu) + m P_L + m^* P_R \end{pmatrix}
  \begin{pmatrix}u\\d\end{pmatrix},
\end{align}
where $m$ is the degenerate mass of the two quarks, and $\rho$ and
$\lambda$ are ``pionic'' sources.  From this we obtain (see
appendix~\ref{app:iso})
\begin{align}
  \label{eq:app_main_result_beta=2}
  Z(\rho,\lambda) = \left\langle
    (-\rho\lambda^*)^{n_R}(-\rho^*\lambda)^{n_L}
  {\det}'\big(D^\dagger D + \rho \rho^*P_R + \lambda \lambda^*P_L \big)
  \right\rangle_\text{YM},
\end{align}
where the prime again indicates that the zero modes of $D^\dagger D$
are omitted.  A nonzero $\theta$-term can be introduced in
\eqref{eq:app_main_result_beta=2} by redefining the pionic sources as
$\rho\to\rho \ee^{i\theta_R/2}$ and $\lambda\to\lambda
\ee^{i\theta_L/2}$, where again $\theta=(\theta_R-\theta_L)/2$.

The symmetries of \eqref{eq:Lf2} with $\rho=-\lambda\in\R$ are as
follows.  The Lagrangian at $\mu=0$ in the absence of sources is
symmetric under $\SU(2)_L \times \SU(2)_R \times \U(1)_B$, which is
broken to $\SU(2) \times \U(1)_B$ by the mass term or the pionic
sources.\footnote{For the mass term we end up with $\SU(2)_{L+R}$,
  while for the pionic sources we end up with a different $\SU(2)$
  subgroup given by the condition $U_R^\dagger t_2U_L=t_2$, where
  $t_2$ is the second generator of $\SU(2)$.} There is also a
$\U(1)_A$ symmetry that is broken by the anomaly.  Without sources the
chemical potential breaks the $\SU(2)_L \times \SU(2)_R \times
\U(1)_B$ symmetry to $\U(1)_L \times \U(1)_R \times \U(1)_B$, where
$\U(1)_{L,R}$ is generated by $t_3\in\su(2)_{L/R}$.  The mass term
breaks this symmetry to $\U(1)_{L+R} \times \U(1)_B$, while the pionic
sources break it to $\U(1)_{L-R}\times\U(1)_B$.  With both mass term
and pionic sources the remaining symmetry is $\U(1)_B$.

\subsection[QCD with adjoint fermions ($\beta=4$)]{\boldmath QCD with
  adjoint fermions ($\beta=4$)} 
\label{sec:beta4}

Finally we consider QCD with fermions in the adjoint
representation.\footnote{The results obtained for the adjoint
  representation easily generalize to other real representations.}
The Dirac operator in the presence of a chemical potential is given by
\begin{align}
  D(\mu)=\gamma_\nu D_\nu + \mu\gamma_4
  \quad\text{with}\quad D_\nu=\partial_\nu \delta_{ab} + (f^c)_{ab}A_\nu^c\,,
\end{align}
where $(f^c)_{ab}=f_{abc}$ denotes the generators of the adjoint
representation (or structure constants).  For $\mu=0$ the Dirac
operator is anti-Hermitian, but for $\mu\ne0$ it again loses its
Hermiticity properties.  There is an anti-unitary symmetry that can be
expressed in two equivalent ways
\cite{Leutwyler:1992yt,Verbaarschot:1994qf,Halasz:1997fc,Kogut:1999iv},
\begin{subequations}
  \begin{align}
    \label{eq:symm_adj_i}
    [CK,iD(\mu)]&=0\quad\text{with}\quad(CK)^2=-\1\,,\\
    \label{eq:symm_adj}
    [\gamma_5CK,D(\mu)]&=0\quad\text{with}\quad(\gamma_5CK)^2=-\1\,.
  \end{align}
\end{subequations}
Hence $\beta=4$ and we can choose a basis in which the Dirac operator
is quaternion real.  The symmetry \eqref{eq:symm_adj} implies
$C\gamma_5 D(\mu)C \gamma_5=D(\mu)^*$, from which it follows that
$\det D(\mu)$ is real (and actually nonnegative because all
eigenvalues occur in quadruplets, see section~\ref{sec:ev}).  Because
of $(CK)^2=-\1$ one can show that for $\mu=0$ (but not for $\mu\ne0$)
the eigenvalues of the Dirac operator are twofold degenerate with
linearly independent eigenstates $\psi$ and $C\psi^*$ (Kramers
degeneracy).

Because the adjoint representation is real, it may be convenient to
describe the fermions in the partition function in terms of Majorana
fields.  However, Majorana fermions cannot be defined in Euclidean
space, and therefore we first write the Lagrangian in Minkowski space
and then analytically continue to Euclidean space by a Wick rotation
\cite{Leutwyler:1992yt,Halasz:1995qb}.  In Minkowski space, the
Lagrangian for $N_f=1$ Dirac fermions with diquark sources reads
\begin{align}
  \mL_M^{(N_f=1)}&=\bar\psi(i\gamma^\nu D_\nu- m -\mu\gamma^0)\psi
  +\frac12\big[\psi^TC(j_RP_R+j_LP_L)\psi+\text{h.c.}\big]\notag\\
  &=\frac12
  \begin{pmatrix}
    \bar\psi^T \\ \psi
  \end{pmatrix}^T
  \begin{pmatrix}
    -C(j_R^*P_L+j_L^*P_R) & i\gamma^\nu D_\nu-m-\mu\gamma^0\\
    C(-i\gamma^\nu D_\nu+m-\mu\gamma^0)C & C(j_RP_R+j_LP_L)
  \end{pmatrix}
  \begin{pmatrix}
    \bar\psi^T \\ \psi
  \end{pmatrix}
  \label{eq:Lf4}
  \notag \\
  &=\frac12
  \begin{pmatrix}
    \psi_c \\ \psi
  \end{pmatrix}^T
  \begin{pmatrix}
    C(j_R^*P_L+j_L^*P_R) & -C(i\gamma^\nu D_\nu-m-\mu\gamma^0)\\
    C(-i\gamma^\nu D_\nu+m-\mu\gamma^0) & C(j_RP_R+j_LP_L)
  \end{pmatrix}
  \begin{pmatrix}
    \psi_c \\ \psi
  \end{pmatrix},
\end{align}
where in the last line we have defined $\psi_c=C\bar\psi^T$.  The
partition function is thus given by
\begin{align}
  Z_M^{(N_f=1)}&=\ev{\Pf\begin{pmatrix}
    C(j_R^*P_L+j_L^*P_R) & -C(i\gamma^\nu D_\nu-m-\mu\gamma^0)\\
    C(-i\gamma^\nu D_\nu+m-\mu\gamma^0) & C(j_RP_R+j_LP_L)
  \end{pmatrix}}_\text{YM}.
\end{align}
This result generalizes trivially to general $N_f$.  The partition
function is then simply the expectation value of a product of $N_f$
Pfaffians, possibly with different masses.

For simplicity we now take the chiral limit.  After Wick rotation, the
partition function for $N_f$ flavors in Euclidean space becomes
\begin{align}
  Z_E=\ev{\Pf\begin{pmatrix}
    C(J_R^*P_L+J_L^*P_R) & CD(\mu)\\
    -CD(\mu)^\dagger & C(J_RP_R+J_LP_L)
  \end{pmatrix}}_\text{YM},
  \label{eq:Z4}
\end{align}
where $J_L$ and $J_R$ are now symmetric matrices of dimension $N_f$.
For $\beta=4$, $N_f$ does not have to be even.  In
appendix~\ref{app:real} we show that this leads to
\begin{align}
  \label{eq:app_main_result_beta=4}
  Z(J_L, J_R) = \left\langle
    \det(-J_R J_L^{\dag})^{n_R/2}\det(-J_R^{\dag}J_L)^{n_L/2}
  {\det}''(D^\dagger D + J_R^{\dag} J_R P_R + J_L^{\dag}J_L P_L)
  \right\rangle_\text{YM},
\end{align}
where ${\det}''$ indicates that the zero modes of $D^{\dag}D$ are
omitted and that each degenerate eigenvalue of $D^\dagger D$ is
counted only once.  It follows from the derivation in
appendix~\ref{app:real} that a nonzero $\theta$-term can be introduced
in \eqref{eq:app_main_result_beta=4} by redefining $J_R\to
J_R\ee^{i\theta_R/2N_fN_c}$ and $J_L\to J_L\ee^{i\theta_L/2N_fN_c}$,
where again $\theta=(\theta_R-\theta_L)/2$.

The symmetries of \eqref{eq:Lf4} extended to general $N_f$ with
degenerate masses and $J_R=-J_L=j\1$ with real $j$ are as follows
\cite{Kogut:2000ek}.  The Lagrangian at $\mu=0$ in the absence of
sources is symmetric under $\SU(2N_f)$, which is broken to $\SO(2N_f)$
by the quark mass or the diquark source.\footnote{As in the case of
  $\beta=1$, the diquark source transforms into the mass term under an
  $\SU(2N_f)$ rotation.} There is also the usual anomalous $\U(1)_A$
symmetry.  Without sources the chemical potential breaks the
$\SU(2N_f)$ symmetry to $\SU(N_f)_L \times \SU(N_f)_R \times \U(1)_B$.
The mass term breaks this symmetry to $\SU(N_f)_{L+R} \times \U(1)_B$,
while the diquark source breaks it to $\SO(N_f)_L \times \SO(N_f)_R$.
With both mass term and diquark source the remaining symmetry is
$\SO(N_f)_{L+R}$.

\section{Eigenvalues and singular values of the Dirac operator}
\label{sec:singular}

In this section we discuss the eigenvalues and singular values of the
Dirac operator and related quantities.  Two preliminary remarks are in
order.

First, some parts of the discussion rely on the index theorem.  In
appendix~\ref{app:index} we show that for a non-Hermitian Dirac
operator such as $D(\mu)$ the index theorem takes the form
\begin{align}
  \frac1{32\pi^2}\int d^4x\,F\tilde F
  =\frac12\big[\ind D(\mu)+\ind D(\mu)^\dagger\big]
\end{align}
and that 
\begin{align}
  \ind D(\mu)=\ind D(\mu)^\dagger\quad\text{almost surely,}
\end{align}
where $\ind D(\mu)=\dim\ker D_R-\dim \ker D_L$, see
\eqref{eq:mu_D_def} for the notation.  The meaning of ``almost
surely'' is that to have $\ind D(\mu)\ne\ind D(\mu)^\dagger$ the gauge
field needs to be fine-tuned, which corresponds to a gauge field set
of measure zero.  In this section we ignore this set of measure zero
and use the index theorem in its usual form.

Second, we implicitly assume a regularization (such as lattice QCD)
that allows us to count the number of eigenstates.  Some of our
arguments rely on the index theorem, which can be (and usually is)
violated by the regulator.  Therefore our results apply only after the
regulator has been removed.  We assume that the procedure of removing
the regulator does not invalidate the results that rely on the index
theorem.

\subsection{Eigenvalues}
\label{sec:ev}

The discussion of this section extends that of \cite{Halasz:1997fc}.
The eigenvalue equation for the Dirac operator is
\begin{align}
  D(\mu)\psi_n=\lambda_n\psi_n\,.
\end{align}
For simplicity we will omit the argument $\mu$ if no confusion is
likely to arise.

Let us first consider $\mu=0$.\footnote{The arguments for $\mu=0$
  apply to all values of the Dyson index $\beta$, except that for
  $\beta=4$ we have the additional Kramers degeneracy already
  mentioned in section~\ref{sec:beta4}.}  In that case the eigenvalues
$\lambda_n$ are purely imaginary, and because of $\{D,\gamma_5\}=0$
the nonzero eigenvalues come in pairs $\pm\lambda_n$.  If the
eigenstate corresponding to $\lambda_n$ is $\psi_n$, then the
eigenstate corresponding to $-\lambda_n$ is $\gamma_5\psi_n$.  There
can also be eigenvalues equal to zero, and the corresponding
eigenstates can be arranged to be eigenstates of $\gamma_5$.  In
general there are $n_R$ ($n_L$) right-handed (left-handed) zero modes
with eigenvalue $+1$ ($-1$) of $\gamma_5$.  The difference $n_R-n_L$
is equal to the topological charge $\nu$ of the underlying gauge field
configuration via the index theorem and stable under perturbations of
the gauge field.  All other possible zero modes are accidental in the
sense that they require a fine-tuning of the gauge field.  This
implies that generically there are only zero modes of one chirality.
In the remainder of this section we assume that there are no
accidental zero modes.

We now discuss what happens to the eigenvalues as $\mu$ is turned on
(for a fixed gauge field and in a finite volume) and first consider
$\beta=1$ (e.g., two-color QCD).  Since $D(\mu\ne0)$ is no longer
anti-Hermitian one would expect the eigenvalues to move into the
complex plane for any $\mu$.  However, using the symmetry
\eqref{eq:symm} one can show that the nonzero eigenvalues come either
in quadruplets $\lambda,-\lambda,\lambda^*,-\lambda^*$ (with
eigenstates $\psi,\gamma_5\psi,C\tau_2\gamma_5\psi^*,C\tau_2\psi^*$)
if $\lambda$ is complex or in pairs $\pm\lambda$ (with eigenstates
$\psi,\gamma_5\psi$) if $\lambda$ is purely real or purely imaginary.
Since at $\mu=0$ the nonzero eigenvalues are generically nondegenerate
(because of level repulsion due to interactions), a quadruplet cannot
be formed for infinitesimally small $\mu$.\footnote{The assumption of
  a finite volume is essential here.} What happens (see
figure~\ref{fig:flow}) is that as a function of $\mu$, the eigenvalues
move along the imaginary axis until two eigenvalues (and their
partners) become degenerate, i.e., the effect of $\mu$ overcomes the
level repulsion.  As $\mu$ is increased further, these four
eigenvalues move into the complex plane and form a quadruplet.
Another possibility is that a pair of originally imaginary eigenvalues
hits zero and then becomes a pair of real eigenvalues.  Finally, two
real eigenvalues (and their partners) can also merge and then become a
quadruplet.\footnote{All of these three possibilities have been
  verified in lattice simulations and random matrix studies.  We thank
  Jacques Bloch for performing the lattice simulations.}  The original
topological eigenvalues equal to zero (but not the accidental ones)
are stable under the perturbation by $\mu$.  The corresponding zero
modes change smoothly with $\mu$ and remain eigenstates of $\gamma_5$
\cite{Abrikosov:1980nx,Abrikosov:1981qb,AragaodeCarvalho:1980de,Schafer:1998up}.

\begin{figure}
  \centering
  \includegraphics{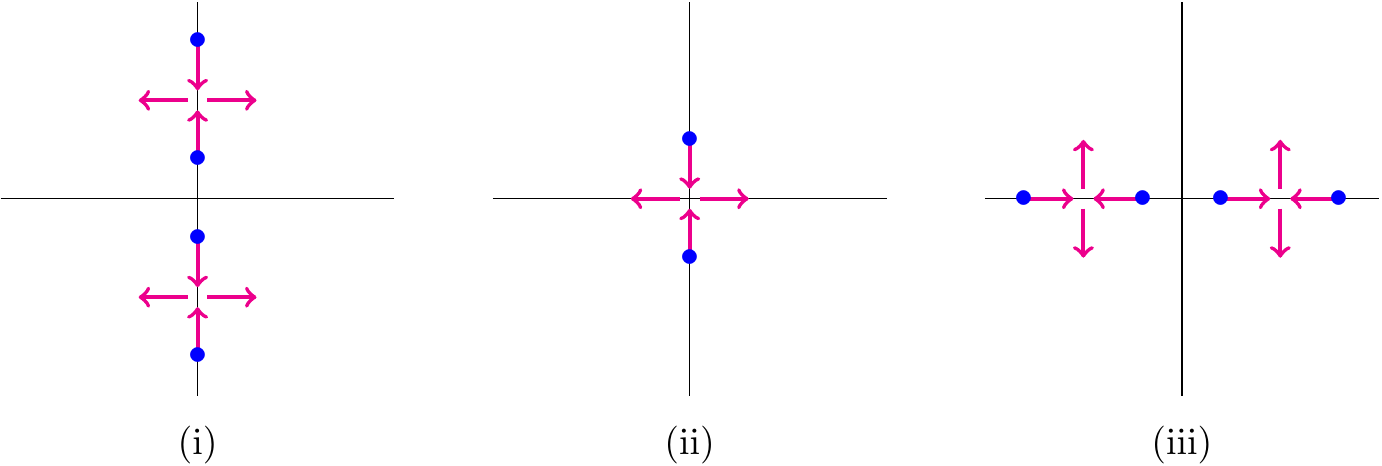}
  \caption{Flow of the Dirac eigenvalues for a fixed gauge field as a
    function of $\mu$ for $\beta=1$.  (i) Two imaginary eigenvalues
    (and their partners) merge and move into the complex plane to form
    a quadruplet. (ii) A pair of imaginary eigenvalues merges at zero
    and becomes a pair of real eigenvalues.  (iii) Two real
    eigenvalues (and their partners) merge and move into the complex
    plane to form a quadruplet.}
  \label{fig:flow}
\end{figure}

For $\beta=2$ there is no anti-unitary symmetry, and the eigenvalues
(which still come in pairs $\pm\lambda$ with eigenstates
$\psi,\gamma_5\psi$) move into the complex plane immediately as $\mu$
is switched on.  Exactly real or imaginary eigenvalues only occur
accidentally, i.e., if the gauge field is fine-tuned.  The topological
zero modes behave as in the case of $\beta=1$.

For $\beta=4$ it follows from the symmetry \eqref{eq:symm_adj} that
the nonzero eigenvalues come in quadruplets
$\lambda,-\lambda,\lambda^*,-\lambda^*$ (with eigenstates
$\psi,\gamma_5\psi,C\gamma_5\psi^*,C\psi^*$) if $\lambda$ is complex
and that the Kramers degeneracy is removed once $\mu$ is switched on.
Hence we do not need the mechanism of figure~\ref{fig:flow}(i) and the
eigenvalues move into the complex plane immediately.  As in the case
of $\beta=2$, exactly real or imaginary eigenvalues only occur
accidentally.  The topological zero modes are twofold degenerate for
both $\mu=0$ and $\mu\ne0$.  They change smoothly with $\mu$ and are
eigenstates of $\gamma_5$ with the same chirality.

\subsection{Singular values}
\label{sc:singu}

Let us now consider the operator $D^\dagger D$.  Since this operator
is Hermitian with nonnegative eigenvalues, we write its eigenvalue
equation in the form
\begin{align}
  \label{eq:sing}
  D^\dagger D\phi_n=\xi_n^2\phi_n\,.
\end{align}
The $\xi_n$ are real and nonnegative.  They are called the singular
values of $D$, the name coming from the singular value decomposition
of a non-Hermitian matrix.%
\footnote{In this paper we assume that the extension of the singular
  value decomposition to non-Hermitian operators is straightforward
  and skip the mathematical foundations.}  The operators $D^\dagger D$
and $DD^\dagger$ share all nonzero eigenvalues, since \eqref{eq:sing}
implies
\begin{align}
  DD^\dagger(D\phi_n)=\xi_n^2(D\phi_n)\,,
\end{align}
and similarly the other way around.  

At $\mu=0$ the $\lambda_n$ and $\xi_n$ are trivially related by
$\xi_n=|\lambda_n|$, and therefore the nonzero singular values are
twofold degenerate (for $\beta=1$ and 2) or fourfold degenerate (for
$\beta=4$).  At $\mu\ne0$ there is no simple relation between the
$\lambda_n$ and $\xi_n$, i.e., knowing only the $\lambda_n$ we cannot
compute the $\xi_n$, and vice versa.  As soon as $\mu$ is turned on
the twofold degeneracy (for $\beta=1$ and 2) is removed, and the
fourfold degeneracy (for $\beta=4$) is reduced to a twofold
degeneracy, see the argument after \eqref{eq:det'_beta=4}.

The operator $D^\dagger D$ commutes with $\gamma_5$, and therefore the
states $\phi_n$ have definite chirality, $\gamma_5\phi_n=\pm\phi_n$
(or can be so arranged if the singular values are degenerate).  Now
assume $\xi_n>0$ and define $\tilde\phi_n=\xi_n^{-1}D\phi_n$, for
which
\begin{align}
  DD^\dagger \tilde\phi_n&=\xi_n^{-1}DD^\dagger D\phi_n
  =\xi_nD\phi_n=\xi_n^2\tilde\phi_n\,,\\
  \gamma_5\tilde\phi_n&=\xi_n^{-1}\gamma_5D\phi_n
  =-\xi_n^{-1}D\gamma_5\phi_n=\mp\tilde\phi_n\,,
\end{align}
i.e., $\tilde\phi_n$ is an eigenstate of $DD^\dagger$ with the same
eigenvalue but chirality opposite to that of $\phi_n$.  Therefore the
number of right-handed (left-handed) nonzero modes of $D^\dagger D$
equals the number of left-handed (right-handed) nonzero modes of
$DD^\dagger$.  Since these operators coincide at $\mu=0$, the number
of right-handed and left-handed nonzero modes of $D^\dagger D$ is
equal at $\mu=0$.  As $\mu$ is turned on from zero, a right-handed (or
left-handed) nonzero mode of $D^\dagger D$ changes its form smoothly
but stays right-handed (or left-handed) since the eigenvalue of
$\gamma_5$ is discrete.  Therefore the numbers of right-handed and
left-handed nonzero modes of $D^\dagger D$ and $DD^\dagger$ are all
equal (assuming that there are no accidental zero modes).

If $D$ has an eigenvalue equal to zero, $D\psi=0$ trivially implies
$D^\dagger D\psi=0$, i.e., there is a corresponding singular value of
$D$ equal to zero, and the zero mode of $D^\dagger D$ is the same as
that of $D$.  If $\psi$ is not a zero mode of $D$, it cannot be a zero
mode of $D^\dagger D$ either since $\langle\psi|D^\dagger
D|\psi\rangle\ne0$.  As discussed in section~\ref{sec:ev}, all zero
modes of $D$ and therefore of $D^\dagger D$ generically have the same
chirality.  Using $D(\mu)^\dagger=-D(-\mu)$ and our earlier
observation that the number of zero modes is stable as a function of
$\mu$, we conclude that (i) the operator $DD^\dagger$ has the same
number of zero modes as $D^\dagger D$, (ii) the zero modes of
$DD^\dagger$ are equal to those of $D(-\mu)$, and (iii) the chirality
of the zero modes of $D^\dagger D$ and $DD^\dagger$ is the
same.\footnote{Conclusion (iii) is not necessarily valid if there are
  accidental zero modes, while (i) and (ii) are always valid, see
  appendix~\ref{app:index}.}

\begin{figure}
  \centering
  \includegraphics{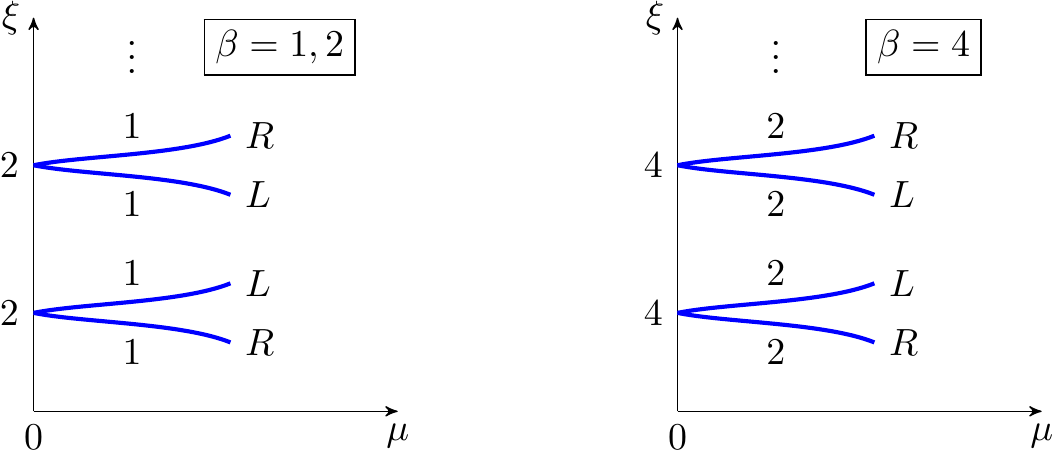}
  \caption{ Schematic flow of the singular values of $D$ for a fixed
    gauge field as a function of $\mu$ for the three symmetry classes.
    The numbers stand for the degeneracy of each singular value.}
  \label{fig:sv_flow}
\end{figure}

In figure \ref{fig:sv_flow} we schematically illustrate the singular
value flow in QCD-like theories with $\beta=1$, $2$, and $4$.  All
eigenstates of $D^\dagger D$ have definite chirality, as indicated by
$R$ and $L$ in the figure.  For $\beta=4$ the zero modes of $D$, and
therefore of $D^\dagger D$, are twofold degenerate for both $\mu=0$
and $\mu\ne0$.

\subsection{Dirac operator with chiral chemical potential}
\label{sec:D5}

In this subsection we derive a relation between the singular values
and the eigenvalues of the Dirac operator with chiral (or axial)
chemical potential defined by
\begin{align}
  \label{eq:D5}
  D_5(\mu)=\gamma_\nu D_\nu+\mu\gamma_4\gamma_5\,.
\end{align}
This operator was introduced in the context of the chiral magnetic
effect \cite{Fukushima:2008xe}.

Since $D_5(\mu)$ is anti-Hermitian its eigenvalues are purely
imaginary, and since it anticommutes with $\gamma_5$ its nonzero
eigenvalues come in pairs with opposite sign.  Now, for any
right-handed (left-handed) spinor $\phi_{R}$ ($\phi_{L}$) we have
\begin{align}
  D_5(\mu)^2\phi_R&=(\gamma_\nu D_\nu+\mu\gamma_4\gamma_5)
  (\gamma_\nu D_\nu+\mu\gamma_4\gamma_5)\phi_R\notag\\
  &=(\gamma_\nu D_\nu+\mu\gamma_4\gamma_5)
  (\gamma_\nu D_\nu+\mu\gamma_4)\phi_R\notag\\
  &=(\gamma_\nu D_\nu-\mu\gamma_4)
  (\gamma_\nu D_\nu+\mu\gamma_4)\phi_R\notag\\
  &=-D(\mu)^\dagger D(\mu)\phi_R
  \intertext{and similarly}
  D_5(\mu)^2\phi_L&=-D(\mu)D(\mu)^\dagger\phi_L\,,\\
  D_5(-\mu)^2\phi_R&=-D(\mu)D(\mu)^\dagger\phi_R\,,\\
  D_5(-\mu)^2\phi_L&=-D(\mu)^\dagger D(\mu)\phi_L\,,
\end{align}
which can also be expressed in terms of the decompositions
\begin{align}
  D_5(\mu)^2&=-D^\dagger DP_R-DD^\dagger P_L\,,\\
  D_5(-\mu)^2&=-D^\dagger DP_L-DD^\dagger P_R\,,\\
  \intertext{or equivalently}
  D^\dagger D&=-D_5(\mu)^2P_R-D_5(-\mu)^2P_L\,,\\
  DD^\dagger&=-D_5(\mu)^2P_L-D_5(-\mu)^2P_R\,.
\end{align}
Let us denote the singular values corresponding to the right-handed
(left-handed) nonzero modes of $D^\dagger D$ by $\xi_{Rn}$
($\xi_{Ln}$).  From the above arguments we conclude that the nonzero
eigenvalues of $D_5(\mu)$ and $D_5(-\mu)$ are given by the sets $\{\pm
i\xi_{Rn}\}$ and $\{\pm i\xi_{Ln}\}$, respectively.

Disregarding accidental zero modes, the zero modes of $D^\dagger D$
and $DD^\dagger$ have only one chirality.  If they are all
right-handed, the zero modes of $D_5(\mu)$ ($D_5(-\mu)$) are
right-handed and given by those of $D^\dagger D$ ($DD^\dagger$).  If
they are all left-handed, the zero modes of $D_5(\mu)$ ($D_5(-\mu)$)
are left-handed and given by those of $DD^\dagger$ ($D^\dagger D$).

Since $D_5(\mu)^{\dag}=-D_5(\mu)$ and $\{D_5(\mu),\gamma_5\}=0$, the
Banks-Casher relation for the chiral condensate can be extended to
nonzero chiral chemical potential without difficulty,
\begin{align}
  \label{eq:bc}
  \langle \bar \psi \psi \rangle &= \pi \rho_{5}(0)\,, 
\end{align}
where $\rho_{5}(\lambda)$ is the spectral density of $D_5(\mu)$ with 
fermionic weight ${\det}^{N_f}D_5(\mu)$.

\section{Banks-Casher-type relations}
\label{sec:bc}

In this section we derive Banks-Casher-type relations for the three
different theories with $\beta=1$, $2$, and $4$.  In each case, a
particular condensate is related to the density of the singular values
at the origin.  For our derivations to be correct it is important that
in the computation of the partition function (including source term
for the desired condensate) the fermionic measure is positive
definite.  If this requirement is not met a probabilistic
interpretation is not possible and a number of subtleties arise that
are discussed in some detail in appendix~\ref{app:complex}.  In the
present section we restrict ourselves to cases in which the measure is
positive definite.

\subsection[Two-color QCD ($\beta=1$)]{\boldmath Two-color QCD ($\beta=1$)}
\label{sec:bc2c}

For simplicity we take $J_R=j_RI$ and $J_L=j_LI$ with $I$ given in
\eqref{eq:I} and numbers $j_R$, $j_L$ that for the time being can be
complex.  It follows from \eqref{eq:ZJ} and the discussion after
\eqref{eq:Z_theta} that the measure is positive definite only if the
combination $\ee^{i\theta}(-j_Rj_L^*)^{N_f/2}$ is real and positive.
If we assume real sources and $\theta=0$, this condition is always
satisfied for $j_R=-j_L$ (corresponding to the scalar diquark
condensate).  If $N_f/2$ is odd, it is violated for $j_R=j_L$
(corresponding to the pseudoscalar diquark condensate).  In the
following we therefore set $j_R=-j_L=j$ with $j$ real,\footnote{We
  will comment on the case $j_R=j_L$ below.} which implies
$J_L^\dagger J_LP_L+J_R^\dagger J_RP_R=j^2\1_{N_f}$.  From
\eqref{eq:ZJ} we then obtain
\begin{align}
  \label{eq:Zj}
  Z(j)=\ev{{\det}^{N_f/2}(D^\dagger D+j^2)}_\text{YM}
  =\Big\langle\prod_n(\xi_n^2+j^2)^{N_f/2}\Big\rangle_\text{YM}\,.
\end{align}
Introducing the notation
\begin{align}
  \ev{O}_j=\frac1{Z(j)}\ev{O\,{\det}^{N_f/2}(D^\dagger D+j^2)}_\text{YM}
\end{align}
for expectation values in the presence of a diquark source, we define
the density of the nonzero\footnote{See appendix~\ref{app:sign} for a
  discussion of the zero-mode contribution.} singular values of the
Dirac operator by
\begin{align}
  \label{eq:rsv}
  \rsv(\xi)=\lim_{V_4\to\infty}\frac1{V_4}
  \Big\langle\sum_n\delta(\xi-\xi_n)\Big\rangle_{j=0}
  \quad\text{for}\quad\xi>0\,,
\end{align}
where $V_4$ is the space-time volume.  The scalar diquark condensate
can then be expressed in terms of this density at the origin,
\begin{align}
  \ev{\psi^TC\gamma_5\tau_2I\psi}&=\lim_{j\to0^+}\lim_{V_4\to\infty}
  \frac1{V_4}\frac\partial{\partial j}\ln Z(j)\notag\\
  &=\lim_{j\to0^+}\lim_{V_4\to\infty}\frac1{V_4}\frac{N_f}2
  \Big\langle\sum_n\frac{2j}{\xi_n^2+j^2}\Big\rangle_j\notag\\
  &=\frac{N_f}2\int_0^\infty d\xi\,\rsv(\xi)\lim_{j\to0^+}
  \frac{2j}{\xi^2+j^2} \notag\\
  &=\frac{N_f}2\pi\rsv(0)\,,
  \label{eq:bc2c}
\end{align}
which is similar to the Banks-Casher relation for the chiral
condensate in terms of the Dirac eigenvalue density at zero.  Our
analysis can be extended to nonzero quark masses.  Assuming for
simplicity $M=m\1$ with real $m$, we again obtain \eqref{eq:Zj}, but
with $D$ and $D^\dagger$ replaced by $D+m$ and $D^\dagger+m$.  This
means that \eqref{eq:bc2c} continues to hold, except that the singular
values are now those of $D+m$.

Some comments are in order.  First, \eqref{eq:bc2c} holds at $\mu=0$
and $\mu\ne0$, whereas the original Banks-Casher relation only holds
at $\mu=0$.  Second, in the derivation of \eqref{eq:bc2c} we have
tacitly dropped the contribution of the singular values equal to zero.
As discussed in appendix~\ref{app:complex}, this is only justified if
the measure is positive definite.  Third, the integral in the third
line of \eqref{eq:bc2c} needs a proper UV regularization.  This was
discussed carefully for the original Banks-Casher relation in
\cite{Leutwyler:1992yt} and works in exactly the same way here.  The
point is that the UV-divergent part disappears in the limit $j\to0^+$.
Fourth, we observe that setting $j_R=j_L=j$ we would have obtained
$\det W=\det(D^\dagger D+j^2)$ just as in \eqref{eq:Zj}, which
formally would have led to \eqref{eq:bc2c} but with the pseudoscalar
diquark condensate $\ev{\psi^TC\tau_2I\psi}$ on the l.h.s.\ instead.
However, for odd $N_f/2$ the measure is not positive definite for $j_R
= j_L$, which invalidates the Banks-Casher relation for the
pseudoscalar condensate, see appendix~\ref{app:sign}.  For even $N_f/2$
the measure is positive for both types of sources, and what condensate
is given by $\rsv(0)$ depends on the choice of sources we
add.\footnote{\label{fn:ising}More precisely, the magnitude of the
  condensate is given by $\rsv(0)$ and its orientation by the external
  sources.  A simple analog is the Ising model, where the direction of
  the spontaneous magnetization at zero temperature depends on the
  direction of the (infinitesimal) external magnetic field.  See also
  appendix~\ref{app:complex}.} In appendix~\ref{app:qcd_ineq} we show
that this is not in contradiction with QCD inequalities.  Finally, we
note that for two flavors and in the chiral limit, the relation
\eqref{eq:bc2c} was obtained earlier by Fukushima
\cite{Fukushima:2008su}, see also \cite{Hands:1999zv,Bittner:2000rf}.
Our result differs from that of \cite{Fukushima:2008su} by a factor of
$1/2$.  The contribution of the zero modes and the positive
definiteness of the measure were not addressed in
\cite{Fukushima:2008su}.

At $\mu=0$, the spectra of the eigenvalues and singular values are
identical, and so are the spectral densities at the origin. This
implies that the chiral condensate (for $m\to 0$) and the diquark
condensate (for $j\to 0$) are of the same magnitude, which is
consistent with the fact that under a global $\SU(2N_f)$ rotation
these condensates can be rotated into each other. As is well known,
$\mu\ne 0$ breaks this degeneracy.  As $\mu$ increases (for $m\ne 0$),
the diquark condensate remains exactly zero until a critical value
$\mu_c=m_\pi/2$ is reached, and then starts growing for $\mu>\mu_c$
\cite{Kogut:2000ek}. This behavior can be naturally understood by our
Banks-Casher-type relation. For $m\ne0$ the relation reads $\langle
\psi\psi\rangle \propto \rsv(0;m)$, where $\rsv(\lambda;m)$ stands for
the singular value density of $D(\mu)+m$, as remarked below
\eqref{eq:bc2c}. For sufficiently small $\mu$, all eigenvalues of
$D(\mu)$ are still localized near the imaginary axis, and the density
of the near-zero modes of $D(\mu)+m$ is zero.  As $\mu$ increases, the
eigenvalues spread out more, and for $\mu>\mu_c$ there is a nonzero
density of eigenvalues at $\pm m$. This signals a nonzero $\rsv(0;m)$,
i.e., the onset of diquark condensation.%
\footnote{A similar discussion starting from a different method can be
  found in \cite{Splittorff:2008sw}.}

It is also possible to express the partition function in terms of the
Dirac operator with chiral chemical potential $D_5(\mu)$ defined in
\eqref{eq:D5}.  This is most easily shown working backwards starting
from \eqref{eq:Zj},
\allowdisplaybreaks[4]
\begin{align}
  \label{eq:Z5}
  Z(j)&=\Big\langle\prod_n (\xi_n^2+j^2)^{{N_f}/2}\Big\rangle_\text{YM}
  =\Big\langle j^{N_f|\nu|}
  \Big[{\prod_n}'(\xi_{Rn}^2+j^2){\prod_n}'(\xi_{Ln}^2+j^2)\Big]^{N_f/2}
  \Big\rangle_\text{YM}
  \notag  \\
  &=\Big\langle{\det}^{{N_f}/{2}}\big( D_5(\mu)+j \big) 
  {\det}^{{N_f}/{2}}\big( D_5(-\mu)+j \big)
  \Big\rangle_\text{YM}\,,
\end{align}
where $|\nu|$ is the number of topological zero modes of $D^\dagger
D$, the primed products are only over nonzero singular values, and in
the last line we have used the relationships between the eigenvalues
of $D_5$ and the singular values derived in section~\ref{sec:D5}.

\subsection[QCD with isospin chemical potential ($\beta=2$)]{\boldmath
  QCD with isospin chemical potential ($\beta=2$)}
\label{sec:iso}

It follows from \eqref{eq:app_main_result_beta=2} that the measure is
positive definite only if the combination
$\ee^{i\theta}(-\rho\lambda^*)$ is real and positive.  We therefore
choose $\theta=0$ and $\rho=-\lambda=j$ with real $j$ (another choice
for which the measure is not positive definite is considered in
appendix~\ref{app:complex-1}).  The partition function is then
\begin{align}
  Z(j)
  =\big\langle\det(D^\dagger D+j^2)\big\rangle_\text{YM}\,,
  \label{eq:Zjiso}
\end{align}
and by a calculation analogous to section~\ref{sec:bc2c} we obtain the
pion condensate
\begin{align}
  \ev{\bar u\gamma_5 d-\bar d\gamma_5 u}=\pi\rsv(0)\,.
\end{align}
Similar comments as in section~\ref{sec:bc2c} apply.  In particular,
we have dropped the contributions of the zero modes (which is
justified because the measure is positive definite), and a proper UV
regularization is understood.  As in section~\ref{sec:bc2c} we can
also express the partition function in terms of $D_5(\mu)$ and obtain
the same result as in \eqref{eq:Z5}.

\subsection[QCD with adjoint fermions ($\beta=4$)]{\boldmath QCD with
  adjoint fermions ($\beta=4$)}

We now take $J_R=j_R\1$ and $J_L=j_L\1$.  The measure in
\eqref{eq:app_main_result_beta=4} is then positive definite only if
the combination $\ee^{i\theta/N_c}(-j_Rj_L^*)^{N_f}$ is real and
positive.  As in section~\ref{sec:bc2c} we therefore assume $\theta=0$
and set $j_R=-j_L=j$ with real $j$ (see appendix~\ref{app:complex-1}
for another choice).  The partition function is then
\begin{align}
  Z(j)
  &=\ev{{\det}^{N_f/2}(D^\dagger D+j^2)}_\text{YM}\,,
\end{align}
and a calculation analogous to section~\ref{sec:bc2c} leads to
\begin{align}
  \ev{\psi^TC\gamma_5\psi}=\frac{N_f}2\pi\rsv(0)\,.
\end{align}
Again, similar comments as in section~\ref{sec:bc2c} apply, and we
could also have expressed the partition function as in \eqref{eq:Z5}.

\section{\boldmath Chiral Lagrangians with diquark source $(\beta=1)$}
\label{sec:eff}

From now on we concentrate on two-color QCD ($\beta=1$) for simplicity
and brevity.  Results similar to those obtained in sections
\ref{sec:eff}--\ref{sec:ls} could also be derived for the theories
with $\beta=2$ and $\beta=4$ using the methods employed here, but we
will not pursue this here.

\subsection{Three different regimes: low, intermediate, and high density}
\label{sec:three}

We now construct the low-energy effective theory for two-color QCD at
nonzero density in the presence of a diquark source and in the chiral
limit.  There are actually three different effective theories,
applicable at low, intermediate, and high density.  In the following
we will refer to them as \L, \I, and \H, respectively.  The effective
theory \L\ is constructed under the assumption of maximal chiral
symmetry breaking at low density \cite{Kogut:1999iv,Kogut:2000ek},
while the effective theory \I\ is constructed assuming the conjectured
BEC-BCS crossover discussed in the introduction.\footnote{From this
  point of view, our results from the effective theory \I\ below
  should be viewed as predictions to be verified in future lattice
  simulations, which would then (dis)confirm the conjectured BEC-BCS
  crossover.  Other possibilities for effective theories in the
  intermediate-density regime have been considered earlier, see, e.g.,
  \cite{Lenaghan:2001sd,Harada:2010vy}.  In the model of
  \cite{Lenaghan:2001sd} certain vector mesons could become massless
  at nonzero density, while in \cite{Harada:2010vy} this does not
  happen.  Lattice studies on this issue are inconclusive
  \cite{Hands:2007uc,Hands:2010gd}.  Here, we assume that all vector
  mesons remain massive at all densities.} On the other hand, the
effective theory \H\ is constructed based on a rigorous weak-coupling
analysis at high density \cite{Kanazawa:2009ks}.

In this subsection we discuss how the three density regimes differ in
their patterns of spontaneous symmetry breaking and the number of
Nambu-Goldstone (NG) modes and comment on the connection between the
three regimes.  The effective theories themselves and mass formulas
for the NG modes will then be derived in the next three subsections,
and their domains of validity will be discussed in
section~\ref{sec:val}.

At very low density one can start from the pattern of chiral symmetry
breaking at \emph{zero} density,
\begin{align}
  \label{eq:SBP_low}
  \SU(2N_f)\to\Sp(2N_f)\,,
\end{align}
and treat the chemical potential and the diquark source as a small
perturbation.  This is the approach taken in
\cite{Kogut:1999iv,Kogut:2000ek}.  The NG modes of the theory are the
same as those of the zero-density theory, i.e., they are collected in
a field $\Sigma$ that parametrizes the coset space
$\SU(2N_f)/\Sp(2N_f)$.  $\SU(N)$ and $\Sp(N)$ have $N^2-1$ and
$N(N+1)/2$ generators, respectively, hence the number of NG modes in
this regime is $N_f(2N_f-1)-1$.  One should keep in mind, however,
that some of these modes acquire a mass as $\mu$ is increased.

At intermediate density, when $\mu$ can no longer be treated as a
small perturbation, we first recall that $\mu$ breaks the original
$\SU(2N_f)$ symmetry to $\SU(N_f)_L\times\SU(N_f)_R\times\U(1)_B$.  A
diquark condensate then breaks this symmetry to
$\Sp(N_f)_L\times\Sp(N_f)_R$ so that the symmetry-breaking pattern is
now
\begin{align}
  \label{eq:SBP}
  \SU(N_f)_L\times\SU(N_f)_R\times\U(1)_B\to\Sp(N_f)_L\times\Sp(N_f)_R\,.
\end{align}
The corresponding NG modes are $\Sigma_L,\Sigma_R\in\SU(N_f)/\Sp(N_f)$
and $V\in\U(1)$, and the total number of NG modes in this regime is
$N_f(N_f-1)-1$.

At very high density the $\U(1)_A$ anomaly is suppressed due to the
screening of instantons
\cite{Son:2001jm,Schafer:2002ty,Schafer:2002yy}.  We therefore need to
take the original $\U(1)_A$ symmetry of the action into account.  It
is no longer broken explicitly by the anomaly but spontaneously by the
diquark condensate so that the symmetry-breaking pattern is
\begin{align}
  \label{eq:SBP_high}
  \SU(N_f)_L\times\SU(N_f)_R\times\U(1)_B\times\U(1)_A
  \to\Sp(N_f)_L\times\Sp(N_f)_R\,.
\end{align}
The NG modes are the same as at intermediate density, except that
there is an additional NG mode corresponding to $\U(1)_A$ which we
call $\eta'$.  In other words, the $\eta'$ mass has become negligible.
Hence the total number of NG modes in this regime is $N_f(N_f-1)$.

Let us make some qualitative comments on the connection between the
three regimes, starting at zero density.  Assuming the diquark sources
to be infinitesimally small, there are $N_f(2N_f-1)-1$ massless NG
modes at $\mu=0$.  As $\mu$ is increased, $N_f(N_f-1)-1$ of them
remain massless while $N_f^2$ of them acquire a $\mu$-dependent mass
\cite{Kogut:2000ek}.  Starting from the effective theory at zero
density, these $N_f^2$ modes can be integrated out.  This yields the
effective theory at intermediate density, which does not depend
explicitly on $\mu$ but contains parameters (i.e., low-energy
constants) that have acquired a $\mu$-dependence through the
integrating-out of the massive modes.  A similar argument applies
starting at high density.  As $\mu$ is lowered, the $\U(1)_A$ anomaly
reappears so that the $\eta'$ becomes massive and can be integrated
out.  These comments are illustrated in figure~\ref{fig:NG} and will
be made more quantitative in the next subsections.

\begin{figure}
  \centering
  \includegraphics{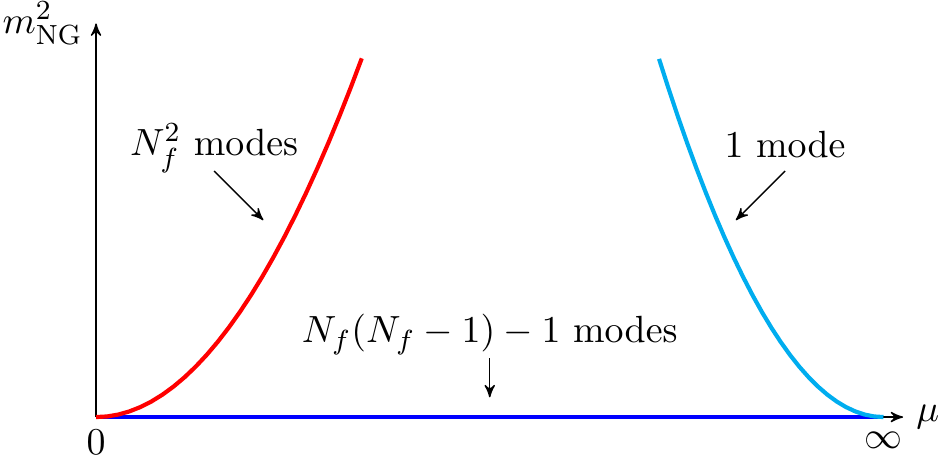}
  \caption{Schematic $\mu$-dependence of the masses of the NG modes.
    See text for details.}
  \label{fig:NG}
\end{figure}

Before proceeding, let us discuss a somewhat subtle issue regarding
the parity of the diquark condensate.  In general, when constructing
an effective theory, one starts from an \emph{assumption} of how the
symmetries of the theory are broken, i.e., this assumption is an input
to the effective theory.  In the construction of the effective
theories below, one assumption we have to put in is whether the scalar
or the pseudoscalar diquark condensate minimizes the ground-state
energy.  For nonzero mass and zero diquark source it was shown by QCD
inequalities that if a diquark condensate forms, it does so in the
scalar channel \cite{Kogut:1999iv}.  However, for zero mass and
nonzero diquark source QCD inequalities do not provide any information
on this question, see appendix~\ref{app:qcd_ineq}.  The assumption we
will make is that the diquark condensate again forms in the scalar
channel.  This assumption is based on instanton dynamics at high
density, see, e.g., the review \cite{Rapp:1999qa}. The instanton
vertex is $c\det(\psi_L^{\dag} \psi_R) + \text{h.c.}$ with $c>0$,
where we have $N_f$ legs each for $\psi_R$ and $\psi_L$
\cite{'tHooft:1976fv}.  Taking the expectation value with respect to
the diquark-condensed ground state, the contribution of the instanton
vertex to the energy is $c' \ev{\smash{\psi_L^\dag
    \psi_L^\dag}}^{N_f/2} \ev{\psi_R \psi_R}^{N_f/2}$ with $c'>0$.
For $N_f=4n+2$ with $n\in\N$ we obtain $c'(\ev{\smash{\psi_L^\dag
    \psi_L^\dag}} \ev{\psi_R \psi_R})^{2n} \times
\ev{\smash{\psi_L^\dag \psi_L^\dag}}\ev{\psi_R \psi_R}$, which is
negative for $\ev{\psi_L \psi_L} = -\ev{\psi_R \psi_R}$ but positive
for $\ev{\psi_L \psi_L} = \ev{\psi_R \psi_R}$. Therefore, the
positive-parity state is favored by instantons.  We assume that this
argument carries over to low density based on the conjectured BEC-BCS
crossover.  For $N_f = 4n$ the argument does not apply since the
contribution to the energy is $c'(\ev{\smash{\psi_L^\dag
    \psi_L^\dag}}\ev{\psi_R \psi_R})^{2n}$, which is not affected by a
relative sign between $\ev{\psi_L \psi_L}$ and $\ev{\psi_R \psi_R}$.
Although in the construction of the effective theories below we do not
distinguish between $N_f=4n$ and $N_f = 4n+2$, we will see that for
$N_f=4n$ the positive- and negative-parity states turn out to be
degenerate, while for $N_f=4n+2$ the pseudoscalar condensate is
suppressed or vanishes completely.  This is also consistent with our
analysis in section~\ref{sec:bc2c}: For $N_f=4n$ the microscopic
theory always has a positive definite measure so that we can derive
Banks-Casher-type relations for both the scalar and the pseudoscalar
diquark condensate, showing that their magnitudes are equal.  For
$N_f=4n+2$ the pseudoscalar diquark source leads to an indefinite
measure in the microscopic theory, and therefore we cannot conclude
anything about the relative magnitude of the two condensates.

\subsection{Effective theory \L\ at low density}
\label{sec:eff_low}

At low density the effective theory derived in
\cite{Kogut:1999iv,Kogut:2000ek} uses as degrees of freedom the NG
modes corresponding to the symmetry-breaking pattern
\eqref{eq:SBP_low} at zero density.  We briefly review and extend the
relevant results here.  In the chiral limit, the leading-order
effective Lagrangian is given by
\begin{align}
  \label{eq:Lefflow}
  \mL_\text{eff}^\L=\frac{F^2}2\tr(\nabla_\nu\Sigma\nabla_\nu\Sigma^\dagger)
  -\Phi_{\L}\re\tr(\bar J\Sigma)
\end{align}
with \allowdisplaybreaks[4]
\begin{align}
  \nabla_\nu\Sigma&=\partial_\nu\Sigma-\mu\delta_{\nu0}
  (B\Sigma+\Sigma B)\,,\\
  \nabla_\nu\Sigma^\dagger&=\partial_\nu\Sigma^\dagger
  +\mu\delta_{\nu0}(\Sigma^\dagger B+B\Sigma^\dagger)\,,\\
  B&=\begin{pmatrix}\1_{N_f} & 0 \\ 0 & -\1_{N_f}\end{pmatrix},
  \qquad \bar J=\begin{pmatrix} J_L & 0 \\ 0 & -J_R^\dagger \end{pmatrix}.
  \label{eq:Jbar}
\end{align}
The field $\Sigma$ parametrizes the coset space $\SU(2N_f)/\Sp(2N_f)$,
\begin{align}
  \label{eq:param_low}
  \Sigma=U\Sigma_dU^T\quad\text{with}\quad
  \Sigma_d=\begin{pmatrix}I\;&0\\0\;&-I\end{pmatrix}\quad\text{and}\quad
  U=\exp\left(\frac{i \pi^a T^a}{2F} \right),
\end{align}
where the $T^a$ ($a=1,\ldots,N_f(2N_f-1)-1$) are the generators of
$\SU(2N_f)/\Sp(2N_f)$.  They are Hermitian and satisfy $\tr(T^a)=0$
and $\tr(T^aT^b)=\delta_{ab}$.  A sum over the repeated index $a$ in
\eqref{eq:param_low} is understood.

There are four minor differences with respect to \cite{Kogut:2000ek}.
First, only the case of $J_R=-J_L$ was considered there.  Second, our
$\Phi_{\L}$ corresponds to their $G$.  Third, the analysis in
\cite{Kogut:2000ek} also includes a nonzero quark mass $m$.  We only
consider $m=0$, in which case the chiral condensate immediately
disappears as $\mu$ is switched on, while the diquark condensate
immediately assumes its full value, see section~12 of
\cite{Kogut:2000ek}.  Fourth, to be consistent with the rest of the
current paper, our convention for the diquark source and consequently
for $\Sigma_d$ differs from \cite{Kogut:2000ek}.  This follows from
requiring parity invariance of the microscopic Lagrangian (which
implies $J_L\leftrightarrow-J_R$ under parity), of the diquark
condensate (which determines the form of $\Sigma_d$), and of the
effective theory (which determines the form of $\bar J$ and the second
term in \eqref{eq:Lefflow}).  As a consistency check, we note that the
minimum of the action is obtained for $\Sigma=\Sigma_d$ and
$J_R=-J_L=jI$ with $j$ real and positive.

In \eqref{eq:Lefflow} there are two low-energy constants. $F$ is the
common decay constant of all NG modes, and the positive parameter
$\Phi_{\L}$ is the magnitude of the diquark condensate per flavor and
handedness in the absence of sources and at $\mu=0$,
\begin{align}
  \label{eq:PhiL}
  \Phi_{\L}=\frac1{N_f}\left|\ev{\psi^T_iC\tau_2I\psi_i}
  \right|_{\bar J=0,\,\mu=0}\qquad(i=L,R)\,.
\end{align}
It is important to note that $F$ and $\Phi_{\L}$ do \emph{not} depend
on $\mu$.

For $J_R=-J_L=jI$ the masses of the NG modes have been computed in
\cite[eq.~(101)]{Kogut:2000ek}, and in the chiral limit we have
$\phi=\alpha=\pi/2$ and $m_\pi=\sqrt{j\Phi}/F$ in these expressions.
Hence there are two types of NG modes,
\begin{subequations}
  \label{eq:GOR_low}
  \begin{align}
    \label{eq:GOR_low_a}
    \text{type 1: } &&& \text{mass}=\sqrt{j\Phi_\L/F^2} &&
    (N_f^2-N_f-1\text{ modes})\,,\\
    \label{eq:GOR_low_b}
    \text{type 2: } &&& \text{mass}=\sqrt{j\Phi_\L/F^2+(2\mu)^2} &&
    (N_f^2\text{ modes})\,.
  \end{align}
\end{subequations}
While the type-1 modes are massless in the $j\to0$ limit,
\eqref{eq:GOR_low_b} shows that $\mu$ is an explicit symmetry-breaking
parameter that makes the type-2 modes massive even for $j=0$.

Let us comment on the source $J_L=J_R$ for the pseudoscalar diquark
condensate.\footnote{This is equivalent to $J_L=-J_R$ with
  $\theta=N_f\pi/2$, see section~\ref{sec:micro1}, i.e., instead of
  studying the $J$-dependence we could equivalently study the
  $\theta$-dependence.  This statement applies at any density.}  For
$N_f=4n$ we define $u=\diag(i\1_{4n},\1_{4n})$.  Since $u\in\SU(8n)$
we can redefine $U\to uU$ in \eqref{eq:param_low}.  The measure is not
changed by this transformation, but in the term $\bar J\Sigma$ in
\eqref{eq:Lefflow} the sign of $J_L$ in \eqref{eq:Jbar} is flipped.
This means that the partition functions (and hence the energies) for
$J_L=-J_R$ and for $J_L=J_R$ are exactly equal.  This is consistent
with the observation that the instanton vertex does not prefer one
condensate over the other for $N_f=4n$, see the discussion at the end
of section~\ref{sec:three}.  For $N_f=2$ we find that a pseudoscalar
diquark source term drops out from \eqref{eq:Lefflow} since
$\re\tr(\bar J\Sigma)=\re\tr[\diag(I,I)U\diag(I,-I)U^T]=0$ for
$U\in\SU(4)$.\footnote{This can be shown using an explicit
  parameterization of the coset space $\SU(4)/\Sp(4)$, see, e.g.,
  \cite{Brauner:2006dv}.}  Hence the pseudoscalar diquark condensate
is zero (in this order of the low-energy expansion), which is again
consistent with the instanton-based argument in
section~\ref{sec:three}.  For $N_f=4n+2$ with $n\ge1$ the situation is
more complicated.  While we currently cannot make any definite
analytical statement, numerical experimentation indicates that for
$J_R=J_L=jI$ the minimum of the energy is larger than for
$J_R=-J_L=jI$.  If true, this would mean that by changing the diquark
source we can generate a nonzero pseudoscalar diquark condensate whose
magnitude is smaller than that of the scalar condensate.

\subsection{Effective theory \H\ at high density}
\label{sec:eff_high}

For technical reasons we now proceed to the regime of very high
density.  The effective chiral Lagrangian including mass term for this
case was derived in \cite{Kanazawa:2009ks}.  Here, we are interested
in diquark sources and therefore set the mass term to zero.

We first consider the case of $N_f\ge4$.  The symmetry-breaking
pattern is given in \eqref{eq:SBP_high}.  By forming linear
combinations of the generators of $\U(1)_B$ and $\U(1)_A$ we can
switch from $\U(1)_B\times\U(1)_A$ to $\U(1)_L\times\U(1)_R$.  We
parametrize the NG modes as
\begin{align}
  \label{eq:param_high}
  \Sigma_i&=U_i I U_i^T \quad\text{with}\quad
  U_i=\exp\left(\frac{i \pi_i^a T^a}{2\fh} \right) \quad (i=L,R)\,,\\
  \label{eq:param_high_LR}
  L&=\exp\left(\frac{i\pi_L^0}{\sqrt{N_f}\fh_0}\right), \qquad 
  R=\exp\left(\frac{i\pi_R^0}{\sqrt{N_f}\fh_0}\right).
\end{align}
The parameterization of the $\Sigma_i$ is similar to
\eqref{eq:param_low}, except that the $T^a$
($a=1,\ldots,N_f(N_f-1)/2-1$) are now the Hermitian generators of
$\SU(N_f)/\Sp(N_f)$, again satisfying $\tr(T^a)=0$ and
$\tr(T^aT^b)=\delta_{ab}$.\footnote{The $T^a$ are related to the $X^a$
  in \cite{Kanazawa:2009ks} by $T^a=X^a/\sqrt{N_f}$.}  Using
$U_iIU_i^T=U_i^2I$ \cite{Kogut:2000ek} 
we can also write $\Sigma_i=U_i^2I$.

The quarks transform under $\SU(N_f)_{L}\times \SU(N_f)_{R}\times
\U(1)_L\times \U(1)_R$ as
\begin{align}
  \psi_L\to \ee^{i\alpha_L}g_L\psi_L\,,\qquad
  \psi_R\to \ee^{i\alpha_R}g_R\psi_R\,,
\end{align}
where $g_i\in\SU(N_f)_i$ and $\ee^{i\alpha_i}\in\U(1)_i$ ($i=L,R$).
The NG modes therefore transform as
\begin{align}
  \Sigma_i\to g_i\Sigma_i g_i^T\quad (i=L,R)\,,\qquad
  L\to L\ee^{2i\alpha_L}\,,\qquad R\to R\ee^{2i\alpha_R}\,.
\end{align}
The transformation properties of $J_L$ and $J_R$ are determined by
requiring that the La\-gran\-gian \eqref{eq:Lf} be invariant under the
flavor symmetries.  This implies
\begin{align}
  J_L\to g_L^*J_Lg_L^\dagger\ee^{-2i\alpha_L}\,,\qquad J_R\to
  g_R^*J_Rg_R^\dagger\ee^{-2i\alpha_R}\,.
  \label{eq:transformation}
\end{align}
Therefore the invariant real combination linear in $J_L$ and $J_R$ is
uniquely determined to be
\begin{align}
  \label{eq:inv}
  \re\big[L\tr(J_L\Sigma_L)-R\tr(J_R\Sigma_R)\big]\,.
\end{align}
As for theory \L, we required parity invariance of the microscopic
theory (implying $J_L\leftrightarrow-J_R$), of the diquark condensate
(implying $\Sigma_L\leftrightarrow\Sigma_R$ and $L\leftrightarrow R$),
and of the effective theory (leading to the relative factor of $-1$ in
\eqref{eq:inv}).  The leading-order effective Lagrangian in the
presence of diquark sources and in the chiral limit is thus given by
\begin{align}
  \label{eq:lagrangian_high}
  \mL_\text{eff}^\H&=\bigg[\frac{N_f\fh_0^2}{2}\big(|\partial_0 L|^2
  +\vh_0^2|\partial_i L|^2 \big)+\frac{\fh^2}{2}\tr\big(
  |\partial_0 \Sigma_L |^2+\vh^2 |\partial_i \Sigma_L |^2 \big)
  + (L \leftrightarrow R)\bigg]\notag\\
  &\quad-\Phi_{\H} \re\tr(J_LL\Sigma_L-J_RR\Sigma_R)
  -\frac{2\fh_0^2}{N_f}m_\text{inst}^2\re\,(L^\dagger R)^{N_f/2}\,,
\end{align}
where $\fh_0,\fh$ and $\vh_0,\vh$ are low-energy constants that
correspond to the decay constants and velocities of the NG modes,
respectively.  The latter are generally different from unity (i.e.,
the speed of light) since Lorentz invariance is lost at $\mu\ne0$.
The minus sign in front of the positive low-energy constant
$\Phi_{\H}$ is chosen so that the minimum of the action is obtained
for $L=R=1$, $\Sigma_L=\Sigma_R=I$, and $J_R=-J_L=jI$ with $j$ real
and positive.  Note that all low-energy ``constants'' in
$\mL_\text{eff}^\H$ depend on $\mu$.  Their relation to physical
observables will be discussed in section~\ref{sec:match}.

In \eqref{eq:lagrangian_high} we have also included a term that
corresponds to the single-instanton contribution to the $\eta'$ mass,
parametrized by $m_\text{inst}$.  This term is symmetric under the
anomaly-free subgroup $(\Z_{2N_f})_A$ of $\U(1)_A$.  From the symmetry
point of view all terms of the form $\re(L^{\dag} R)^{nN_f/2}$ ($n
\geq 1$) are allowed and contribute to the $\eta'$ mass.
Microscopically, these terms correspond to $n$-instanton vertices. At
sufficiently large $\mu$, the instanton ensemble can be regarded as a
dilute gas \cite{Son:2001jm,Schafer:2002ty,Schafer:2002yy} which does
not form instanton molecules \cite{Yamamoto:2008zw}. The diluteness of
the instanton gas is parametrized by the dimensionless quantity
(proportional to the instanton density) $a \propto (\Lambda_{\rm
  QCD}/\mu)^{b(N_f)} \ll 1$ with $b(N_f)= (22-2N_f)/3$.  Since the
probability $\sim a^n$ that $n$ instantons ($n \geq 2$) are at the
same point is highly suppressed we can neglect the multi-instanton
vertices and only keep the one-instanton contribution.  Note that
$m_\text{inst}$ is a decreasing function of $\mu$, with
$m_\text{inst}\to0$ for $\mu\to\infty$
\cite{Son:2001jm,Schafer:2002ty,Schafer:2002yy}.

It is convenient to combine the $\U(1)_i$ field with
$\SU(N_f)_i/\Sp(N_f)_i$ by defining
\begin{align}
  \label{eq:Sigma_tilde}
  \tilde \Sigma_L=L \Sigma_L = \exp\left(\frac{i\pi_L^A T^A}{\fh_A}
  \right)I \quad\text{and}\quad (L \leftrightarrow R)\,,
\end{align}
where the $T^A$ ($A=0,\ldots,N_f(N_f-1)/2-1$) are now the generators
of $\U(N_f)/\Sp(N_f)$ with $T^0=\1/\sqrt{N_f}$ so that
$\tr(T^AT^B)=\delta_{AB}$.\footnote{We use the convention that
  uppercase indices corresponding to $\U(N_f)/\Sp(N_f)$ start at zero,
  while lowercase indices corresponding to $\SU(N_f)/\Sp(N_f)$ start
  at 1.}  We also defined $\fh_A=\fh$ for $A\ge1$.  To second order in
the $\pi$-fields we have
\begin{align}
  \label{eq:re_Sigma_tilde}
  \re\tilde \Sigma_i=\left(1-\frac{\pi_i^A\pi_i^B T^A T^B}
    {2\fh_A \fh_B}\right)I\,.
\end{align}
Assuming $J_R=-J_L=jI$ with $j$ real and positive, this yields a
Gell-Mann--Oakes--Renner (GOR) type mass formula for the $\pi_i^A$.
As in theory \L\ there are two types of NG modes,
\begin{subequations}
  \label{eq:GOR_high}
  \begin{align}
    \label{eq:GOR_high_a}
    \text{type 1: } &&& m_A=\sqrt{j\Phi_{\H}/\fh_A^2} &&
    (N_f^2-N_f-1\text{ modes})\,,\\
    \label{eq:GOR_high_b}
    \text{type 2: } &&& m_{\eta'}=\sqrt{j\Phi_\H/\fh_0^2+m_\text{inst}^2}
    && (1\text{ mode})\,.
  \end{align}
\end{subequations}
Note that there are two type-1 modes for each $A\ge1$, but only a
single one for $A=0$.  Note also the similarity with
\eqref{eq:GOR_low}.  Now $m_\text{inst}$ plays the role of the
symmetry-breaking parameter which makes the type-2 mode massive as
$\mu$ is lowered.

Let us now consider the case of $N_f=2$, in which $\Sigma_L$ and
$\Sigma_R$ are absent because $\SU(2)\sim\Sp(2)$.  Thus the effective
Lagrangian contains only the fields $L$ and $R$.  Since any $2\times
2$ antisymmetric matrix is proportional to $I$ we can write $J_L=j_L
I$ and $J_R=j_RI\ (j_L,\,j_R\in\C)$ without loss of generality.  These
terms transform as
\begin{align}
  \label{eq:transf2}
  j_L\to j_L \ee^{-2i\alpha_L}\,,\qquad j_R\to j_R\ee^{-2i\alpha_R}\,,
\end{align}
which follows from \eqref{eq:transformation} and
$g_i^*Ig_i^\dagger=(\det g_i^*)I=I$ ($i=L,R$). Therefore the invariant
real combination linear in $j_R$ and $j_L$ is given by
\begin{align}
  \re\big(j_LL-j_RR\big)\,.
\end{align}
Using the same parameterization of $L$ and $R$ as in
\eqref{eq:param_high_LR}, the effective Lagrangian for $N_f=2$ reads
\begin{align}
  \label{eq:lagrangian_high2}
  \mL_\text{eff}^\H=\fh_0^2\left[|\partial_0 L|^2+\vh_0^2|\partial_i L|^2 
    + (L \leftrightarrow R)\right] + 2\Phi_{\H} \re(j_LL-j_RR) 
  -\fh_0^2m_\text{inst}^2 \re(L^\dagger R)\,,
\end{align}
where the plus sign in front of $\Phi_\H$ and the factor of $-1$
between the two terms following it have been chosen so that the
minimum of the action is obtained for $L=R=1$ and $j_R=-j_L=j$ with
$j$ real and positive.  The GOR-type relation for this case is
identical to \eqref{eq:GOR_high} with $A=0$.

We again comment on the case of $J_L=J_R$.  Note first that for
$m_\text{inst}=0$ (i.e., at asymptotically high density) the fields
$L$ and $R$ in \eqref{eq:lagrangian_high} or
\eqref{eq:lagrangian_high2} can rotate independently, and hence the
left and right diquark condensates in the ground state can be rotated
separately by varying the directions of $J_L$ and $J_R$.  This is no
longer true when $m_\text{inst}\ne 0$.  Since the anomaly term favors
$L=R$ energetically, the fields $L$ and $R$ can no longer rotate
independently.  Now let us again discuss the various cases of $N_f$.
For $N_f=4n$ we can redefine $U_L\to uU_L$ in \eqref{eq:param_high}
with $u=\diag(i\1_{4n})\in\SU(4n)$, which flips the sign of $\Sigma_L$
and thus absorbs a sign flip of $J_L$ in \eqref{eq:lagrangian_high}.
So again the energies for $J_L=-J_R$ and $J_L=J_R$ are equal, in
agreement with the instanton-based argument in
section~\ref{sec:three}.  For $N_f=2$ and $j_R=-j_L>0$ the diquark
source term and the anomaly term in \eqref{eq:lagrangian_high2} are
minimized simultaneously at $L=R=1$, and thus the ground state is not
changed by the anomaly term. However, for $j_R=j_L>0$ there is a
competition between these two terms, and they cannot be minimized
simultaneously. Therefore the pseudoscalar diquark condensate can be
realized only if the diquark sources are strong enough to overcome the
penalty due to the anomaly term.  For $N_f=4n+2$ with $n\ge1$ the
situation is similar but a bit more complicated.  For $J_R=-J_L=jI$
both the diquark source term and the anomaly term in
\eqref{eq:lagrangian_high} can be minimized simultaneously.  For
$J_R=J_L=jI$ we suspect, although we currently cannot prove it
analytically, that this cannot be done.  A more quantitative study is
required to determine the magnitude of the pseudoscalar diquark
condensate in this case.

\subsection{Effective theory \I\ at intermediate density}
\label{sec:eff_int}

At intermediate density, the coupling constant is not small enough to
treat instantons as a dilute screened gas, and hence the $\U(1)_A$
anomaly can no longer be treated as a small perturbation.  In other
words, $m_\text{inst}$ increases as $\mu$ is lowered so that the
$\eta'$ mass in \eqref{eq:GOR_high_b} is not necessarily small,
implying that the $\eta'$ should be integrated out from the effective
theory.  Technically, this means that $L$ and $R$ should be replaced
by a single $\U(1)$ field $V$, i.e., the effective Lagrangian
\eqref{eq:lagrangian_high} for $N_f\ge4$ changes to
\begin{align}
  \mL_\text{eff}^\I&=N_ff_0^2\left[|\partial_0 V|^2
  +v_0^2|\partial_i V|^2 \right] +\frac{f^2}{2}\tr\left[|\partial_0
  \Sigma_L |^2 +v^2 |\partial_i \Sigma_L |^2 
  + (L \leftrightarrow R)\right]\notag\\
  &\quad -\Phi_{\I} \re\left[V\tr (J_L\Sigma_L-J_R\Sigma_R)\right]
  \label{eq:int}  
\end{align}
corresponding to the symmetry-breaking pattern \eqref{eq:SBP}.  We
have renamed the low-energy ``constants'' to distinguish them from
those of theory \H.  Note that they again depend on $\mu$.  In
section~\ref{sec:match} we will discuss how they are related to the
low-energy constants of theory \L\ and \H\ at low and high density,
respectively.

The NG modes $\Sigma_{L/R}$ are pa\-ra\-me\-trized as in
\eqref{eq:param_high} with $\fh$ replaced by $f$, while
\begin{align} 
  V&=\exp\left(\frac{i\pi_V^0}{\sqrt{2N_f}f_0}\right).
  \label{eq:param_V}
\end{align}
To second order in the $\pi$-fields we now have
\begin{align}
  \label{eq:SigmaV}
  \re (\Sigma_iV)=\left(1-\frac{(\pi_V^0)^2}{4N_ff_0^2}
    -\frac{\pi_V^0\pi_i^aT^a}{\sqrt{2N_f}f_0f}
    -\frac{\pi_i^a\pi_i^bT^aT^b}{2f^2}\right)I\,.
\end{align}
Assuming again $J_R=-J_L=jI$ with $j$ real and positive, we obtain a
GOR-type mass formula analogous to \eqref{eq:GOR_high_a},
\begin{align}
  \label{eq:GOR_int}
  m_A=\sqrt{j\Phi_\I/f_A^2}\qquad(N_f^2-N_f-1\text{ modes})\,,
\end{align}
where $f_A=f$ for $A\ge1$.  Note that we have only type-1 modes in
theory \I.

For $N_f=2$ the effective Lagrangian changes to
\begin{align}
  \label{eq:int2}
  \mL_\text{eff}^\I=2f_0^2\left[|\partial_0 V|^2                         
  +v_0^2|\partial_i V|^2 \right]+2\Phi_{\I}\re[(j_L-j_R)V]\,,
\end{align}
and the GOR-type relation for the single NG mode is
$m_0=\sqrt{j\Phi_{\I}/f_0^2}$ as in \eqref{eq:GOR_int}.

Finally, we again consider the case of $J_L=J_R$.  For $N_f=4n$ the
argument and the conclusion are exactly the same as in theory \H.  For
$N_f=2$ the pseudoscalar diquark source term in \eqref{eq:int2} drops
out trivially since $j_L=j_R$, and hence the pseudoscalar diquark
condensate is zero (in this order of the low-energy expansion) as in
theory \L.  For $N_f=4n+2$ we suspect, although we currently cannot
prove it analytically, that for $J_R=J_L=jI$ the minimum of the
ground-state energy is larger than for $J_R=-J_L=jI$.  If true, we are
led to the same conclusion as in theory \L.

\subsection{Domains of validity}
\label{sec:val}

In the following discussion we assume that the diquark sources are
infinitesimal.  In general, there are two conditions for an effective
theory formulated in terms of NG modes to be applicable: (i) the
masses of all NG modes must be much smaller than the mass scale
$m_\ell$ of the lightest non-NG particle, and (ii) the typical scale
$p$ of observables computed within the effective theory must also be
much smaller than $m_\ell$.  Of course, $m_\ell$ itself must be
nonzero.  To figure out the domains of validity of the three effective
theories at nonzero density we must determine the mass scale $m_\ell$
of each theory, which generically is a function of $\mu$.

For the effective theory \L\ we have $m_\ell(\L)\sim \Lambda$, where
$\Lambda$ is the mass of the lightest non-NG particle at zero density.
As $\mu$ is increased from zero, some of the NG modes of \L\ acquire a
mass proportional to $\mu$.  The effective theory \I\ is obtained from
\L\ by integrating out these modes so that for \I\ at low density we
have $m_\ell(\I)\sim\mu$.

The situation at high density is somewhat more complicated.  At
asymptotically high density the $\eta'$ is massless and
$\Delta\ll\mu$.  For the effective theory \H\ we have
$m_\ell(\H)\sim\Delta$.  There are two ways to see this.  First,
$\Delta$ plays the role of a constituent quark mass so that the
lightest non-NG particles (color singlet diquarks and mesons) weigh
about $2\Delta$ \cite{Son:2000by}.  Second, the higher-order vertices
in the effective Lagrangian are suppressed by $1/\Delta$, while loop
integrals are suppressed by $1/\fh_A$ \cite{Schafer:2003vz}. Since
$\fh_A\sim\mu$ \cite{Son:1999cm} and $\Delta\ll\mu$ the cutoff is
$\Delta$.  Let us now lower the density so that the $\eta'$ becomes
massive, but let $\mu$ be large enough so that we still have
$m_{\eta'}\ll\Delta\ll\mu$
\cite{Son:2001jm,Schafer:2002ty,Schafer:2002yy}.  The effective theory
\I\ is now obtained from \H\ by integrating out the $\eta'$ so that
for \I\ in this regime we have $m_\ell(\I)\sim m_{\eta'}$.  As the
density is lowered further there are two possible scenarios:
\begin{enumerate}
\item There could be a ``critical'' chemical potential $\mu_c$ at
  which $m_{\eta'}=\Delta$ and below which $m_{\eta'}>\Delta$.  In
  that case $m_\ell(\I)\sim\Delta$ for
  $\mu_\text{BCS}\lesssim\mu\lesssim\mu_c$, where $\mu_\text{BCS}$ is
  the chemical potential above which we are in the BCS
  regime.\footnote{Although there is no phase transition between the
    BEC and BCS regimes, we can define $\mu_\text{BCS}$ as the
    chemical potential above which the minimum of the dispersion
    relation of the fermionic quasiparticles changes from $p=0$ to $p
    \neq 0$, with $p$ being the momentum \cite{leggett2006quantum}.}
\item We could have $m_{\eta'}<\Delta$ for all
  $\mu\gtrsim\mu_\text{BCS}$.  Then $m_\ell(\I)\sim m_{\eta'}$ for all
  $\mu\gtrsim\mu_\text{BCS}$.
\end{enumerate}
Since we only know the functions $\Delta(\mu)$ and $m_{\eta'}(\mu)$ at
asymptotically large $\mu$ we do not know which of these two scenarios
is correct.\footnote{In these asymptotic functions we always have
  $m_{\eta'}<\Delta$, but this does not tell us anything about the
  regime where $\mu$ is not asymptotically large.} This can only be
decided by a full dynamical calculation, e.g., in lattice QCD.
However, the two scenarios can be combined in the statement
$m_\ell(\I)\sim\min(\Delta(\mu),m_{\eta'}(\mu))$ for
$\mu\gtrsim\mu_\text{BCS}$.\footnote{In the preceding arguments we
  have completely ignored gluons, even though they are lighter than
  the $\eta'$ at sufficiently large $\mu$ \cite{Schafer:2002yy}, since
  their interaction with NG modes is assumed to be negligibly small,
  see the discussion in \cite{Kanazawa:2009ks}.}

\begin{table}[t]
  \centering
  \begin{tabular}{|c||c|c|}
    \hline
    effective theory & condition (i) & condition (ii) \\ \hline \hline
    \L\ & $\mu\ll \Lambda$ & $p\ll \Lambda$\\ \hline
    \I\ (scenario 1) & --- &
    \begin{tabular}{ll}
      $p\ll\mu$ & for $\mu\lesssim \Lambda$\\
      $p\ll\Delta$ & for $\mu_\text{BCS}\lesssim\mu\lesssim\mu_c$\\
      $p\ll m_{\eta'}$ & for $\mu\gtrsim\mu_c$
    \end{tabular}\\ \hline
    \I\ (scenario 2) & --- &
    \begin{tabular}{ll}
      $p\ll\mu$ & for $\mu\lesssim \Lambda$\\
      $p\ll m_{\eta'}$ & for $\mu\gtrsim\mu_\text{BCS}$
    \end{tabular}\\ \hline
    \H\ & $m_{\eta'}\ll\Delta$ & $p\ll\Delta$ \\ \hline
  \end{tabular}
  \caption{Domains of validity of the three effective theories 
    \L, \I, and \H, see text for details.}
  \label{tab:val}
\end{table}

\begin{figure}[t]
  \centering
  \includegraphics{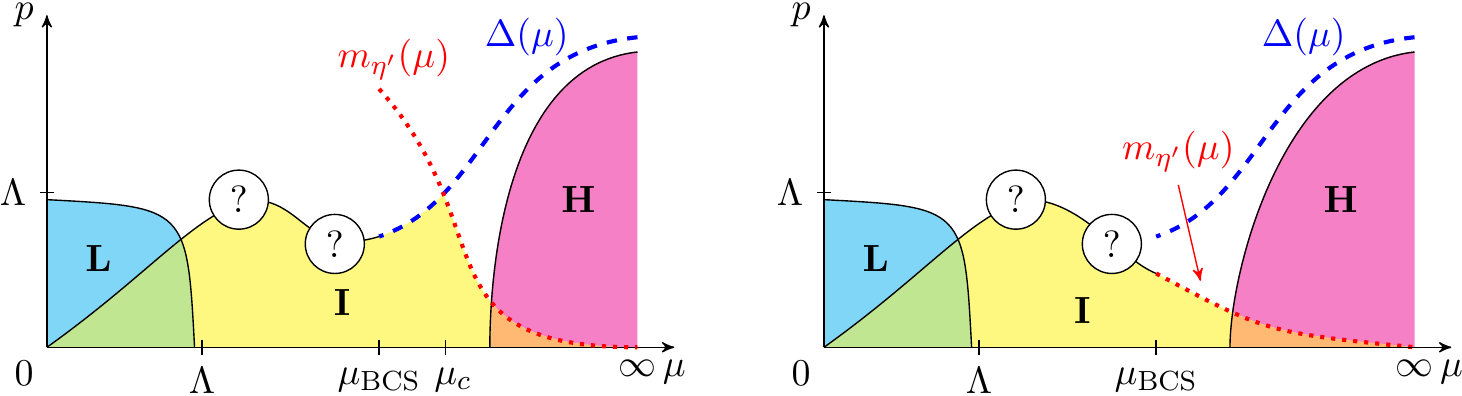}
  \caption{Domains of validity of the three effective theories \L, \I,
    and \H\ in scenario 1 (left) and scenario 2 (right), see text for
    details.}
  \label{fig:val}
\end{figure}

Our discussion is summarized in table~\ref{tab:val} and
figure~\ref{fig:val}.  Three comments are in order.  First, condition
(i) is always satisfied for the effective theory \I\ since the
$N_f(N_f-1)-1$ NG modes shown in figure~\ref{fig:NG} are always
massless.  In other words, \I\ is applicable at any $\mu$ as long as
the scale of the observable is sufficiently small.  Second, as the
``intermediate'' density is increased from ``low'' to ``high'', the
mass scale of the lightest non-NG mode changes from $\mu$ to
$m_{\eta'}$.  Finding the precise $\mu$-dependence of $m_\ell$ in the
intermediate region again requires a full dynamical calculation.  All
we currently know are the two limits $\mu$ and $m_{\eta'}$.  Third, in
the overlap regions of the different effective theories (i.e., the
green and orange areas in figure~\ref{fig:val}) one has a choice of
which theory to use, but this choice depends on the observable.  (This
is similar to the choice between $\SU(2)$ and $\SU(3)$ flavor chiral
perturbation theory in QCD as a function of the strange quark mass
\cite{Gasser:1984gg,Gasser:2007sg}.)  For example, in the regime
$\mu\ll \Lambda$ one could also use the effective theory \I, but this
would only work for observables with $p\ll\mu$, whereas the effective
theory \L\ could be used for observables with $p\ll \Lambda$.
Similarly, the effective theory \I\ could also be used in the regime
$m_{\eta'}\ll\Delta$.  This would only work for observables with $p\ll
m_{\eta'}$, whereas the effective theory \H\ could be used for
observables with $p\ll\Delta$.  While at first sight it may seem that
one should always work with the effective theory that allows for a
larger range of observables, it may be technically simpler to work
with the effective theory \I\ if one is only interested in an
observable for which this effective theory is valid.

\subsection{Matching of the low-energy constants}
\label{sec:match}

Let us now comment on the relation between low-energy constants and
physical observables, and on the matching of the low-energy constants
between the different effective theories.  Note that there are also
low-energy constants corresponding to higher-orders in the effective
Lagrangians, which we have not shown explicitly.

Let us start with theory \I, because this is simplest as the diquark
source is the only symmetry-breaking perturbation. In that case $f_0$
and $f$ are equal to the physical decay constants at a given $\mu$ in
the limit $J\to 0$, and
\begin{equation}
  \Phi_\I = \frac{1}{N_f}
  \left|\big\langle\psi_i^TC\tau_2I\psi_i\big\rangle\right|_{J=0}
  \qquad(i=L,R)\,.
\end{equation}
For nonzero $J$ there will be corrections (similar to the chiral
corrections in chiral perturbation theory) due to which the low-energy
constants will deviate from the physical quantities.  Next we consider
theory \L.  It has two symmetry-breaking perturbations, $\mu$ and $J$.
The low-energy constants do not depend on these external
parameters. Rather, they are equal to the decay constants and the
diquark condensate at $\mu=J=0$, see \eqref{eq:PhiL}. Now let us
recall that the theories \L\ and \I\ have overlapping domains of
validity at low density. All physical quantities should be independent
of which effective theory we use. As long as we work at any
\emph{finite} order of the low-energy expansion in theory $\L$, the
results thus obtained could be different from those of theory
\I. However, the discrepancy will disappear if we sum up the
contributions of the heavier NG modes (with mass $\sim \mu$) in theory
\L\ to all orders.%
\footnote{This is similar to what is encountered in weak-coupling
  perturbation theory.  In this case observables (such as the cross
  section) depend on the renormalization scale at any finite order of
  perturbation theory, but this dependence decreases as we go to
  higher orders \cite{Srednicki:2007qs}.}  We expect, for any fixed
$\mu\ll \Lambda$, the relation
\begin{equation}
  |\langle\psi\psi\rangle|_{J=0} = \Phi_\I 
  = \Phi_\L + \text{(Corrections due to the propagation of 
  type-2 modes)}\,,
\end{equation}
and likewise for $F$, $f_0$, and $f$.  The relations between the
low-energy constants of theory \L\ and \I\ can be made more precise by
explicitly integrating out the type-2 modes of theory \L. In the
course of this procedure, the (initially $\mu$-independent) low-energy
constants acquire a $\mu$-dependence in much the same way as the
low-energy constants of $\SU(2)$ chiral perturbation theory acquire a
dependence on the strange quark mass $m_s$ when kaons are integrated
out of $\SU(3)$ chiral perturbation theory
\cite{Gasser:1984gg,Gasser:2007sg,Gasser:2009hr,Ivanov:2011mr}.
According to these calculations, the corrections to the low-energy
constants at $O(p^2)$ ($F$ and $B$ in the standard notation) due to
the integrating-out of kaons become arbitrarily small when $m_s$ gets
small.  If we assume that this finding persists in our present
context, we expect
\begin{equation}
  \lim_{\mu\to 0}f_0, f = F 
  \qquad \text{and}\qquad 
  \lim_{\mu\to 0}\Phi_\I = \Phi_\L\,. 
\end{equation}
However, we do not expect this smooth matching to extend to the
low-energy constants of higher orders, because it is known that in
$\SU(2)$ chiral perturbation theory they receive corrections of the
form $(1/m_s)^n$ with $n>0$ and thus blow up as $m_s\to 0$.  Hence
they cannot reduce to the low-energy constants of $\SU(3)$ chiral
perturbation theory
\cite{Gasser:1984gg,Gasser:2007sg,Gasser:2009hr,Ivanov:2011mr}.

The discussion for theory \H\ and its matching with theory \I\
proceeds analogously. In theory \H, the role of $\mu$ in theory \L\ is
now played by $m_\text{inst}$, which in turn is a function of
$\mu$. For sufficiently high density the domains of validity of theory
\H\ and \I\ overlap, and for any fixed $\mu$ (provided that
$m_\text{inst}\ll\Delta$) we expect the relation
\begin{equation}
  |\langle\psi\psi\rangle|_{J=0} = \Phi_\I 
  = \Phi_\H + \text{(Corrections due to the propagation of
    $\eta'$)}
\end{equation}
to hold (and likewise for $f_0$, $\fh_0$ and $f$, $\fh$).  Based on
our discussion at low density we now expect the matching (at any fixed
$\mu$)
\begin{equation}
  \label{eq:match_IH}
  \lim_{m_\text{inst}\to 0}f_0, f = \tilde{f}_0, \tilde{f} 
  \qquad \text{and}\qquad 
  \lim_{m_\text{inst}\to 0}\Phi_\I = \Phi_\H\,.
\end{equation}
However, there is a subtlety specific to high density.  For
$\mu\to\infty$ the low-energy ``constants'' are actually infinite
since $f_0,f\sim\mu$ \cite{Son:1999cm} and $\Phi\sim\mu^2\Delta/g$
\cite{Schafer:1999fe}.  Therefore $\fh_0$, $\fh$, and $\Phi_\H$ cannot
be defined as constants at $\mu=\infty$.  Accordingly, theory \H\
cannot be defined at $\mu=\infty$ in the same way as theory \L\ could
be defined at $\mu=0$.\footnote{This is not meant as a negative
  statement since the asymptotic behavior of the low-energy constants
  in terms of $\mu$ is known.  Similar situations occur when
  considering the large-$N_c$ limit or scattering processes in the
  limit of high energies.} That is why we defined theory \H\ and its
low-energy ``constants'' at fixed (and finite) $\mu$.  This implies
that $m_\text{inst}$ is also fixed and cannot be considered as an
independent symmetry-breaking parameter anymore.  However, we can
start from theory \H\ at a given $\mu$ and formally integrate out the
$\eta'$ to obtain theory \I\ at the same value of $\mu$.  The
low-energy ``constants'' of \I\ thus acquire a dependence on
$m_\text{inst}$, and we can now formally send $m_\text{inst}\to 0$,
still at the same fixed $\mu$.  This is how \eqref{eq:match_IH} should
be understood.

At the end of this section, let us comment on possible extensions.  We
could have performed a more comprehensive analysis with nonzero
diquark sources, which would give a small mass proportional to $\sqrt
j$ to the NG modes.  However, in this paper we are only interested in
the limit $j\to0$, and therefore we have not performed such an
analysis.\footnote{But see section~\ref{sec:ss_conn} in which we are
  forced to consider the case of small but nonzero $j$ to understand
  an apparent discontinuity of a very particular observable.}  In
principle one could also add explicit quark masses.  This makes the
analysis even more complicated since the coset space one should use to
construct the effective theory now depends not only on $\mu$ but also
on the quark masses.  From the arguments presented in this section it
should be clear how to proceed, but we do not pursue this issue
further.

\section{\boldmath Smilga-Stern-type relations $(\beta=1)$}
\label{sec:ss}

In \cite{Smilga:1993in} Smilga and Stern computed the slope of the
density of Dirac eigenvalues at the origin in the QCD vacuum
($\beta=2$) using effective-theory techniques.  Their result was
confirmed by partially quenched chiral perturbation theory for
degenerate \cite{Osborn:1998qb} and nondegenerate
\cite{Zyablyuk:1999aj} masses.  Later it was generalized to theories
with $\beta=1$ and $4$ at $\mu=0$ \cite{Toublan:1999hi}.  In this
section we adapt the method of \cite{Smilga:1993in} to the singular
values at $\mu\ne 0$, i.e., we compute the slope of the singular value
density of the Dirac operator in two-color QCD at nonzero $\mu$
($\beta=1$), using the effective theories constructed in
section~\ref{sec:eff}.  We will obtain three different results at
infinite, intermediate, and zero density, respectively.  In
section~\ref{sec:ss_conn} we will discuss the relation between these
results as a function of $\mu$.  Throughout this section we work in
the chiral limit for simplicity.

\subsection{Infinite density}
\label{sec:ss_high}

For technical reasons we now start at infinite density so that we can
set $m_\text{inst}=0$ in \eqref{eq:lagrangian_high}.  As in earlier
sections we set $J_R=-J_L=J$, where $J$ is an antisymmetric $N_f\times
N_f$ matrix that has $N_f(N_f-1)/2$ independent components.  We can
decompose $J$ as
\begin{align}
  \label{eq:J}
  J=I\sum_Aj_At^A=jI+I\sum_aj_at^a\,,
\end{align}
where the $t^A$ are the generators of $\U(N_f)/\Sp(N_f)$ and the $j_A$
are real parameters with $j_0=j\sqrt{N_f}$.  Such a decomposition is
possible since the dimension of $\U(N_f)/\Sp(N_f)$ is $N_f(N_f-1)/2$
and thus equal to the number of degrees of freedom of $J$, and since
$It^A$ is antisymmetric for all $A$ (which follows from
$t^AI=I(t^A)^T$ \cite{Kogut:2000ek}).  As before, the sum over $A$
starts at 0, while the sum over $a$ starts at 1.  The $t^A$ are
identical to the $T^A$ defined below \eqref{eq:Sigma_tilde}, but we
denote them by a different symbol (in agreement with the notation of
\cite{Toublan:1999hi}) since they are used in a different context: The
$T^A$ are used to parametrize the NG modes of the effective theory
living in $\U(N_f)_L/\Sp(N_f)_L\times\U(N_f)_R/\Sp(N_f)_R$, while the
$t^A$ are used to parametrize the source $J$, which exists already in
the microscopic theory.  The choice of $\U(N_f)/\Sp(N_f)$ to
parametrize $J$ is natural since it is the space in which the diquark
condensate aligns depending on the choice of $J$.

For $N_f\ge4$ we now consider the scalar susceptibility
\begin{align}
  \label{eq:Kab}
  K_{ab}(j)=\lim_{V_4\to\infty}\frac1{V_4}\partial_{j_a}\partial_{j_b}
  \ln Z(J)\Big|_{\text{all }j_a=0}\,,
\end{align}
which we will calculate both from the microscopic theory (two-color
QCD) and from the low-energy effective theory.

Let us start on the QCD side.  Assuming that at high density there are
no zero modes, \eqref{eq:ZJ} yields
\begin{align}
  Z(J)=\ev{{\det}^{1/2}(D^\dagger D+J^\dagger J)}_\text{YM}
  =\Big\langle\prod_n{\det}^{1/2}(\xi_n^2+J^\dagger J)\Big\rangle_\text{YM}\,.
\end{align}
After a bit of algebra we obtain from \eqref{eq:Kab}
\begin{align}
  K_{ab}(j) &= \delta_{ab}\lim_{V_4 \rightarrow \infty} \frac{1}{V_4}
  \bigg\langle \sum_n\frac{\xi_n^2-j^2}{(\xi_n^2+j^2)^2} \bigg\rangle_j 
  = \delta_{ab}\int_0^{\infty} d\xi \,\rsv(\xi;j)
  \frac{\xi^2-j^2}{(\xi^2+j^2)^2}\,,
  \label{eq:KQCD}
\end{align}
where $\rsv(\xi;j)$ is defined as in \eqref{eq:rsv} but with nonzero
$j$.  

On the low-energy effective theory side we need to differentiate the
log of the effective partition function
\begin{align}
  \label{eq:Zeff_high}
  Z_\text{eff}^\H=\int \dd\tilde\Sigma_L\,\dd\tilde\Sigma_R\,
  \exp\bigg(-\int d^4x\,\mL_\text{eff}^\H\bigg)
\end{align}
with respect to $j_a$ and $j_b$.  Using \eqref{eq:lagrangian_high}
with $m_\text{inst}=0$ and \eqref{eq:re_Sigma_tilde} we obtain
\begin{align}
  \label{eq:Keff}
  K_{ab}(j)=&\sum_{ABCD}
  \frac{\Phi_\H^2}{4\fh_A \fh_B \fh_C \fh_D}
  \tr(t^aT^AT^B)\tr(t^bT^CT^D) \notag\\
  &\qquad \times \frac{1}{V_4} \left\langle \int d^4x\, d^4y \,
  (\pi_L^A \pi_L^B + \pi_R^A \pi_R^B)(x)(\pi_L^C \pi_L^D + \pi_R^C
  \pi_R^D)(y) \right\rangle_j \notag\\
  =&\sum_{AB} \frac{\Phi_\H^2}{8\fh_A^2 \fh_B^2}
  \tr(t^a \{T^A, T^B\})\tr(t^b \{T^A, T^B\})\notag\\
  &\qquad\times\left\langle \int d^4x \, \Big(\pi_L^A \pi_L^B(x)
    \pi_L^A \pi_L^B(0) + \pi_R^A \pi_R^B(x)  
    \pi_R^A \pi_R^B(0)\Big) \right\rangle_j^{\! \connect},
\end{align}
where ``\connect'' denotes the connected part of the correlation
function.  To obtain \eqref{eq:Keff} we have done the contractions
$A=C$, $B=D$ and $A=D$, $B=C$ and symmetrized in $A,B$.  The
contraction $A=B$, $C=D$ corresponds to disconnected diagrams and
yields zero.\footnote{The contribution of this contraction is
  proportional to $\sum_b\tr(t^aT^bT^b)\propto\tr(t^a)=0$ since
  $\sum_bT^bT^b$ is the difference of the quadratic Casimir operators
  of $\SU(N_f)$ and $\Sp(N_f)$ and hence proportional to $\1$.}  The
dependence of \eqref{eq:Keff} on $j$ is contained in the masses of the
NG modes.  Evaluating the connected part in one-loop approximation%
\footnote{This is a valid approximation as we are interested in the
  infrared limit of the theory.}  we find that it diverges for
$j\to0$,
\begin{align}
  \left\langle \int d^4x \, \pi_L^A \pi_L^B(x) \pi_L^A \pi_L^B(0)
  \right\rangle_j^{\! \connect}
  &=\int \frac{d^4p}{(2\pi)^4}\,\frac1{(p^2+m_A^2)(p^2+m_B^2)}
  \sim\frac1{16\pi^2}\ln\left(\frac{\tilde\Lambda} j\right),
  \label{eq:log1}
\end{align}
where we have used \eqref{eq:GOR_high} and $\tilde\Lambda$ is the
momentum cutoff of the integral, in which we have also absorbed
$\fh_{A,B}$ and $\Phi_\H$.  Thus the dependence of the integral on $A$
and $B$ has disappeared, and \eqref{eq:Keff} becomes
\begin{align}
  K_{ab}(j)\sim
  Q_{ab}\frac{\Phi_\H^2}{16\pi^2}\ln\left(\frac{\tilde\Lambda} j\right)
\end{align}
with
\begin{align}
  \label{eq:Q}
  Q_{ab}=\sum_{AB} \frac{1}{4\fh_A^2 \fh_B^2}\tr(t^a \{T^A, T^B\})
  \tr(t^b \{T^A, T^B\})\,,
\end{align}
where we have included a factor of 2 for the left- and right-handed NG
modes in the loop.  To evaluate $Q_{ab}$ we consider three cases:
\begin{enumerate}
\item For $A=B=0$, the contribution to $Q_{ab}$ vanishes since
  $\tr(t^a)=0$,
  \begin{align}
    Q_{ab}^{(1)}=0\,.
  \end{align}
\item For $A=0,\ B\neq 0$ and $A\neq 0,\ B=0$, 
  the contribution to $Q_{ab}$ reads
  \begin{equation}
    \label{eq:Q2}
    Q_{ab}^{(2)}=\frac2{4\fh_0^2 \fh^2}\frac4{N_f}\sum_c
    \tr(t^aT^c)\tr(t^bT^c)
    =\frac{2\delta_{ab}}{N_f \fh_0^2 \fh^2}\,,
  \end{equation}
  where in the first equation the factor of 2 reflects the two
  possibilities above and in the second equation we have used
  $\tr(t^aT^c)=\delta_{ac}$.  
\item For $A\neq 0$ and $B\neq0$, the contribution to $Q$ can be
  obtained from \cite[eq.~(48)]{Toublan:1999hi} by replacing
  $2N_f \rightarrow N_f$ and multiplying by 8 to correct for the
  difference in the normalization of the generators (which in
  \cite{Toublan:1999hi} is $\tr(t^at^b)=\delta_{ab}/2$).  This yields
  \begin{equation}
    \label{eq:Q3}
    Q_{ab}^{(3)}=\delta_{ab}\frac{(N_f-4)(N_f+2)}{2N_f\fh^4}\,.
  \end{equation}
\end{enumerate}
Summing up these contributions, $Q_{ab}$ is given by
\begin{equation}
  Q_{ab}=Q_{ab}^{(1)}+Q_{ab}^{(2)}+Q_{ab}^{(3)}
  =\delta_{ab}\left[\frac{(N_f-4)(N_f+2)}{2N_f \fh^4}
    +\frac{2}{N_f \fh_0^2 \fh^2}\right]\,,
\end{equation}
and our final result for the scalar susceptibility from the effective
theory is
\begin{equation}
  \label{eq:Keff2}
  K_{ab}(j)\sim\delta_{ab}\left[\frac{(N_f-4)(N_f+2)}{2N_f \fh^4}
    +\frac{2}{N_f \fh_0^2 \fh^2} \right]
  \frac{\Phi_\H^2}{16\pi^2} \ln\left(\frac{\tilde\Lambda}{j}\right).
\end{equation}

Now let us compare \eqref{eq:KQCD} and \eqref{eq:Keff2}.  First of
all, note that the constant part $\rsv(0)$ in \eqref{eq:KQCD} does not
contribute to $K_{ab}$ since
\begin{align}
  \int_0^{\infty} d\xi\,\frac{\xi^2-j^2}{(\xi^2+j^2)^2}=0\,.
\end{align}
Hence only the difference $\rsv(\xi)-\rsv(0)$ is relevant.  To
reproduce the singularity $\sim \ln(\tilde\Lambda/j)$ in
\eqref{eq:Keff2}, we must have
\begin{equation}
  \rsv(\xi)-\rsv(0)=C\xi \quad\text{for}\quad \xi>0
\end{equation}
in the vicinity of $\xi=0$, where $C=\rsv'(0)$.  Now,
\begin{equation}
  \label{eq:log2}
  \int_0^{\tilde\Lambda} d\xi\,\frac{C\xi(\xi^2-j^2)}{(\xi^2+j^2)^2}
  \sim C\ln\left(\frac{\tilde\Lambda}{j} \right),
\end{equation}
and thus the slope of the singular value density at the origin is
given by
\begin{equation}
  \label{eq:slope_high}
  \rsv'(0)=
  \left[\frac{(N_f-4)(N_f+2)}{2N_f \fh^4}+\frac{2}{N_f \fh_0^2 \fh^2} \right]
  \frac{\Phi_\H^2}{16\pi^2}\,.
\end{equation}
Note that $\Phi_\H$, $\fh_0$, and $\fh$ are functions of $\mu$.  To
what extent this result is still valid at $\mu<\infty$ will be
discussed in section~\ref{sec:ss_conn}.

For $N_f=2$ the Smilga-Stern method used above does not work since in
\eqref{eq:J} we then have $J=jI$ so that $K_{ab}(j)$ cannot be defined
as in \eqref{eq:Kab}.  A similar problem occurs in the derivation of
the slope of the Dirac eigenvalue density, where the Smilga-Stern
method fails for $N_f=1$.  In that case the slope could still be
computed using partially quenched perturbation theory
\cite{Toublan:1999hi}, and it was found that the result obtained from
the Smilga-Stern method remains valid for $N_f=1$.  It is therefore
tempting to speculate that \eqref{eq:slope_high} remains valid for
$N_f=2$, but to confirm this we would have to compute $\rsv(\xi)$ in
partially quenched perturbation theory.  Such a rather complicated
calculation is deferred to future work.

Let us add two comments here.  First, we could relax the assumption
$J_R=-J_L$, and in particular we could set $J_R$ (or $J_L$) to zero.%
\footnote{This results in a projection on the topologically trivial
  sector (see \eqref{eq:ZJ}), which is immaterial in the $p$-regime.}
This would give us the slope of the density of the left-handed (or
right-handed) singular values, which is $1/2$ of the full slope
because the factor of 2 mentioned after \eqref{eq:Q} would be absent.
Second, we have now computed $\rsv(0)$ and $\rsv'(0)$ and therefore
obtained information on the singular value density near zero.  An
analytical result can also be computed for asymptotically large $\xi$,
which, owing to asymptotic freedom, can be described by the free
theory without coupling to the gauge field. Thus the whole singular
value spectrum can be understood at least qualitatively by an
interpolation of two tractable limits.  We obtain (for an arbitrary
number $N_c\ge2$ of colors)
\begin{align}
  \label{eq:rho_pert}
  \rsv(\xi)\to \frac{N_c}{2\pi^2}\,\xi(\xi^2+2\mu^2)
  \quad\text{for}\quad \xi\to\infty\,.
\end{align}
For $\mu\to 0$ this reduces to twice the Dirac eigenvalue density in
the free limit, as expected. An outline of the derivation is given in
appendix \ref{app:rho_pert}.

\subsection{Intermediate density}
\label{sec:ss_int}

The calculation at intermediate density is very similar to that at
infinite density.  Since the fundamental microscopic theory is
unchanged, equations \eqref{eq:J} through \eqref{eq:KQCD} also remain
unchanged.  On the effective theory side, \eqref{eq:Zeff_high} is
replaced by
\begin{align}
  Z_\text{eff}^\I=\int\dd\Sigma_L\,\dd\Sigma_R\,\dd V\,
  \exp\left(-\int d^4x\,\mL_\text{eff}^\I\right).
\end{align}
We could now go through a similar calculation as in
section~\ref{sec:ss_high}, using \eqref{eq:int} through
\eqref{eq:SigmaV}.  However, it is easier to note that the only
difference is the replacement of the two $\U(1)$ fields $L$ and $R$ by
a single $\U(1)$ field $V$.  The only contributions of the $\U(1)$
fields to $Q_{ab}$ are in $Q_{ab}^{(2)}$, and it follows from the
calculation in section~\ref{sec:ss_high} that we can obtain
$Q_{ab}^{(2)}$ for the present case by dividing the result in
\eqref{eq:Q2} by 2.  Everything else remains unchanged so that the
slope for $N_f\ge4$ is now given by
\begin{align}
  \label{eq:slope_int}
  \rsv'(0)=
  \left[\frac{(N_f-4)(N_f+2)}{2N_f f^4}+\frac{1}{N_f f_0^2 f^2} \right]
  \frac{\Phi_\I^2}{16\pi^2}\,.  
\end{align}
Again, $\Phi_\I$, $f_0$, and $f$ are functions of $\mu$.  It is
tempting to speculate that \eqref{eq:slope_int} remains valid for
$N_f=2$.

\subsection{Zero density}
\label{sec:ss_zero}

Let us now consider strictly zero chemical potential.  Again,
equations \eqref{eq:J} through \eqref{eq:KQCD} remain unchanged,
but the coset space of the effective theory is now
$\SU(2N_f)/\Sp(2N_f)$.  It follows from the calculation in
section~\ref{sec:ss_high} that $Q_{ab}$ is now entirely given by
$Q_{ab}^{(3)}$ and that we can obtain $Q_{ab}^{(3)}$ for the present
case from the result in \eqref{eq:Q3} by replacing $N_f\to2N_f$, $\fh\to
F$, and dividing by 2 since the left- and right-handed modes are
already contained in $\SU(2N_f)/\Sp(2N_f)$.  This yields for $N_f\ge2$
\begin{align}
  \label{eq:slope_zero}
  \rsv'(0)=\frac{(N_f-2)(N_f+1)}{N_f F^4}\frac{\Phi_\L^2}{16\pi^2}\,,
\end{align}
where $\Phi_\L$ and $F$ are now independent of $\mu$.  To what extent
this result is still valid at nonzero $\mu$ will be discussed in the
next subsection.

\subsection{Relation between the three results}
\label{sec:ss_conn}

In the previous three subsections we have obtained three different
results for the slope $\rsv'(0)$ for $\mu=\infty$, intermediate $\mu$,
and $\mu=0$, respectively. At first glance it does not seem possible
to interpolate them smoothly, and we thus encounter a puzzle: How are
the three results related?  What actually happens to the spectrum if
$\mu$ is continuously changed?  Below we argue that there is no puzzle
here. To simplify the presentation we divide our discussion into three
parts, the first one for three-color QCD at $\mu=0$ as an instructive
model case for our problem, the second one for the low-density region,
and the third one for the high-density region of two-color QCD. The
arguments are analogous, however.

\subsubsection[Zero density $(\beta=2)$: A journey from $N_f=2$ to $3$]
{\boldmath Zero density $(\beta=2)$: A journey from $N_f=2$ to $3$}
\label{sc:zero_beta=2_test}

So far we have argued that there are three effective theories for
dense two-color QCD and that as a function of $\mu$ they change
smoothly from one to another. This situation has an exact counterpart
in three-color QCD for $N_f=3$ at $\mu=0$, where the strange quark
mass serves as a ``knob'' to interpolate between the chiral
perturbation theories for $N_f=2$ and $N_f=3$.  Therefore we study
this simpler case first before considering the more exotic case of
dense two-color QCD.

The original Smilga-Stern relation \cite{Smilga:1993in}, derived for
the Dirac eigenvalue density (not the singular value density) in
three-color QCD $(\beta=2)$ at $\mu=0$ with $N_f$ flavors in the
chiral limit, reads
\begin{align}
  \label{eq:ss_orig}
  \rho(\lambda) = \frac{\Sigma}{\pi} + 
  \frac{\Sigma^2}{32\pi^2F^4}\frac{N_f^2-4}{N_f}|\lambda| + o(\lambda)
\end{align}
with the chiral condensate $\Sigma$ and the pion decay constant $F$.
The coefficient of $|\lambda|$ depends on $N_f$.  For example, the
slope vanishes for $N_f=2$ but is nonzero for $N_f=3$. It is not clear
from this expression alone how the slope changes if we add a nonzero
strange quark mass to change the number of light flavors continuously.

The generalization of \eqref{eq:ss_orig} to nonzero degenerate masses
was given in \cite{Osborn:1998qb}%
\footnote{There is a typo in \cite[Eq.~(84)]{Osborn:1998qb}. The
  second term in the square brackets of that equation must be
  multiplied by $\pi$, as is evident from
  \cite[Eq.~(83)]{Osborn:1998qb}.  } and later extended to
nondegenerate sea quark masses in \cite{Zyablyuk:1999aj}.  Therefore
we can employ the results of \cite{Zyablyuk:1999aj} to study the
density $\rho(\lambda)$ and the slope $\rho'(\lambda)$ for $N_f=3$
with a massive strange quark. Setting $(m_1,m_2,m_3)=(0,0,m)$ in
\cite[Eq.~(17)]{Zyablyuk:1999aj} we plot $\rho(\lambda)$ and
$\rho'(\lambda)$ in figure~\ref{fig:slope_crsov} as a function of
$\lambda/m$.

\begin{figure}[t]
  \centering
  \includegraphics[height=51mm]{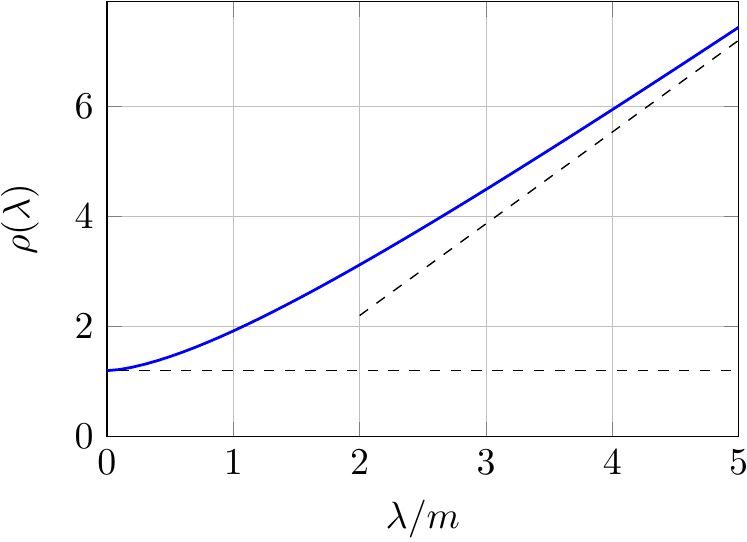}\hfill
  \includegraphics[height=51mm]{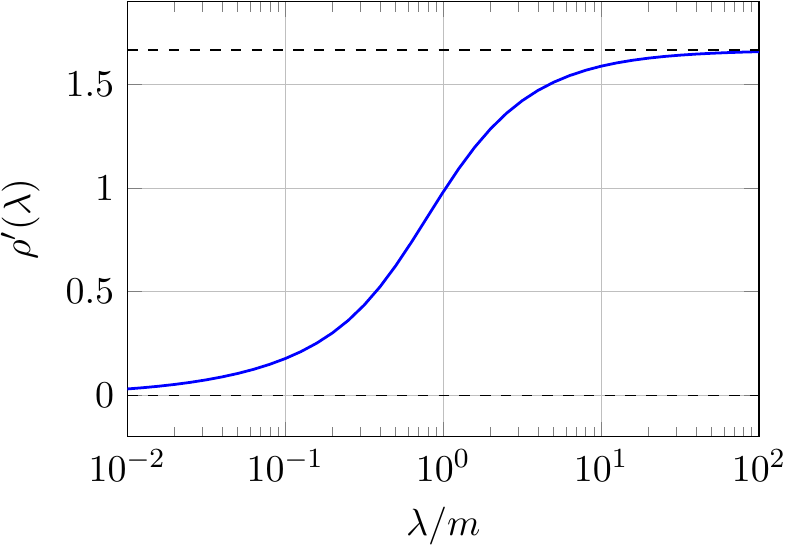}
  \caption{Left: Dirac eigenvalue density in three-color QCD at
    $\mu=0$ for two massless flavors and one flavor with mass $m$, in
    units of $m\Sigma^2/32\pi^2F^4$.  The intercept with the vertical
    axis is arbitrary due to renormalization \cite{Zyablyuk:1999aj}.
    Right: Slope of the density in units of ${\Sigma^2}/{32\pi^2F^4}$.
    The dotted lines in both plots correspond to the slopes $0$
    ($N_f=2$) and ${5}/{3}$ ($N_f=3$).}
  \label{fig:slope_crsov}
\end{figure}

The curves nicely interpolate between two limits: For $\lambda\ll m$
the strange quark is heavy relative to the probed scale and we get the
slope for $N_f=2$: $\rho'(\lambda)\propto(N_f^2-4)/{N_f}=0$.  In the
limit $m\ll\lambda$ $(\ll \Lambda\sim 4\pi F)$ the strange quark is
light relative to the probed scale and we get the slope for $N_f=3$:
$\rho'(\lambda)\propto(N_f^2-4)/{N_f}={5}/{3}$.  The transition occurs
smoothly around $\lambda \sim m$.  Thus we can draw the conclusion
that no contradiction arises from different values of the slopes for
$N_f=2$ and $N_f=3$, because at nonzero $m$ they correspond to
different domains of the spectrum.

This finding can be interpreted within partially quenched chiral
perturbation theory as follows. In this method we add valence flavors
and compute the spectral density from the valence quark mass
dependence of the chiral condensate \cite{Verbaarschot:2004gj}.  We
therefore deal with two classes of mesons, one being made of only sea
quarks, and the other being made of valence quarks (and sea
quarks). If the latter (``valence mesons'') are much heavier than the
former (``sea mesons''), i.e., $\lambda\gg m$, then all three sea
flavors contribute to the valence quark mass dependence of the chiral
condensate, implying $N_f=3$. Conversely, if $\lambda\ll m$, the
heavier sea mesons are decoupled, reducing the computation to $N_f=2$.
The transition between these two cases occurs around $\lambda\sim m$,
i.e., when the masses of the sea and valence mesons are roughly equal.
  
\subsubsection{Low density}
\label{sec:match_low}

In this subsection we discuss the relation between the results
\eqref{eq:slope_int} and \eqref{eq:slope_zero} for $\mu\ll\Lambda$.
The chemical potential plays exactly the same role as the strange
quark mass in the previous subsection, both acting as explicit
symmetry-breaking parameters.  This analogy is the basis of our
following argument.

We first note that $\rsv(\xi)$ can be
computed in partially quenched chiral perturbation theory, starting from
\begin{align}
  \label{eq:pq}
  Z_{N_f+2|2}(j;j_v,j_v')=\bigg\langle{\det}^{N_f/2}(D^\dagger D+j^2)
    \frac{\det(D^\dagger D+j_v^2)}{\det(D^\dagger D+j_v'^2)}
    \bigg\rangle_\text{YM}
\end{align}
and setting $j_v=j_v'=i\xi+\eps$ (with $\eps\to0^+$) at the end of the
calculation.  We will not actually perform this computation but use
\eqref{eq:pq} for qualitative estimates.  In comparison to the usual
setting \cite{Osborn:1998qb,Toublan:1999hi}, $j$ and $j_v,j_v'\sim\xi$
correspond to the sea and valence quark masses, respectively. In the
following $j$ is always assumed to be infinitesimal. Our main concern
is the competition between $\mu$ and $j_v$.

We can now replace $Z_{N_f+2|2}(j;j_v,j_v')$ by an effective partition
function formulated in terms of NG modes. We have two options, either
the partially quenched extension of theory \L, or that of theory
\I. Let us discuss them separately.

\begin{itemize}
\item Theory \I: For the partially quenched extension of theory \I\ to
  be valid, the condition (i) of section~\ref{sec:val} must be
  satisfied for all NG modes, i.e., their masses must be much smaller
  than $m_\ell$. The new ingredient here is that, in addition to the
  NG modes discussed in section~\ref{sec:eff_int}, we now also have NG
  modes containing the valence quarks corresponding to $j_v$ and
  $j_v'$, which we call valence NG modes.  At low density their masses
  are of order $\sqrt{j_v\Lambda}\sim \sqrt{\xi\Lambda}$, see
  \eqref{eq:GOR_low_a} with $F\sim\Phi_\L^{1/3}\sim\Lambda$, and hence
  the condition is
  \begin{align}
    \xi \ll \frac{\mu^2}{\Lambda}\,.
  \end{align}
  In this domain the slope $\rsv'(\xi)$ is given by
  \eqref{eq:slope_int} at leading order of the low-energy
  expansion.
\item Theory \L: For the partially quenched extension of theory \L\ to
  be valid, the masses of all NG modes must again be much smaller than
  $m_{\ell}$. This time the masses of the valence NG modes are of
  order $\sqrt{j_v\Lambda}\sim \sqrt{\xi\Lambda}$ and
  $\sqrt{j_v\Lambda+\mu^2}\sim\sqrt{\xi\Lambda+\mu^2}$, see
  \eqref{eq:GOR_low} with $F\sim\Phi_\L^{1/3}\sim\Lambda$. The
  condition is therefore
  \begin{align}
    \xi \ll \Lambda\,.
  \end{align}
  This is not end of the story, however: Theory \L\ is more
  complicated than theory \I, because we have two scales $\mu$ and
  $\sqrt{j_v\Lambda}$ (analogous to $m$ and $\lambda$ in section
  \ref{sc:zero_beta=2_test}) whose ratio controls the final result in
  a nontrivial way.  Based on our experience in section
  \ref{sc:zero_beta=2_test} we expect the following.
  \begin{itemize}
  \item For $\mu\ll\sqrt{j_v\Lambda}\sim\sqrt{\xi\Lambda}$ : All sea
    NG modes contribute, and the result for $\rsv'(\xi)$ agrees with
    the result \eqref{eq:slope_zero} at $\mu=0$
  \item For $\mu\gg\sqrt{j_v\Lambda}\sim\sqrt{\xi\Lambda}$ : The NG
    modes with masses of order $\mu$ decouple from the computation
    of $\rsv(\xi)$, and theory \L\ reduces to theory \I\ in which the
    heavy modes have been integrated out. The slope thus agrees with
    \eqref{eq:slope_int} from theory \I.
  \end{itemize}
\end{itemize}

Putting everything together, we see that in the regime $\mu\ll\Lambda$
the results from theory \I\ and \L\ for the slope $\rsv'(\xi)$ are
valid in the following domains:
\begin{alignat}{2}
  & \text{\eqref{eq:slope_int} from \I}: &
  &\xi\ll\frac{\mu^2}\Lambda\,,\\
  & \text{\eqref{eq:slope_zero} from \L}: \qquad &
  \frac{\mu^2}\Lambda&\ll\xi\ll\Lambda\,.
\end{alignat}
These findings are illustrated in figure~\ref{fig:slope} (left).  The
slope is first given by \eqref{eq:slope_int}, and we conjecture that
it changes smoothly to the value given by \eqref{eq:slope_zero}.  To
avoid confusion we point out that for $\mu=0$ the window in which the
slope is given by \eqref{eq:slope_int} shrinks to zero so that the
slope at the origin is given by \eqref{eq:slope_zero}. 
\begin{figure}[t]
  \centering
  \includegraphics{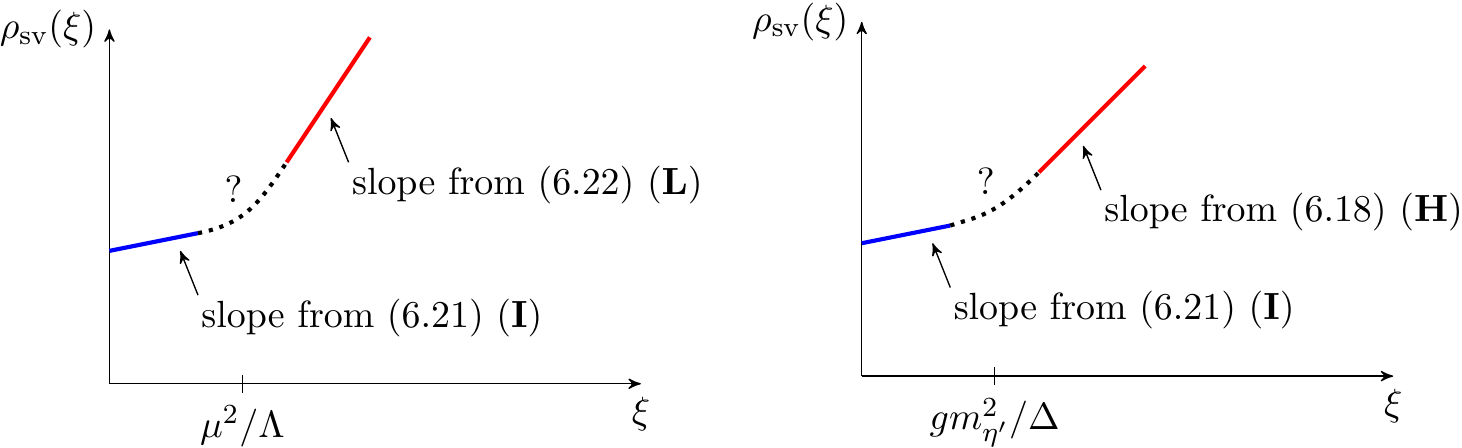}
  \caption{The behavior of the singular value density near zero.
    Left: low density ($\mu\ll\Lambda$), right: high density (so that
    $m_{\eta'}\ll\Delta$), see text for details.}
  \label{fig:slope}
\end{figure}

We now present a nontrivial cross-check of our conclusion that does
not hinge on partially quenched chiral perturbation theory.  Let us
return to our discussions in sections
\ref{sec:ss_high}--\ref{sec:ss_zero} without valence quarks and
consider taking the limits $\mu\to 0$ and $j\to 0$ while keeping the
condition $\mu^2\gg j\Lambda$.  In this case the two types of NG modes
have masses of order $\sqrt{j\Lambda}$ (lighter) and $\mu$ (heavier),
respectively.  The one-loop integral \eqref{eq:log1} in the effective
theory can be evaluated for the cases of two lighter, two heavier, or
one lighter and one heavier NG modes circulating around the loop.  For
these cases we obtain
\begin{align}
  \int^\Lambda d^4p\,\frac1{(p^2+j\Lambda)^2}
  &\sim\ln\frac\Lambda j\,,\\
  \int^\Lambda d^4p\,\frac1{(p^2+\mu^2)^2}
  &\sim\ln\frac{\Lambda^2}{\mu^2}\,,\\
  \int^\Lambda d^4p\,\frac1{(p^2+j\Lambda)(p^2+\mu^2)}
  &\sim\frac{\mu^2\ln\frac{\Lambda^2}{\mu^2}
  -j\Lambda\ln\frac\Lambda j}{\mu^2-j\Lambda}
  \sim\ln\frac{\Lambda^2}{\mu^2}\quad\text{for}\quad\mu^2\gg j\Lambda\,.
\end{align}
Therefore the sum of all one-loop contributions to $K_{ab}$ is given
by
\begin{align}
  \label{eq:oneloop}
  \alpha\ln \frac\Lambda j+\beta\ln\frac{\Lambda^2}{\mu^2}
\end{align}
with prefactors $\alpha$ and $\beta$ that also contain traces of the
generators.  The two infrared singularities in \eqref{eq:oneloop}
(generated by $j\to 0$ and $\mu\to 0$, respectively) must be matched
by corresponding singularities in the microscopic theory.  Motivated
by figure~\ref{fig:slope}, let us assume that $\rsv(\xi)$ can be
approximated by a straight line with slope $\alpha$ for $\xi<\xi_c$ and by
another straight line with slope $\alpha+\beta$ for $\xi>\xi_c$, with $\xi_c$
an unknown function of $\mu$.  Then \eqref{eq:KQCD} becomes
\begin{align}
  \int_0^\Lambda d\xi&\,\rsv(\xi)\,\frac{\xi^2-j^2}{(\xi^2+j^2)^2}
  \notag\\
  &= \int_0^{\xi_c} d\xi\,(\alpha\xi+\text{const.})\,
  \frac{\xi^2-j^2}{(\xi^2+j^2)^2}
  +\int_{\xi_c}^\Lambda d\xi\,\big[(\alpha+\beta)\xi+\text{const.}\big]\,
  \frac{\xi^2-j^2}{(\xi^2+j^2)^2} \notag\\
  &\sim  \alpha\ln\frac{\xi_c}j+(\alpha+\beta)\ln\frac\Lambda{\xi_c}
  =\alpha\ln\frac\Lambda j+\beta\ln\frac\Lambda{\xi_c}\,.
  \label{eq:sing2}
\end{align}
Matching \eqref{eq:oneloop} and \eqref{eq:sing2} yields 
\begin{align}
  \label{eq:bound_low}
  \xi_c\sim\frac{\mu^2}\Lambda\,,
\end{align}
in agreement with the argument based on partially quenched chiral
perturbation theory.

\subsubsection{High density}
\label{sc:large_mu_relat}

We now clarify the relation between \eqref{eq:slope_high} from theory
\H\ and \eqref{eq:slope_int} from theory \I\ at large $\mu$.  The
arguments are analogous to those at small $\mu$, except that
$m_\text{inst}$ now plays the role of $\mu$ as an external
symmetry-breaking parameter. First of all we require that the masses
of all NG modes (sea and valence) must be sufficiently below $m_\ell$.
At high density the masses of the valence NG modes in the partially
quenched theory are of order
$\sqrt{j_v\Delta/g}\sim\sqrt{\xi\Delta/g}$ and
$\sqrt{j_v\Delta/g+m_\text{inst}^2}\sim\sqrt{\xi\Delta/g+m_\text{inst}^2}$,
see \eqref{eq:GOR_high} with $\fh_A\sim\mu$ \cite{Son:1999cm} and
$\Phi_\H\sim\mu^2\Delta/g$ \cite{Schafer:1999fe}, where $g$ is the
running coupling constant.  Using the relevant results for $m_\ell$ in
table~\ref{tab:val} we obtain the following bounds on the values of
$\xi$ below which $\rsv(\xi)$ can be computed from the partially
quenched extensions of the effective theories \I\ or \H\ in the regime
where $m_{\eta'}\ll\Delta$:
\begin{alignat}{2}
  \label{eq:xiI_high}
  \xi&\ll\frac{gm_{\eta'}^2}\Delta \qquad && \text{for theory \I}\,,\\
  \label{eq:xiH}
  \xi&\ll g\Delta && \text{for theory \H}\,.
\end{alignat}
Thus $\rsv'(\xi)$ is given by \eqref{eq:slope_int} from theory \I\ in
the range \eqref{eq:xiI_high}. We note that the scale
$gm_{\eta'}^2/\Delta$ goes to zero rapidly as $\mu\to \infty$.

On the other hand, theory \H\ has two scales, $\sqrt{j_v\Delta/g}$ and
$m_\text{inst}$.  A rerun of the arguments at small $\mu$ then shows
that in the regime where $m_{\eta'}\ll\Delta$ the slope $\rsv'(\xi)$
is given by the results \eqref{eq:slope_high} and \eqref{eq:slope_int}
in the following domains,
\begin{alignat}{2}
  & \text{\eqref{eq:slope_int} from \I}: &
  \xi&\ll\frac{gm_{\eta'}^2}\Delta\,,\\
  & \text{\eqref{eq:slope_high} from \H}: \qquad &
  \frac{gm_{\eta'}^2}\Delta&\ll\xi\ll g\Delta\,.
\end{alignat}
This is illustrated in figure~\ref{fig:slope} (right).  For
$\mu\to\infty$ the window in which the slope is given by
\eqref{eq:slope_int} shrinks to zero so that the slope at the origin
is given by \eqref{eq:slope_high}.

The second argument presented in section~\ref{sec:match_low} works in
exactly the same way here and leads to
\begin{align}
  \label{eq:bound_high}
  \xi_c\sim\frac{gm_{\eta'}^2}\Delta
\end{align}
as expected.

Note that the result \eqref{eq:slope_int} is actually valid for all
$0<\mu<\infty$.  We can replace $\rsv'(0)$ by $\rsv'(\xi)$ in
\eqref{eq:slope_int} for sufficiently small $\xi$, the upper bound of
which is given in \eqref{eq:bound_low} and \eqref{eq:bound_high} at
small and large $\mu$, respectively. At intermediate density we do not
have an estimate for the upper bound because $m_\ell$ is unknown in
this region.

So far we have explained what we believe is the most reasonable
behavior of the singular value density at nonzero $\mu$ based on the
partial quenching technique and the analogy to three-color QCD at
$\mu=0$ with a heavy strange quark. For a solid proof of our
conjecture shown in figure \ref{fig:slope} one would have to compute
the slope in partially quenched chiral perturbation theory explicitly,
but this is beyond the scope of this paper. It would be interesting to
check our new Smilga-Stern-type relations by lattice simulations. This
is possible in principle as the infamous sign problem is absent in
this theory.

\section{\boldmath Finite-volume analysis: Leutwyler-Smilga-type sum
  rules and random matrix theories $(\beta=1)$}
\label{sec:ls}

\subsection[The $\eps$-regime]{\boldmath The $\eps$-regime}
\label{sec:eps}

In this section we study two-color QCD with diquark sources in a
finite volume $V_4=L^4$ (and again in the chiral limit).  As in QCD
there is a regime, the so-called $\eps$-regime \cite{Gasser:1987ah},
in which the kinetic terms in the effective chiral Lagrangian can be
neglected so that the theory becomes zero-dimensional and the
partition function is dominated by the zero-momentum modes of the NG
particles.  The condition for the $\eps$-regime is
\begin{align}
  \label{eq:eps}
  \frac1{m_\ell}\ll L\ll\frac1{m_\text{NG}}\,,
\end{align}
where $m_\ell$ is again the mass scale of the lightest non-NG particle
and $m_\text{NG}$ is the mass scale of the NG particles that are
included in the effective theory.  The first inequality in
\eqref{eq:eps} means that the contribution of the non-NG particles to
the partition function can be neglected, while the second inequality
means that the Compton wavelength of the NG particles is larger than
the size of the box, which in turn implies that the functional
integral over the NG fields can be replaced by a zero-mode integral
over the coset space parametrized by them.  For the three different
effective theories in section~\ref{sec:eff} the values of $m_\ell$ are
given in section~\ref{sec:val}, and the $m_\text{NG}$ are given in
\eqref{eq:GOR_low}, \eqref{eq:GOR_high}, and \eqref{eq:GOR_int}.  Note
that for zero diquark sources and finite $L$ the second inequality in
\eqref{eq:eps} is always satisfied in theory \I\ since
$m_\text{NG}=0$.  Note also that while the domains of validity of the
effective theories overlap (see section~\ref{sec:val}), this is not
the case for the corresponding $\eps$-regimes: At low density it
follows from \eqref{eq:eps} that $1/\mu\ll L$ for \I\ and $L\ll1/\mu$
for \L, and these conditions are mutually exclusive.  Similarly, at
high density we have $1/m_{\eta'}\ll L$ for \I\ and $L\ll1/m_{\eta'}$
for \H.

At $\mu=0$, it is well known that there exists a scale $E_T$ below
which the eigenvalue spectrum of the Dirac operator obeys chiral
random matrix theory \cite{Osborn:1998qb,Damgaard:1998xy}.  The scale
$E_T$ is called Thouless energy, borrowing the nomenclature from
mesoscopic physics.  The equivalence between random matrix theory and
the zero-momentum limit of the partially quenched effective theory
shows that $E_T$ can be understood as the energy above which the
condition \eqref{eq:eps} no longer holds in the partially quenched
theory and the modes with nonzero momentum start to contribute.

We now comment on the Thouless energy for the singular value spectrum
at $0<\mu<\infty$, based on the partially quenched extension of the
three effective theories introduced in section~\ref{sec:match_low}.
For simplicity we let $j=0$ in \eqref{eq:pq} and concentrate on the
``spectral mass'' $j_v$.  We assume that the condition $L\gg 1/m_\ell$
is satisfied for all cases considered below.

\begin{itemize}
\item {Theory \I} $(0<\mu<\infty)$: From \eqref{eq:GOR_int} we find
  that the masses of the valence NG modes are given by
  $m_\text{vNG}^2\sim j_v\Phi_\I/f^2\sim \xi\Phi_\I/f^2$, where
  all low-energy constants depend on $\mu$ implicitly.%
  \footnote{We are sloppy about the distinction between $f$ and $f_0$
    here, but they are of the same order of magnitude so that the
    distinction does not change our discussion.}  For $m_\text{vNG}\ll
  1/L$ the $j_v$-dependence of the partition function is governed by
  the NG modes with zero momentum.  Thus the Thouless energy is
  determined by
  \begin{align}
    \sqrt{E_T\Phi_\I/f^2} = \frac{1}{L}
    \qquad \to \qquad 
    E_T = \frac{f^2}{L^2\Phi_\I} \,.
  \end{align}
\item {Theory \H} $(m_{\eta'}\ll \Delta)$: From \eqref{eq:GOR_high} we
  have $m_\text{vNG}^2 \sim j_v\Phi_\H/\fh^2\sim j_v\Delta/g$ (see
  section \ref{sc:large_mu_relat}) and
  $j_v\Phi_\H/\fh^2+m_\text{inst}^2 \sim j_v\Delta/g+m_\text{inst}^2$.
  Assuming $m_\text{inst}\ll 1/L$ we obtain
  \begin{align}
    \sqrt{E_T\Delta/g} = \frac{1}{L}
    \qquad \to \qquad E_T = \frac{g}{L^2\Delta}\,.
  \end{align}
\item {Theory \L} $(\mu\ll\Lambda)$: From \eqref{eq:GOR_low} we have
  $m_\text{vNG}^2\sim j_v\Phi_\L/F^2$ and
  $j_v\Phi_\L/F^2+(2\mu)^2$, where this time the low-energy
  constants are $\mu$-independent.  Assuming $\mu\ll 1/L$ we obtain
  \begin{align}
    \sqrt{E_T\Phi_\L/F^2} = \frac{1}{L}
    \qquad \to \qquad E_T = \frac{F^2}{L^2\Phi_\L}\,.
  \end{align}
\end{itemize}

Thus in all cases $E_T\propto 1/\sqrt{V_4}$. Essentially, this is due
to the fact that the diquark source enters the effective Lagrangian
linearly at any $\mu$.  This completes our discussion of the Thouless
energy.

In the $\eps$-regime we can compute exact sum rules for the inverse
singular values of the Dirac operator.  This will be done in the
following three subsections for the three different density regimes.
For technical reasons we start with intermediate density this time.
We will also derive chiral random matrix theories that allow us to
compute microscopic correlation functions of the singular
values.\footnote{In addition, a chiral random matrix theory for QCD
  with isospin chemical potential ($\beta = 2$) will be derived in
  appendix~\ref{app:isospin}.} The sum rules are simply moments of
these correlation functions. The results for $\rsv(\xi)$ from the
random matrix theories are valid in the range $\xi\ll E_T$.

\subsection{Intermediate density}
\label{sec:ls_int}

In analogy to the analysis of Leutwyler and Smilga
\cite{Leutwyler:1992yt} we first project the partition function onto
sectors of fixed topological charge and then expand it in powers of
the diquark sources.  This is done both in the microscopic theory and
in the effective theory.  Matching the coefficients of the sources
then yields sum rules for the inverse singular values.

On the QCD side, we start from \eqref{eq:ZJ} and assume that there are
no accidental zero modes.  Without loss of generality we therefore set
$n_L=0$ and $n_R=\nu\ge0$.\footnote{As usual, for $\nu<0$ the final
  results of this section remain valid, except that $\nu$ must be
  replaced by $|\nu|$.}  We also introduce the notation
\begin{align}
  \ev{O}_\nu=\frac
    {\ev{O\,{\det}'(D^\dagger D)^{N_f/2}}_{\text{YM},\nu}}
    {\ev{{\det}'(D^\dagger D)^{N_f/2}}_{\text{YM},\nu}}\,,
\end{align}
where the subscript $\nu$ indicates that the average is only over
gauge fields with fixed topological charge $\nu$ and the prime, as
always, means that the zero modes are omitted.  Using
\begin{align}
  {\det}'\big(
    D^\dagger D+J_L^\dagger J_LP_L+J_R^\dagger J_RP_R
  \big)
  ={\prod_n}'\det(\xi_{Ln}^2+J_L^\dagger J_L)(\xi_{Rn}^2+J_R^\dagger J_R)
\end{align}
the partition function for fixed topology is given by
\begin{align} 
  \label{eq:Znu_JLJR}
  \frac{Z_\nu(J_L,J_R)}{\ev{{\det}'(D^\dagger D)^{N_f/2}}_{\text{YM},\nu}}
  =\big[\Pf(J_L^\dagger)\Pf(J_R)\big]^\nu \left\langle{\prod_n}'
  {\det}^{1/2}\left(\1+\frac{J_L^\dagger J_L}{\xi_{Ln}^2}\right)
  \left(\1+\frac{J_R^\dagger J_R}{\xi_{Rn}^2}\right) \right\rangle_\nu\,.
\end{align}
Using the formula
\begin{align}
  \det(\1+\eps)=1+\tr\eps+\frac{1}{2} \big[(\tr
  \eps)^2-\tr(\eps^2)\big] + O(\eps^3) 
\end{align}
we expand $Z_\nu$ in powers of the diquark sources,
\begin{align}
  &\frac{Z_\nu(J_L,J_R)}{\ev{{\det}'(D^\dagger D)^{N_f/2}}_{\text{YM},\nu}
    \big[\Pf(J_L^\dagger)\Pf(J_R)\big]^\nu}
  \notag
  \\
  & = 1 
  +\bigg[   \frac{\tr J_L^\dagger J_L}{2}
  \bigg\langle{{\sum_n}'}\frac{1}{\xi_{Ln}^2}\bigg\rangle_\nu 
  -\frac{\tr (J_L^\dagger J_L)^2}{4}
  \bigg\langle{\sum_{n}}'\frac{1}{\xi_{Ln}^4}\bigg\rangle_\nu 
  +\frac{(\tr J_L^\dagger J_L)^2}{8}
  \bigg\langle\bigg({\sum_{n}}'\frac{1}{\xi_{Ln}^2}\bigg)^2\bigg\rangle_\nu
  \notag
  \\
  & \qquad\qquad  +(L\leftrightarrow R)\bigg]  + 
  \frac{\tr J_L^\dagger J_L\tr J_R^\dagger J_R}{4}
  \bigg\langle\bigg({\sum_{n}}'\frac{1}{\xi_{Ln}^2}\bigg)
  \bigg({\sum_{m}}'\frac{1}{\xi_{Rm}^2}\bigg)\bigg\rangle_\nu
  +O(J^6)\,.
  \label{eq:Z_micro_expanded}
\end{align}

On the low-energy effective theory side, let us first consider the
case $N_f\ge4$.  We start from \eqref{eq:int} and neglect the kinetic
terms to obtain the finite-volume partition function
\begin{align}
  \label{eq:ZeffI}
  Z^\text{eff}(J_L,J_R)=\int d\Sigma_L\, d\Sigma_R\, dV\,
  \exp\big[V_4\Phi_\I\re\tr(J_L\Sigma_L-J_R\Sigma_R)V\big]\,,
\end{align}
where the integrals over $\Sigma_{L,R}$ and $V$ are no longer
functional integrals but simple integrals over $\SU(N_f)/\Sp(N_f)$ and
$\U(1)$, respectively.  It is convenient to define
\begin{align}
  \label{eq:Jtilde}
  \tilde J_i=J_iV_4\Phi_\I\qquad (i=L,R)\,.
\end{align}
As explained after \eqref{eq:Z_theta}, a nonzero $\theta$-angle can be
introduced by redefining, e.g., $J_L\to J_L\ee^{-i\theta/N_f}$ and
$J_R\to J_R\ee^{i\theta/N_f}$.  To project onto topological sectors we
note that the inversion of \eqref{eq:Z_theta} is
\begin{align}
  \label{eq:Z_nu}
  Z_\nu=\frac1{2\pi}\int_0^{2\pi}d\theta\,\ee^{-i\nu\theta}Z(\theta)\,,
\end{align}
which in the present case gives
\begin{align}
  Z_\nu^\text{eff}(J_L,J_R)&=\int d\Sigma_L\,d\Sigma_R\,dV\,
  \frac{d\theta}{2\pi}\,
  \exp\big[-i\nu\theta+\re\tr(\tilde J_L\ee^{-i\theta/N_f}\Sigma_L
  -\tilde J_R\ee^{i\theta/N_f}\Sigma_R)V\big]\notag\\
  &=\int d\Sigma_L\,d\Sigma_R\,dV\,\frac{d\theta}{2\pi}\,
  \exp\big[-iN_f\nu\theta+\re\tr(\tilde J_L\ee^{-i\theta}\Sigma_L
  -\tilde J_R\ee^{i\theta}\Sigma_R)V\big]\notag\\
  &=\int d\Sigma_L\,d\Sigma_R\,dL\,dR\,~
  (LR^\dagger)^{\frac12N_f\nu}\,
  \exp\big[\re\tr(\tilde J_LL\Sigma_L-\tilde J_RR\Sigma_R)\big]\,,
  \label{eq:Zint}
\end{align}
where in the second step we have redefined $\theta\to N_f\theta$ and
used the periodicity of the $\theta$-integral to extend the
integration region from $[0,2\pi/N_f]$ back to $[0,2\pi]$.  In the
last step we have introduced the new $\U(1)$ integration variables
\begin{align}
  L=\ee^{-i\theta}V\quad\text{and}\quad R=\ee^{i\theta}V\,.
\end{align}
Parametrizing $\Sigma_i=U_iIU_i^T$ with $U_i\in\SU(N_f)/\Sp(N_f)$ as
in \eqref{eq:param_high} we note that
\begin{align}
  \re\tr(\tilde J_RR\Sigma_R)
  =\re\tr(\tilde J_R^\dagger R^\dagger\Sigma_R^\dagger)
  =\re\tr(-\tilde J_R^\dagger R^\dagger U_R^*IU_R^\dagger)\,.
\end{align}
Redefining the integration variables $R\to R^\dagger$ and $U_R\to
U_R^*$ and combining $L$ and $R$ with $\Sigma_L$ and $\Sigma_R$ as in
\eqref{eq:Sigma_tilde} we thus obtain from \eqref{eq:Zint}
\begin{align}
  Z_\nu^\text{eff}(J_L,J_R)&=\int d\Sigma_L\,d\Sigma_R\,dL\,dR\,
  (LR)^{\frac12N_f\nu}\,
  \exp\big[\re\tr(\tilde J_LL\Sigma_L+\tilde J_R^\dagger R\Sigma_R)\big]
  \notag\\
  &=\int d\tilde\Sigma_L\,d\tilde\Sigma_R\,
  {\det}^{\nu/2}(\tilde \Sigma_L\tilde\Sigma_R)
  \exp\big[\re\tr(\tilde J_L\tilde\Sigma_L
  +\tilde J_R^\dagger\tilde\Sigma_R)\big]\,.
  \label{eq:Ztmp}
\end{align}
Recalling $\tilde\Sigma_i=U_iIU_i^T$ with $U_i\in\U(N_f)/\Sp(N_f)$ we
note that we can extend the integration to be over $U_i\in\U(N_f)$
because the generators of $\Sp(N_f)$ leave $I$ invariant and therefore
drop out of the combination $U_iIU_i^T$.  We now define, for a complex
antisymmetric matrix $X$ of even dimension $N_f$, the function
\cite{Dalmazi:2001hd}
\begin{align}
  \label{eq:f}
  g_\nu(X)=\int_{\U(N_f)}dU\,(\det  U)^\nu
  \exp\big[\re\tr(X^\dagger UIU^T)\big]\,,
\end{align}
in terms of which we have
\begin{align}
  \label{eq:Zff}
  Z_\nu^\text{eff}(J_L,J_R)=g_\nu(\tilde J_L^\dagger) g_\nu(\tilde
  J_R)\,.
\end{align}
The expansion of the integral \eqref{eq:f} in powers of $X$ was
computed in \cite{Dalmazi:2001hd}.  Adapting eq.~(40) of that
reference to our case, we have with $\alpha=N_f+\nu-3$
\begin{align}
  \label{eq:gnu}
  g_\nu(X)&=\mathcal N\,(\Pf X)^\nu
  \left[1+A_\alpha\tr X^\dagger X+B_\alpha(\tr X^\dagger X)^2
  -C_\alpha\tr(X^\dagger X)^2+O(X^6)\right]
\end{align}
with coefficients
\begin{align}
  A_\alpha=\frac{1}{2(\alpha+2)}\,,\quad 
  B_\alpha=\frac{\alpha+1}{8\alpha(\alpha+2)(\alpha+3)}\,,\quad
  C_\alpha=\frac{1}{4\alpha(\alpha+2)(\alpha+3)}
\end{align}
and a normalization factor $\mathcal N$ that depends on $N_f$ and
$\nu$ but is not important for our present purposes.  ($\mathcal N$
can be determined from \cite{Smilga:1994tb} if desired, and for
$\nu=0$ we have $\mathcal N=1$ for all $N_f$.)
Equations~\eqref{eq:Zff} and \eqref{eq:gnu} imply that in the limit
$J_{L/R}\to0$ only the $\nu=0$ sector survives.  This is analogous to
QCD at $\mu=0$, where in the chiral limit the topological
susceptibility vanishes \cite{Crewther:1977ce}.  We thus obtain
\begin{align}
  \frac{Z_\nu^\text{eff}(J_L,J_R)}{\mathcal N^2
    \big[\Pf(\tilde J_L^\dagger)\Pf(\tilde J_R)\big]^\nu}
  &=1+\big[A_\alpha\tr\tilde J_L^\dagger \tilde J_L
  +B_\alpha(\tr\tilde J_L^\dagger \tilde J_L)^2
  -C_\alpha\tr (\tilde J_L^\dagger \tilde J_L)^2+
  (L\leftrightarrow R)\big] \notag\\
  &\quad +A^2_\alpha(\tr\tilde J_L^\dagger \tilde J_L)
  (\tr\tilde J_R^\dagger \tilde J_R)
  +O(\tilde J^6)\,.
  \label{eq:Znu_eff}
\end{align}

We can now match the right-hand sides of \eqref{eq:Z_micro_expanded}
and \eqref{eq:Znu_eff} to obtain sum rules for the inverse singular
values,
\begin{subequations}
  \label{eq:sum_N_f>2}
  \begin{align}
    \label{eq:sum_N_f>2-1}
    \bigg\langle{\sum_{n}}'\frac{1}{\xi_{Ln}^2}\bigg\rangle_\nu
    &=\bigg\langle{\sum_{n}}'\frac{1}{\xi_{Rn}^2}\bigg\rangle_\nu
    =2\tilde A_\alpha\,,\;\;\:
    \bigg\langle\bigg({\sum_{n}}'\frac{1}{\xi_{Ln}^2}\bigg)^2\bigg\rangle_\nu
    =\bigg\langle\bigg({\sum_{n}}'\frac{1}{\xi_{Rn}^2}\bigg)^2\bigg\rangle_\nu
    =8\tilde B_\alpha\,,\\
    \bigg\langle{\sum_{n}}'\frac{1}{\xi_{Ln}^4}\bigg\rangle_\nu
    &=\bigg\langle{\sum_{n}}'\frac{1}{\xi_{Rn}^4}\bigg\rangle_\nu
    =4\tilde C_\alpha\,,\;\;\:
    \bigg\langle\bigg({\sum_{m}}'\frac{1}{\xi_{Rm}^2}\bigg)
    \bigg({\sum_{n}}'\frac{1}{\xi_{Ln}^2}\bigg)\bigg\rangle_\nu
    =4\tilde A_\alpha^2\,,
    \label{eq:sum_N_f>2-2}
  \end{align}
\end{subequations}
where $\tilde A_\alpha=(V_4\Phi_\I)^2A_\alpha$, $\tilde
B_\alpha=(V_4\Phi_\I)^4B_\alpha$, and $\tilde
C_\alpha=(V_4\Phi_\I)^4C_\alpha$.  These sum rules imply
\begin{align}
  \label{eq:ls_4}
  \bigg\langle{\sum_{n}}'\frac{1}{\xi_{n}^2}\bigg\rangle_\nu
  =4\tilde A_\alpha\,,\quad
  \bigg\langle\bigg({\sum_{n}}'\frac{1}{\xi_{n}^2}\bigg)^2\bigg\rangle_\nu
  =8\tilde A^2_\alpha+16\tilde B_\alpha\,,\quad
  \bigg\langle{\sum_{n}}'\frac{1}{\xi_{n}^4}\bigg\rangle_\nu
  =8\tilde C_\alpha\,.
\end{align}
Note that the conditions
\begin{align}
  \label{eq:cond}
  \bigg\langle\bigg({\sum_{n}}'\frac{1}{\xi_{n}^2}\bigg)^2\bigg\rangle_\nu
  \ge\bigg\langle{\sum_{n}}'\frac{1}{\xi_{n}^2}\bigg\rangle_\nu^2
  \quad\text{and}\quad 
  \bigg\langle\bigg({\sum_n}'\frac{1}{\xi_{Ln}^2}
  -{\sum_n}' \frac{1}{\xi_{Rn}^2}\bigg)^2\bigg\rangle_\nu
  =16\tilde B_\alpha-8\tilde A_\alpha^2\ge0
\end{align}
both lead to the inequality $2B_\alpha\geq A^2_\alpha$, which is
satisfied nontrivially for all $N_f\ge4$ since
\begin{align}
  2B_\alpha-A^2_\alpha=\frac{1}{2\alpha(\alpha+2)^2(\alpha+3)}>0\,.
\end{align}
The second condition in \eqref{eq:cond} also makes it clear that in
general the left- and right-handed sums are different (for a fixed
gauge field).  Higher-order sum rules can be computed by expanding the
partition functions to higher order in the diquark sources.

Let us now consider $N_f=2$ and for simplicity take $j_R$ and $j_L$
real.  In that case \eqref{eq:Z_micro_expanded} becomes
\begin{align}
  \frac{Z_\nu(j_L,j_R)}{(-j_Lj_R)^\nu}&=1+j_L^2j_R^2
  \bigg\langle\bigg({\sum_{n}}'\frac{1}{\xi_{Ln}^2}\bigg)
  \bigg({\sum_{m}}'\frac{1}{\xi_{Rm}^2}\bigg)\bigg\rangle_\nu\notag\\
  &\quad+\bigg[j_L^2\bigg\langle{\sum_n}'\frac{1}{\xi_{Ln}^2}\bigg\rangle_\nu
  +j_L^4\bigg\langle{\sum_{m<n}}'\frac{1}{\xi_{Lm}^2\xi_{Ln}^2}\bigg\rangle_\nu
  +(L\leftrightarrow R)\bigg]+O(j^6)\,.
  \label{eq:Z_micro_twoflavor}\end{align}
The finite-volume partition function obtained from
\eqref{eq:int2} in the $\varepsilon$-regime is
\begin{align}
  Z^\text{eff}(j_L,j_R)&=\int_{\U(1)}dV\,
  \exp\big[-2\big(\tilde j_L-\tilde j_R\big)\re V\big]
\end{align}
with $\tilde j_i=j_iV_4\Phi_\I$.
Introducing a $\theta$-angle as before we thus have
\begin{align}
  Z_\nu^\text{eff}(j_L,j_R)&=\int dV\,\frac{d\theta}{2\pi}\,
  \exp\big[-i\nu\theta-2\re(\tilde j_L\ee^{-i\theta/2}
  -\tilde j_R\ee^{i\theta/2})V\big]\notag\\
  &=\int_{\U(1)} dL\,L^\nu\exp\big[-2\tilde j_L\re L\big]
  \int_{\U(1)} dR\,(R^\dagger)^\nu\exp\big[2\tilde j_R\re R\big]\notag\\
  &=I_\nu(-2\tilde j_L)I_\nu(2\tilde j_R)\,,
  \label{eq:Zeff2}
\end{align}
in analogy with the steps that led to \eqref{eq:Zint}.  In the last
line we have recognized the integral representation of the modified
Bessel function $I_\nu$.  Expanding in the diquark sources gives
\begin{align}
  \frac{Z_\nu^\text{eff}(j_L,j_R)}{(\nu!)^{-2}
    (-\tilde j_L\tilde j_R)^\nu}
  =1+\frac{\tilde j_L^2+\tilde j_R^2}{\nu+1}
  +\frac{\tilde j_L^2\tilde j_R^2}{(\nu+1)^2}
  +\frac{\tilde j_L^4+\tilde j_R^4}{2(\nu+1)(\nu+2)}
  +O(\tilde j^6)\,,
  \label{eq:Z_expanded_twoflavor}
\end{align}
and matching \eqref{eq:Z_micro_twoflavor} and
\eqref{eq:Z_expanded_twoflavor} yields the sum rules
\begin{subequations}
  \begin{gather}
    \label{eq:ls_2}
    \bigg\langle{\sum_n}'\frac{1}{\xi_{Rn}^2}\bigg\rangle_{\!\!\nu}
    =\bigg\langle{\sum_n}'\frac{1}{\xi_{Ln}^2}\bigg\rangle_{\!\!\nu}
    =\frac{(V_4\Phi_\I)^2}{\nu+1}\,,
    \quad 
    \bigg\langle\bigg({\sum_m}'\frac{1}{\xi_{Rm}^2}\bigg)
    \bigg({\sum_n}'\frac{1}{\xi_{Ln}^2}\bigg)\bigg\rangle_{\!\!\nu}
    =\frac{(V_4\Phi_\I)^4}{(\nu+1)^2}\,,\\
    \bigg\langle{\sum_{m<n}}'\frac{1}{\xi_{Rm}^2\xi_{Rn}^2}\bigg\rangle_\nu
    =\bigg\langle{\sum_{m<n}}'\frac{1}{\xi_{Lm}^2\xi_{Ln}^2}\bigg\rangle_\nu
    =\frac{(V_4\Phi_\I)^4}{2(\nu+1)(\nu+2)}\,.
  \end{gather}
\end{subequations}
Note that the sum rules in \eqref{eq:ls_2} follow from those in
\eqref{eq:sum_N_f>2} by setting $N_f=2$.

We observe that there is a ``decoupling rule'' for $N_f \geq 2$,
\begin{align}
  \bigg\langle\bigg(\sum_m\frac{1}{\xi_{Rm}^2}\bigg)
  \bigg(\sum_n\frac{1}{\xi_{Ln}^2}\bigg)\bigg\rangle_\nu
  =\bigg\langle\sum_n\frac{1}{\xi_{Rn}^2}\bigg\rangle_\nu
  \bigg\langle\sum_n\frac{1}{\xi_{Ln}^2}\bigg\rangle_\nu\,.
\end{align}
A similar decoupling of the left- and right-handed modes is expected
to hold for higher moments as well since the finite-volume partition
function factorizes, see \eqref{eq:Zff}.  It is therefore sufficient
to consider only one of the factors.  Considering $N_f \geq 4$ for the
time being, we divide \eqref{eq:Znu_JLJR} and \eqref{eq:Zff} by $[\Pf
J_L^\dagger]^\nu$ and then take the limit $J_L\to0$ to obtain
\begin{align}
  \int_{\U(N_f)}dU\,(\det U)^\nu
  \exp\big[V_4\Phi_\I\re\tr(J_R^\dagger UIU^T)\big]
  \propto[\Pf J_R]^\nu\left\langle{\prod_n}'{\det}^{1/2}
  \left(\1+\frac{J_R^\dagger J_R}{\xi_{Rn}^2}\right)\right\rangle_\nu
\end{align}
modulo an unimportant normalization factor.  This has the same form as
the mass-dependent finite-volume partition function at $\mu=0$ for
$N_f \geq 2$ \cite{Smilga:1994tb},
\begin{align}
  Z_\nu^\text{eff}(\mathcal M)
  &=\int_{\U(2N_f)}dU\,(\det U)^\nu\exp\left[\frac{1}{2}{V_4\Sigma}
  \re\tr(\mathcal M^\dagger UIU^T)\right]\notag\\
  &\propto(\Pf\mathcal M)^\nu\left\langle\prod_{\lambda_n>0}{\det}^{1/2}
  \left(\1+\frac{\mathcal M^\dagger \mathcal M}{\lambda_{n}^2}\right)
  \right\rangle_\nu \quad\text{with}\quad 
  \mathcal M=\begin{pmatrix}0&M\\-M^T&0\end{pmatrix},
  \label{eq:Zm}
\end{align}
where $M$ is the quark mass matrix, $\Sigma$ is the absolute value of
the chiral condensate, the $i\lambda_n$ are the Dirac eigenvalues at
$\mu=0$, and the normalization factor was omitted again.  This implies
the correspondence
\begin{align}
  \big\{\Phi_\I,\; J_R,\; N_f,\; \xi_{Rn}\big\}_{\mu}
  \longleftrightarrow
  \big\{\tfrac12\Sigma,\; \mathcal M,\; 2N_f,\; \lambda_n(>0)\big\}_{\mu=0}
  \label{eq:rule}
\end{align}
and similarly for $R\leftrightarrow L$.  For $N_f=2$ this
correspondence remains valid since the $\U(2)$-integral in
\eqref{eq:Zm} then gives $I_\nu(-mV_4\Sigma)$ with the quark mass $m$,
which is to be compared with $I_\nu(2j_RV_4\Phi_\I)$ in
\eqref{eq:Zeff2}.

It is well known that to lowest order in the $\eps$-regime the system
can alternatively by described by chiral random matrix theory.  For
our present case, i.e., two-color QCD at intermediate density in the
chiral limit and in the presence of diquark sources, we obtain the
chiral random matrix theory\footnote{Random matrix theories for
  singular values are also considered in a very different context in
  \cite{Burda2010,Burda:2011ky}.}
\begin{align}
  Z_\nu^\text{RMT}(\hat J_L,\hat J_R)=g_\nu^\text{RMT}(\hat J_L^\dagger)
  g_\nu^\text{RMT}(\hat J_R)
\end{align}
with
\begin{align}    
  g_\nu^\text{RMT}(\hat J)&=\int dA\, \ee^{-N\tr(A^TA)}
  \Pf\begin{pmatrix}
    \hat J^\dagger\otimes\1_N & \1_{N_f}\otimes A \\
    \1_{N_f}\otimes(-A^T) & \hat J\otimes\1_{N+\nu}
  \end{pmatrix}\notag\\
  &=\int dA\, \ee^{-N\tr(A^TA)}\prod_{f=1}^{N_f/2}\hat j_f^\nu\,
  \det\big(AA^T+|\hat j_f|^2 \big)\,,
  \label{eq:Zrmt}
\end{align}
where $A$ is a real matrix of dimension $N\times(N+\nu)$.  In the
second line, we have applied an orthogonal transformation $O$ to bring
the antisymmetric matrix $\hat J$ to the standard form
\begin{align}
  \hat J=O\begin{pmatrix}
    0&\diag(\hat j_1,\ldots,\hat j_{N_f/2})\\
    -\diag(\hat j_1,\ldots,\hat j_{N_f/2})&0
  \end{pmatrix}O^T\,,
\end{align}
where the $\hat j_f$ can be complex.  For $\nu<0$ we have to replace
$\hat j_f^\nu$ by $(\hat j_f^*)^{|\nu|}$ and $AA^T$ by $A^TA$ in
\eqref{eq:Zrmt}.  Note that the second line of \eqref{eq:Zrmt} clearly
exhibits the correspondence of the random matrix theory for our
present case and for two-color QCD at $\mu=0$ with nonzero masses and
without diquark sources \cite{Verbaarschot:1994qf}.

Going through the standard steps of converting the random matrix
theory to a sigma-model we find in the limit $N\to\infty$ that
$g_\nu^\text{RMT}(\hat J)=g_\nu(\sqrt2N\hat J)$, which shows that the
random-matrix partition function is equivalent to the finite-volume
partition function,\footnote{This equivalence is expected to extend to
  the microscopic correlations of the singular values, although this
  still needs to be proven using the partially quenched theory,
  similar to \cite{Basile:2007ki}.} provided that the dimensionless
random-matrix diquark sources and singular values (i.e., the square
roots of the eigenvalues of $A^TA$) are related to the physical
quantities by
\begin{equation}
  \label{eq:J_scale_high}
  \hat J_i=J_iV_4\Phi_\I/\sqrt2N
  \qquad\text{and}\qquad   \hat\xi_i=\xi_i V_4\Phi_\I/\sqrt 2N
  \qquad(i=L,R)\,.  
\end{equation}

Note that the partition function only factorizes in the chiral limit.
If we include quark masses the random-matrix partition function is
\begin{align}
  \label{eq:Zrmt_m}
  Z_\nu^\text{RMT}(\hat J_L,\hat J_R,\hat M)
  =\int dA_L\,dA_R\,\ee^{-N\tr(A_L^TA_L+A_R^TA_R)}
  \Pf\begin{pmatrix}
    \hat J_L & A_L & -\hat M^T & 0\\
    -A_L^T & \hat J_L^\dagger & 0 & -\hat M^\dagger\\
    \hat M & 0 & -\hat J_R^\dagger & -A_R\\
    0 & \hat M^* & A_R^T & -\hat J_R
  \end{pmatrix},
\end{align}
where $A_L$ and $A_R$ are again real $N\times(N+\nu)$ matrices.  This
is a natural extension of the random matrix theory constructed at high
density in the absence of diquark sources
\cite{Kanazawa:2009en,Kanazawa:2011tf}.  How the dimensionless
random-matrix quark masses are related to the dimensionful masses
depends on the density.  At high density we have $\hat
M=M\sqrt{3V_4/N}\Delta/2\pi$ \cite{Kanazawa:2009en}.  What sets the
scale for the masses at intermediate density is a dynamical question
that cannot be answered with the methods we employ here.

For $\mu=0$, the microscopic spectral correlations of the Dirac
eigenvalues have been computed in chiral random matrix theory
\cite{Nagao:1995pz,Nagao:2000cb}.  The microscopic scale is defined by
$z=\lambda V_4\Sigma$, and as an example we quote the microscopic
spectral density \cite{Verbaarschot:1994ia}
\begin{align}
  \label{eq:rhos}
  \rho^{N_f,\,\nu}_s(z)=
  \frac{z}{2}\big[J_{a}^2(z)-J_{a+1}(z)J_{a-1}(z)\big]
  +\frac{1}{2}J_{a}(z)\Big[1-\int_0^z dw\,J_{a}(w)\Big]\,,
\end{align}
where $a=2N_f+|\nu|$ and $J$ denotes the Bessel function.  According
to the correspondence \eqref{eq:rule}, the microscopic density of the
Dirac singular values for a given chirality,
\begin{align}
  \rsvm^{N_f,\,\nu}(x)=\lim_{V_4\to\infty} \Big\langle
  \sum_n\delta(x-x_n)\Big\rangle_\nu\quad\text{with}\quad
  x_n=\xi_{Rn}V_4\Phi_\I\text{ or }\xi_{Ln}V_4\Phi_\I\,,
\end{align}
is given by
\begin{align}
  \rsvm^{N_f,\,\nu}(x)=2\rho_s^{\frac12N_f,\,\nu}(2x)
\end{align}
for $N_f \geq 2$, and analogously for higher-order correlation
functions.  By construction all sum rules derived in this section are
moments of suitable microscopic correlation functions.  We have
checked numerically that this is indeed true for the moments of the
microscopic density of the singular values.

Note that for $N_f=2$ we could not obtain sum rules for, e.g.,
$\ev{\sum_n1/\xi_{Rn}^4}_\nu$ or $\ev{\sum_n1/\xi_{Ln}^4}_\nu$ from
the expansion of the partition function.  However, we can obtain them
as moments of the microscopic density.  It follows from the Taylor
expansion of \eqref{eq:rhos} that the above two sums diverge for
$\nu=0$ and $1$, which is consistent with the observation that
formally setting $N_f=2$ in \eqref{eq:sum_N_f>2-2} gives a negative or
infinite result for $\nu=0$ or $1$, respectively.  But for $\nu\ge2$
these sums converge, and the result is identical to the one in
\eqref{eq:sum_N_f>2-2} with $N_f=2$.  Since in the random matrix
theory there is nothing special about $N_f=2$ we expect that all sum
rules derived for $N_f\ge4$ remain valid for $N_f=2$, if convergent.

\subsection{High density}
\label{sec:ls_high}

At asymptotically high density topology is suppressed and
$m_\text{inst}=0$.  In that case we expect to obtain the $\nu=0$
subset of the results of the preceding subsection.  However, as long
as $m_\text{inst}$ is still nonzero the $\nu\ne0$ sectors are not
completely suppressed.  These expectations will be confirmed below.

On the QCD side nothing changes, i.e., we match to
\eqref{eq:Z_micro_expanded}.  On the low-energy effective theory side,
let us again start with $N_f\ge4$.  Neglecting the kinetic terms in
\eqref{eq:lagrangian_high} gives
\begin{align}
  Z^\text{eff}(J_L,J_R)&=\frac1{I_0(\kappa)}
  \int d\Sigma_L\,d\Sigma_R\,dL\,dR\,
  \exp\big[\re\tr(\tilde J_LL\Sigma_L-\tilde J_RR\Sigma_R)
  +\kappa\re(L^\dagger R)^{N_f/2}\big]
\end{align}
with $\tilde J_i=J_iV_4\Phi_\H$ and
$\kappa=2V_4\fh_0^2m_\text{inst}^2/N_f\ge0$.\footnote{Although $\tilde
  f_0\sim\mu$ we have $\kappa\to0$ for $\mu\to\infty$ since
  $m_\text{inst}^2$ goes to zero much faster than $1/\mu^2$
  \cite{Schafer:2002ty}.} The integrals over $L,R$ are over $\U(1)$,
while those over $\Sigma_{L,R}$ are over $\SU(N_f)/\Sp(N_f)$.  The
normalization factor $1/I_0(\kappa)$ has been added to ensure
$Z_\text{eff}(0,0)=1$ as in \eqref{eq:ZeffI}. Introducing a
$\theta$-angle as before and projecting onto topological sectors using
\eqref{eq:Z_nu} we obtain
\begin{align}
  Z_\nu^\text{eff}(J_L,J_R) 
  &=\frac1{I_0(\kappa)}
  \int d\Sigma_L\,d\Sigma_R\,dL\,dR\,\frac{d\theta}{2\pi}\,\notag\\
  &\quad\times\exp\big[-i\nu\theta
  +\re\tr(\tilde J_L\ee^{-i\theta/N_f}L\Sigma_L
  -\tilde J_R\ee^{i\theta/N_f}R\Sigma_R)
  +\kappa\re(L^\dagger R)^{N_f/2}\big]\notag\\
  &=\int d\Sigma_L\,d\Sigma_R\,dL\,dR\,(LR^\dagger)^{\frac12N_f\nu}
  \exp\big[\re\tr(\tilde J_LL\Sigma_L-\tilde J_RR\Sigma_R)\big]\notag\\
  &\quad\times\frac1{I_0(\kappa)}\int\frac{d\theta}{2\pi}\,
  \exp\big(-i\nu\theta+\kappa\cos\theta\big)\notag\\
  &=\big[Z_\nu^\text{eff}(J_L,J_R)\text{ of theory \I}
  \big]_{\Phi_\I\to\Phi_\H}
  \times \frac{I_{\nu}(\kappa)}{I_0(\kappa)}\,,
  \label{eq:ZeffH}
\end{align}
where in the second line we have first redefined $L\to
L\ee^{i\theta/N_f}$, $R\to R\ee^{-i\theta/N_f}$ and then
$\ee^{-2i\theta/N_f}\to \ee^{-2i\theta/N_f}LR^\dagger$, and in the
last line we have compared with \eqref{eq:Zint}.  From
\eqref{eq:ZeffH} we can draw several conclusions.  First, all
Leutwyler-Smilga-type sum rules in theory \H\ are identical to those
of theory \I\ (except for the replacement $\Phi_\I\to\Phi_\H$) since
the relative factor $I_{\nu}(\kappa)/I_0(\kappa)$ is independent of
$J_{L/R}$ and therefore drops out when computing expectation values in
sectors of fixed topology.  Second, in the limit $\kappa\to\infty$ the
relative factor goes to 1 for any $\nu$, while for finite $\kappa$ the
sectors with $\nu\ne0$ are suppressed and disappear completely for
$\kappa\to0$ as expected.  Third, even in the presence of the anomaly
the finite-volume partition function still factorizes into a left- and
a right-handed part as in \eqref{eq:Zff}.

For $N_f=2$ the argument goes through in exactly the same way.

Since the sum rules for theories \I\ and \H\ are identical, it is
natural to expect that the random matrix theory and the microscopic
correlation functions of the singular values for \H\ are identical to
those of \I, with $\Phi_\I\leftrightarrow\Phi_\H$ (again, this would
have to be proven using the partially quenched theory).  The only
difference is in the summation over topological charge.  If, e.g., one
wants to sum the microscopic spectral density over all sectors, one
needs to take into account the $\nu$-dependent relative factor in
\eqref{eq:ZeffH}, i.e.,
\begin{align}
  \rsvm^\I(x,\theta)=\frac{\sum\limits_\nu\rsvm^\nu(x)\,
    \ee^{i\nu\theta}\, Z_\nu^\I} 
  {\sum\limits_\nu^{} \ee^{i\nu\theta}\, Z_\nu^\I}\quad\text{and}\quad
  \rsvm^\H(x,\theta)=\frac{\sum\limits_\nu\rsvm^\nu(x)\,
    \ee^{i\nu\theta}\, Z_\nu^\I\, I_{\nu}(\kappa)}
  {\sum\limits_\nu^{} \ee^{i\nu\theta}\, Z_\nu^\I\, I_{\nu}(\kappa)}\,,
\end{align}
where we have suppressed the arguments $J_L$ and $J_R$ in $Z_\nu^\I$
and $\rsvm$.  Note that in the limit of zero diquark sources
$Z_\nu^\I=0$ for $\nu\ne0$ so that $\rsvm(x,\theta)=\rsv^0(x)$ for
both \I\ and \H.

\subsection{Low density}

On the low-energy effective theory side, we use \eqref{eq:Lefflow} and
for simplicity take $J_R=-J_L=jI$ with real $j$.  Neglecting the
kinetic terms we obtain the finite-volume partition function
\begin{align}
  Z^\text{eff}(\mu,j)=\int d\Sigma\,
  \exp\big[\mu^2F^2V_4\tr(\Sigma B\Sigma^\dagger B)
  +jV_4\Phi_\L\re\tr(\ee^{-i\theta/N_f}\Sigma_d\Sigma)\big]\,,
  \label{eq:Zeff_low}
\end{align}
where the integration is over $\SU(2N_f)/\Sp(2N_f)$, the $B^2$-term
has been absorbed in the normalization, and a $\theta$-angle has been
introduced by $J_R\to J_R\ee^{i\theta/N_f}$ and $J_L\to
J_L\ee^{-i\theta/N_f}$.  The expansion of \eqref{eq:Zeff_low} in powers
of $j$ for arbitrary $N_f$ is formally possible, but the analytical
calculation of the expansion coefficients is a challenging
mathematical problem which we do not address here.  Instead, we will
obtain partial results for the special case $N_f=2$.  Because of the
group isomorphisms $\SU(4)\simeq\SO(6)$ and $\Sp(4)\simeq\SO(5)$ we
can regard the coset as $\SO(6)/\SO(5)\simeq S^5$
\cite{Smilga:1994tb}.  This approach has been explored in detail
\cite{Brauner:2006dv}, and in the following we use the formulation of
that reference.  Adapting the unnumbered equation just before eq.~(7)
of \cite{Brauner:2006dv} to our case, we have modulo a $j$-independent
normalization factor
\begin{align}
  Z^\text{eff}(\mu,j)=\int_{S^5} d\vec n\,
  \exp\left[z(n_2^2+n_4^2)+2wn_4\cos\frac{\theta}{2}\right]
\end{align}
with
\begin{align}
  z = 8\mu^2F^2V_4\quad\text{and}\quad w = 2jV_4\Phi_\L\,.
\end{align}
Projecting onto fixed topology using \eqref{eq:Z_nu} we obtain
\begin{align}
  Z_\nu^\text{eff}(\mu,j)&=\int_{S^5} d\vec n\,
  \ee^{z(n_2^2+n_4^2)} \int_{\U(1)}d\theta\,
  \ee^{-i\nu\theta+2wn_4\cos\tfrac\theta2}\notag\\
  &=\int_{S^5} d\vec n\, \ee^{z(n_2^2+n_4^2)}\,
  I_{2\nu}(2wn_4)\notag\\
  &=\int_{S^5} d\vec n\,
  \ee^{z(n_2^2+n_4^2)}\,\frac{(wn_4)^{2\nu}}{(2\nu)!}
  \bigg[1+\frac{(wn_4)^2}{2\nu+1}+O(j^4)\bigg]\,.
  \label{eq:Z2eff}
\end{align}
As before, we always assume $\nu\ge0$.  For $\nu<0$ we have to replace
$\nu\to|\nu|$.

On the QCD side we again use \eqref{eq:Z_micro_expanded}, which now
becomes
\begin{align}
  Z_\nu(\mu,j)\propto j^{2\nu}\bigg[1+j^2\bigg\langle
  {\sum_n}'\frac1{\xi_n^2}\bigg\rangle_\nu+O(j^4)\bigg]\,,
  \label{eq:Z2qcd}
\end{align}
where we omitted a $j$-independent normalization factor.  Comparing
the ratio of the $j^{2\nu}$ and $j^{2\nu+2}$ terms in \eqref{eq:Z2eff}
and \eqref{eq:Z2qcd} we obtain the sum rule
\begin{align}
  \label{eq:sumlownu}
  \bigg\langle {\sum_n}'\frac1{\xi_n^2}\bigg\rangle_\nu
  =\frac{4}{2\nu+1} (V_4\Phi_\L)^2
  \frac{{\displaystyle\int_{S^5}} d\vec n\,
  \ee^{z(n_2^2+n_4^2)}\,n_4^{2\nu+2}}
  {{\displaystyle\int_{S^5}} d\vec n\, \ee^{z(n_2^2+n_4^2)}\,n_4^{2\nu}}\,.
\end{align}
For $z=0$, \eqref{eq:sumlownu} reproduces the sum rule in
\cite{Smilga:1994tb}, as it should.%
\footnote{In this comparison we need to be careful. First, our
  $\Phi_\L$ corresponds to $\Sigma/2$ in \cite{Smilga:1994tb}.
  Second, all singular values are doubly degenerate (for $\beta=1$) at
  $\mu=0$, which means that the sum rule in \cite{Smilga:1994tb}
  should be compared to our result divided by 2.} In the opposite
limit $z\to\infty$, the r.h.s.\ converges to $2(V_4\Phi_\L)^2/(\nu+1)$
in agreement with \eqref{eq:ls_2} (with
$\Phi_\I\leftrightarrow\Phi_\L$).  The analytical calculation of the
r.h.s.\ for finite $z$ is nontrivial, and we only consider the special
case $\nu=0$, for which we obtain
\begin{align}
  \label{eq:sumlow0}
  \bigg\langle {\sum_n}' \frac1{\xi_n^2}\bigg\rangle_0 = 
  2(V_4\Phi_\L)^2\left(\frac z{\ee^z-z-1}+1-\frac2z\right)\,,
\end{align}
the plot of which is shown as a function of $z$ in figure
\ref{fig:sum_rule_L}.  The proof of this result is given in
appendix~\ref{app:sumlow0}.  The sum rules for the right- and
left-handed singular values are also given by \eqref{eq:sumlow0} but
with the r.h.s.\ divided by 2.  Higher-order sum rules can be obtained
as usual by expanding \eqref{eq:Z2eff} and \eqref{eq:Z2qcd} to higher
orders in $j$.

\begin{figure}
  \centering
  \includegraphics{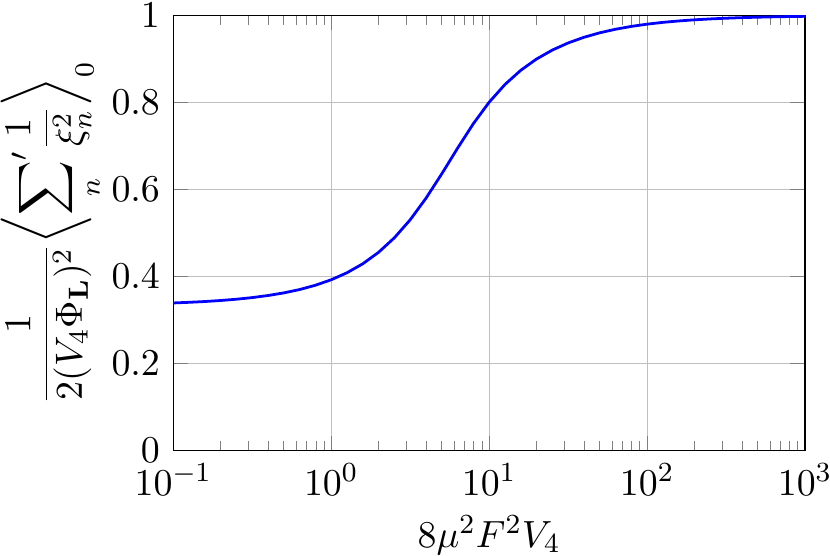}
  \caption{$\mu$-dependence of the spectral sum rule
    \eqref{eq:sumlow0} in theory \L\ in the sector $\nu=0$.  The curve
    converges to $1/3$ as $\mu\to 0$ and to $1$ as $\mu\to\infty$.}
  \label{fig:sum_rule_L}
\end{figure}

As before, to lowest order in the $\eps$-regime we can also describe
the system by a random matrix theory, which at low density we can
obtain by adding diquark sources to the known two-matrix model
\cite{Akemann:2008mp,Akemann:2010tv}, resulting in
\begin{multline}
  Z_\nu^\text{RMT}(\hat \mu, \hat J_L,\hat J_R,\hat M) =
  \\
  \label{eq:Zrmt_low}
  \int dC\,dD\,\ee^{-2N\tr(C^TC+D^TD)}
  \Pf\bep
    \hat J_L & C-\hat\mu D & -\hat M^T & 0\\
    -C^T+\hat\mu D^T  & \hat J_L^\dagger & 0 & -\hat M^\dagger\\
    \hat M & 0 & -\hat J_R^\dagger & -C-\hat\mu D\\
    0 & \hat M^* & C^T+\hat\mu D^T & -\hat J_R
  \eep.
\end{multline}
For earlier approaches, see \cite{Vanderheyden:2000ti,Klein:2004hv}.
In \eqref{eq:Zrmt_low}, $C$ and $D$ are again real $N\times(N+\nu)$
matrices.  Converting \eqref{eq:Zrmt_low} to a sigma-model we find in
the limit $N\to\infty$ that the random matrix parameters are related
to physical quantities by
\begin{align}
  \hat\mu^2=2\mu^2F^2V_4/N\,,\quad
  \hat M=MV_4\Phi_\L/2N\,,\quad
  \hat J_i=J_iV_4\Phi_\L/2N\,,\quad
  \hat \xi_i=\xi_iV_4\Phi_\L/2N\,,
\end{align}
where in the chiral limit the $\xi_L$ ($\xi_R$) are the singular
values of $C-\hat\mu D$ ($C+\hat\mu D$).  Note that for $\hat\mu\ne1$
the partition function does not factorize even in the chiral limit.
However, for $\hat\mu=1$ (``maximum non-Hermiticity'')
\eqref{eq:Zrmt_low} reduces to \eqref{eq:Zrmt_m}, and if the chiral
limit is taken it factorizes again.  The computation of the
microscopic correlation functions of the eigenvalues and/or singular
values from this random matrix theory is a complicated mathematical
task which we do not attempt here.

In section~\ref{sec:eps} we have seen that the $\eps$-regimes of the
effective theories \L\ and \I\ do not overlap. Nevertheless, as just
noted, the random matrix theory \eqref{eq:Zrmt_m} for \I\ is the
$\hat\mu\to1$ limit of the random matrix theory \eqref{eq:Zrmt_low}
for \L.  This is consistent with the observation that the sum of the
sum rules for $1/\xi_{Rn}^2$ and $1/\xi_{Ln}^2$ in \eqref{eq:ls_2} is
the $z\to\infty$ limit of \eqref{eq:sumlownu} for all $\nu$ (with
$\Phi_\I\leftrightarrow\Phi_\L$).  We are therefore tempted to
conjecture that all sum rules of \I\ are the $z\to\infty$ limits of
the sum rules for \L.

\section{Conclusions and outlook}
\label{sec:conclusion}

In this paper, we have studied the singular values of the Dirac
operator in QCD-like theories with Dyson indices $\beta=1$, $2$, and
$4$ at nonzero chemical potential $\mu$.  We pointed out that the
Dirac singular values are always real and nonnegative and that the
scale of the singular value spectrum is set by the diquark condensate
at any $\mu$.  This is in contrast to the Dirac eigenvalues, which
spread into the complex plane at nonzero $\mu$ and whose scale is set
by the chiral condensate at small $\mu$ \cite{Akemann:2007rf} and by
the BCS gap at large $\mu$
\cite{Kanazawa:2009ks,Kanazawa:2009en,Akemann:2010tv}.  We derived
Banks-Casher-type relations for all three values of $\beta$ and then
concentrated on the $\beta=1$ case, for which we identified three
different low-energy effective theories with diquark sources at low,
intermediate, and high density within the whole BEC-BCS crossover
region, and clarified how they are related to each other from the
point of view of integrating out heavy degrees of freedom.  We derived
exact results, such as Smilga-Stern-type relations and
Leutwyler-Smilga-type sum rules, which (together with the
Banks-Casher-type relation) concern the connection between the
singular value spectrum and diquark condensation.  We have also
identified the $\epsilon$-regimes of the effective theories and
constructed the corresponding chiral random matrix theories, from
which the microscopic spectral correlation functions of the singular
values can be determined.  Our results can in principle be tested in
future lattice QCD simulations.  This should provide a value of the
diquark condensate at any density, by which the conjectured BEC-BCS
crossover could be confirmed numerically.

It is evident from our results that the existence of a nonzero diquark
condensate at any chemical potential implies an accumulation of the
Dirac singular values at the origin at any quark density.  In the case
of the QCD vacuum, near-zero Dirac eigenvalues responsible for chiral
symmetry breaking are believed to originate from instantons, as
illustrated by the instanton liquid model \cite{Schafer:1996wv}.  One
might thus naively expect that the accumulation of near-zero singular
values in dense QCD is also attributable to instantons. However, this
is not the case.  Although instanton effects are presumably important
at small and intermediate density, they are suppressed at sufficiently
high density where the one-gluon exchange interaction is more
important for the formation of diquark pairing. The presence of the
Fermi surface is crucial in this mechanism.

Let us discuss some possible extensions of the present work.  First,
here we have concentrated on nonzero diquark sources without quark
masses, but the generalization to include quark masses looks
straightforward. In that way we can study not only the diquark
condensate but also the chiral condensate or the BCS gap as a function
of the chemical potential.  Second, the results obtained in
sections~\ref{sec:eff} through \ref{sec:ls} could also be generalized
to theories with $\beta=2$ and $\beta=4$.  Third, it would be very
interesting to generalize our results to the color-superconducting
phases of three-color QCD.  Unfortunately, this is not straightforward
since the diquark source is no longer gauge invariant, though the
magnitude of the diquark condensate is.  The object $D(\mu)^\dagger
D(\mu)$ obviously exists in three-color QCD and can be studied, but it
is currently unclear to us how its spectrum is related to physical
observables.  It would also be interesting to find out whether (and if
so, how) the gauge invariant four-quark condensate is related to the
Dirac eigenvalue or singular value spectrum in three-color QCD.

We conclude with a short discussion of a phenomenon analogous to the
conjectured BEC-BCS crossover in QCD-like theories studied in this
paper, i.e., the conjectured hadron-quark continuity in three-color
QCD at nonzero baryon density \cite{Schafer:1998ef}.  Recently, the
existence of a bound state of baryons, the H dibaryon (with mass
$m_H$), was observed in lattice QCD simulations at zero density
\cite{Beane:2010hg,Inoue:2010es}. This implies (for three degenerate
flavors) that, as we go to nonzero baryon density, first BEC of H
dibaryons occurs, since due to its binding energy a bosonic H dibaryon
has a smaller excitation energy $m_H/2 - \mu_B$ (per baryon) than a
baryon. In this BEC state $\U(1)_B$ and chiral symmetry are broken
spontaneously, and thus this state has the same symmetry-breaking
pattern as the color-flavor-locked phase at high density, where both
$\U(1)_B$ and chiral symmetry are broken by the diquark condensate.
Therefore these two phases can be continuously connected without any
phase transition (hadron-quark continuity) as first conjectured by
Sch\"afer and Wilczek \cite{Schafer:1998ef}.  An explicit realization
of this conjecture was given within generalized Ginzburg-Landau theory
\cite{Hatsuda:2006ps,Yamamoto:2007ah,Schmitt:2010pf} and the
Nambu--Jona-Lasinio model \cite{Abuki:2010jq}, where it was shown that
the QCD axial anomaly can lead to a crossover between the hadronic and
the CFL phase.  Although it is not yet clear what happens physically
between these two regions, there could be successive changes of
states, from BEC of dibaryons to BCS pairing of dibaryons to BEC of
diquarks and finally to BCS pairing of diquarks.

In the real world, flavor symmetry is explicitly broken due to the
heavy strange quark mass, and the existence of the H dibaryon as well
as the scenario above may be modified.  Actually, it is empirically
believed that a nuclear liquid-gas phase transition takes place as the
baryon density increases. In order to understand high-density matter
in the real world, it would thus be crucial to take into account the
effects of flavor symmetry breaking.  Unfortunately, these effects
cannot be directly studied in lattice QCD simulations even for
QCD-like theories since, e.g., nondegenerate quark masses in two-color
QCD give rise to the sign problem. It is an important future problem
to figure out how one can study flavor symmetry breaking at nonzero
density on the lattice.

\acknowledgments We thank J.\ Bloch, K.\ Maeda, and H.\ Suzuki for
helpful discussions and comments.  TK is supported by the Alexander
von Humboldt foundation.  TW is supported by DFG (SFB/TR-55) and by
the EU (ITN STRONGnet).  NY is supported by a Japan Society for the
Promotion of Science (JSPS) Postdoctoral Fellowship for Research
Abroad.

\appendix

\section{Definitions and conventions}
\label{app:conv}

Unless stated otherwise we always work in Euclidean space.  The
$\gamma$-matrices are Hermitian and satisfy
$\{\gamma_\mu,\gamma_\nu\}=2\delta_{\mu\nu}$ ($\mu,\nu=1,\ldots,4$).
We choose the chiral representation given by
\begin{align}
  \gamma_i=\begin{pmatrix}0&-i\sigma_i\\i\sigma_i&0\end{pmatrix},\qquad
  \gamma_4&=\begin{pmatrix}0&\1\\ \1&0\end{pmatrix},\qquad
  \gamma_5=\gamma_1\gamma_2\gamma_3\gamma_4=
  \begin{pmatrix}\1&0\\0&-\1\end{pmatrix},
\end{align}
where the $\sigma_i$ ($i=1,2,3$) are the usual Pauli matrices.  
A Dirac spinor $\psi$ can be written in terms of two Weyl spinors
$\psi_R$ and $\psi_L$,
\begin{align}
  \psi=\begin{pmatrix}\psi_R\\\psi_L\end{pmatrix},
\end{align}
and we have
\begin{align}
  \bar\psi=\psi^\dagger\gamma_4=
  \big( \psi_L^\dagger ~ \psi_R^\dagger \big)\,.
\end{align}
The projection operators on the right- and left-handed sectors are
\begin{align}
  P_R=\frac12(\1+\gamma_5)\quad\text{and}\quad P_L=\frac12(\1-\gamma_5)\,,
\end{align}
and we have, with a slight abuse of notation,
\begin{equation}
  \begin{aligned}
    P_R\psi&=\psi_R\,,\qquad & \bar\psi P_R&=\psi_L^\dagger\,,\\
    P_L\psi&=\psi_L\,, & \bar\psi P_L&=\psi_R^\dagger\,.
  \end{aligned}
\end{equation}
The charge conjugation matrix satisfies
\begin{align}
  C\gamma_\mu C^{-1}=-\gamma_\mu^T\quad\text{and}\quad C^T=-C\,.
\end{align}
We adopt the choice $C=i\gamma_4\gamma_2$, for which we have
\begin{align}
  C^{-1}&=C^\dagger=C\,,\\
  [C,\gamma_1]&=\{C,\gamma_2\}=[C,\gamma_3]=\{C,\gamma_4\}=[C,\gamma_5]=0\,.
\end{align}
We also define the $N_f$-dimensional antisymmetric matrix
\begin{align}
  \label{eq:I}
  I=\begin{pmatrix}
    0&-\1_{N_f/2}\\
    \1_{N_f/2}&0
  \end{pmatrix}.
\end{align}

\section{Partition functions with diquark sources}
\label{app:Z-pfaffian}

In this appendix we derive the singular value representations of the
partition functions with diquark sources for the theories with
$\beta=1$, $2$, and $4$, taking exact zero modes into account and
showing how the positivity of the path integral measure is determined.

\subsection[Two-color QCD $(\beta=1)$]{\boldmath Two-color QCD $(\beta=1)$}
\label{app:Z1}

We first consider two-color QCD.  We need to evaluate the Pfaffian of
the operator $W$ in \eqref{eq:W} in the chiral limit, regarded as an
infinite-dimensional antisymmetric matrix,
\begin{align}
  \Pf(W) = 
  \Pf \begin{pmatrix}
    C\tau_2(J_RP_R+J_LP_L) & -D(\mu)^T 
    \\
    D(\mu) & -C\tau_2(J_R^\dagger P_L+J_L^\dagger P_R)
  \end{pmatrix}.
\end{align}
Note that the transpose $D(\mu)^T$ includes transposition of the
space-time indices, in addition to the color and spinor indices.
Below we will give two treatments of $\Pf(W)$. They lead to the same
form for the partition function, which underscores the correctness of
the result.

\subsubsection{Rigorous derivation}

To evaluate the Pfaffian it is useful to employ a specific functional
basis.  Here we use the eigenfunctions of $D^\dagger D$ and
$DD^\dagger$ introduced in section~\ref{sc:singu},%
\footnote{Note that these eigenfunctions are ordinary $c$-number
  functions and not Grassmannian.}
\begin{subequations}
  \label{eq:def_phi}
  \begin{alignat}{2}
    \label{eq:def_phi_}
    D^\dagger D \phi_n &= \xi_n^2\phi_n\,, \qquad &
    \int d^4x\, \phi_m^\dagger \phi_n &= \delta_{mn} \,,
    \\
    \label{eq:def_tphi_}
    DD^\dagger \tilde{\phi}_n &= \xi_n^2 \tilde \phi_n\,, &
    \int d^4x\, \tilde\phi_m^\dagger \tilde\phi_n &= \delta_{mn}\,,
  \end{alignat}
\end{subequations}
where for $\xi_n>0$ we have
\begin{align}
  \label{eq:rel_bases}
  \phi_n = \frac{1}{\xi_n}D^\dagger \tilde \phi_n
  \quad\text{and}\quad
  \tilde\phi_n = \frac{1}{\xi_n}D \phi_n\,.
\end{align}
Then the fields $\psi$ and $\bar\psi$ in
\eqref{eq:Psi} can be expanded in the bases $\{\phi_n\}$ and
$\{\tilde\phi_n^\dagger\}$, respectively.  We need the help of the
following lemma to proceed.
\begin{lem}
\label{eq:lemma}
  Without loss of generality, we can assume
  \begin{subequations}
    \begin{align}
      \phi_n & = p_n C\tau_2 \phi_n^*\,, 
      \label{eq:phi_parity} \\
      \tilde \phi_n & = - p_n C\tau_2 \tilde \phi_n^*
      \label{eq:tphi_parity}
    \end{align}
  \end{subequations}
  for some $p_n\in\C$ with $|p_n| = 1$.
\end{lem}
\begin{proof}
  For two-color QCD, it can be shown from \eqref{eq:symm_i} that
  \begin{subequations}
    \begin{align}
      (D^\dagger D)^* &= C\tau_2 D^\dagger D C\tau_2 \,,\\
      (DD^\dagger)^* &= C\tau_2 DD^\dagger C\tau_2 \,.
    \end{align}
  \end{subequations}
  Using these properties as well as \eqref{eq:def_phi}, it follows
  that
  \begin{subequations}
    \begin{align}
      D^\dagger D (C\tau_2 \phi_n^*) & = \xi_n^2 (C\tau_2\phi_n^*)\,,\\
      DD^\dagger (C\tau_2 \tilde\phi_n^*) & = \xi_n^2
      (C\tau_2\tilde\phi_n^*)\,. 
    \end{align}
  \end{subequations}
  Therefore both $\phi_n$ and $C\tau_2\phi_n^*$ $\big(\tilde \phi_n$
  and $C\tau_2\tilde \phi_n^*\big)$ are eigenfunctions of $D^\dagger
  D$ $\big(DD^\dagger\big)$ with the same eigenvalue $\xi_n^2$. Then
  two possibilities arise:
  \begin{enumerate}
  \item $\phi_n$ and $C\tau_2\phi_n^*$ are linearly
    independent, and the eigenvalue $\xi_n^2$ is (at least)
    doubly degenerate.
  \item $\phi_n$ and $C\tau_2\phi_n^*$ are linearly dependent,
    and the eigenvalue $\xi_n^2$ is not degenerate.
  \end{enumerate}
  In the first case, we can redefine ${\mathcal N_1}(\phi_n + p_n
  C\tau_2\phi_n^*)$ as $\phi_n$ and ${\mathcal N_2}(\phi_n - p_n
  C\tau_2\phi_n^*)$ as $\phi_{n+1}$ with normalization constants
  $\mathcal{N}_{1,2}$ and arbitrary $p_n\in\U(1)$ so that%
  \footnote{This redefinition does not change the chirality because
    $C\tau_2$ commutes with $\gamma_5$.}
  \begin{align}
    \phi_n = p_n C\tau_2\phi_n^*\qquad\text{and}\qquad 
    \phi_{n+1} = - p_n C\tau_2\phi_{n+1}^* \equiv
    p_{n+1}C\tau_2\phi_{n+1}^* \,,
  \end{align}
  where we used $(C\tau_2)^*=C\tau_2$ and $(C\tau_2)^2=\1$.  In the
  second case, it can easily be shown that there exists a phase
  $p_n\in\U(1)$ such that $\phi_n = p_n C\tau_2\phi_n^*$, so the
  desired relation holds automatically.  This completes the proof of
  \eqref{eq:phi_parity}. Note that we have shown this for both zero
  and nonzero modes.

  Let us proceed to the proof of \eqref{eq:tphi_parity}.  For nonzero
  modes ($\xi_n > 0$), we have
  \begin{align}
    \tilde\phi_n^* = \frac{1}{\xi_n}D^*\phi_n^* 
    = \frac{1}{\xi_n}(-C\tau_2DC\tau_2) p_n^*C\tau_2\phi_n
    = - p_n^* C\tau_2 \frac{1}{\xi_n} D \phi_n 
    = -p_n^* C\tau_2 \tilde\phi_n\,.
  \end{align}
  Multiplying both sides by $-p_n C\tau_2$ we obtain
  \eqref{eq:tphi_parity}.  Finally we consider the case of zero modes
  $(\xi_n=0)$. Although there is no obvious relation between the zero
  modes of $D$ and those of $D^\dagger$, their numbers are equal, see
  \eqref{eq:dim_eq_2}.  Choosing the phases of the zero modes of
  $DD^\dagger$ one can make \eqref{eq:tphi_parity} hold.
\end{proof}

Now it is straightforward to obtain the matrix elements of $W$ in the
bases $\{\phi_n\}$ (for $\psi$) and $\{\tilde\phi_n^\dagger\}$ (for
$\bar\psi$).  We consider separately the four blocks of $W$.
\begin{itemize}
\item $(1,1)$-block:
  \begin{align}
    \int d^4x\, \phi_m^T \big[C\tau_2 (J_RP_R+J_LP_L)\big] \phi_n
    &= \int d^4x\, p_m \phi_m^\dagger (J_RP_R+J_LP_L) \phi_n
    \notag \\
    &= p_n \delta_{mn} J_{R/L}\,,
  \end{align}
  where in the first step we have used \eqref{eq:phi_parity} and in
  the second step $R/L$ corresponds to the handedness of $\phi_n$.
\item $(2,1)$-block:
  \begin{align}
    \int d^4x\, \tilde\phi_m^\dagger D\phi_n = \xi_n \delta_{mn}\,.
  \end{align}
\item $(1,2)$-block:
  \begin{align}
    \int d^4x\, \phi_m^T (-D^T) \tilde\phi_n^* 
    = -\xi_n \delta_{mn}\,.
  \end{align}
\item $(2,2)$-block:
  \begin{align}
    \int d^4x\, \tilde\phi_m^\dagger 
    \big[-C\tau_2 (J_R^\dagger P_L+J_L^\dagger P_R)\big] \tilde\phi_n^*
    &= \int d^4x\, \tilde\phi_m^\dagger 
    (J_R^\dagger P_L+J^\dagger _LP_R) p_n^* \tilde \phi_n
    \notag \\
    &= p^*_n \delta_{mn} J^\dagger_{R/L} \,,
  \end{align}
  where in the first step we have used \eqref{eq:tphi_parity} and in
  the second step $R/L$ corresponds to the \emph{opposite} of the
    handedness of $\tilde\phi_n$.
\end{itemize}
Collecting all results, we obtain
\begin{align}
  \Pf(W)  = 
  \Pf \bep
  \diag(p_n)\otimes J_{R/L}  &  -\diag(\xi_n)\otimes \1_{N_f}
  \\
  \diag(\xi_n)\otimes \1_{N_f}  &  \diag(p_n^*)\otimes J^\dagger_{R/L}
  \eep
  = \prod_{n} \Pf \bep
  J_{R/L}  &  -\xi_n 
  \\
  \xi_n  &  J^\dagger_{R/L}
  \eep,
\end{align}
where we used $p_np_n^*=1$. The contribution from zero modes can be
read off as
\begin{align}
  \label{eq:pf__zero}
  \prod_{n:\; \xi_n=0}\Pf \bep
  J_{R/L}  &  0 \\
  0  & J^\dagger_{R/L}
  \eep
  =\big[\Pf(J_R) \Pf(J_L^\dagger)\big]^{n_R}
  \big[\Pf(J_R^\dagger) \Pf(J_L)\big]^{n_L}\,,
\end{align}
where $n_{R,\,L}\geq 0$ denotes the number of zero modes of each
handedness.

The contribution from a nonzero mode is given by%
\footnote{Rigorously speaking, $\Pf$ and ${\det}^{1/2}$ may differ by
  a sign.  However, a $\nu$-independent multiplicative constant can
  safely be omitted without changing expectation values of
  observables.}
\begin{align}
  \label{eq:pf__nonzero}
  {\det}^{1/2}\bep J_{R/L} & -\xi_n \\ \xi_n & J_{R/L}^\dagger \eep 
  = \begin{cases}
     {\det}^{1/2}\big(\xi_n^2+J_R^\dagger J_R\big) &
     \text{for right-handed }\phi_n\,,
     \\
     {\det}^{1/2}\big(\xi_n^2+J_L^\dagger J_L\big) &
     \text{for left-handed }\phi_n\,,
   \end{cases} 
\end{align}
where we used the fact that the handedness of $\phi_n$ is opposite to
that of $\tilde\phi_n$ for $\xi_n\ne 0$. Since \eqref{eq:pf__nonzero}
is manifestly positive definite, the (non-) positivity of the measure
is determined by \eqref{eq:pf__zero}.

Summarizing, we find that the full partition function of two-color QCD
with diquark sources in the chiral limit reduces to the following
expression in terms of the singular values of $D(\mu)$,
\begin{align}
  Z(J_L, J_R) & = \bigg\langle 
  \big[\Pf(J_R) \Pf(J_L^\dagger)\big]^{n_R}
    \big[\Pf(J_R^\dagger) \Pf(J_L)\big]^{n_L}
  \notag
  \\
  & \qquad 
  \times {\prod_n}'{\det}^{1/2}\left(\xi_n^2+J^\dagger_RJ_RP_R
  +J^\dagger_LJ_L P_L \right)  \bigg\rangle_\text{YM}\,,
  \label{eq:app_main_result_beta=1}
\end{align}
where the primed product runs over all nonzero singular values.  The
above expression allows for accidental zero modes, but they are of
measure zero, see section~\ref{sec:ev}.  Therefore generically we have
$(n_R,\ n_L)=(\nu,\ 0)$ for $\nu\geq 0$ and $(0,\ -\nu)$ for $\nu<0$.
Note also that $\Pf(J^\dagger_R) \Pf(J_L) = \big[ \Pf(J_R)
\Pf(J^\dagger_L)\big]^*$.

\subsubsection{Short derivation}
\label{sc:beta=1_sloppy}

The derivation of the singular value representation of the partition
function in the last subsection is rigorous but lengthy. Here we will
give a less rigorous but shorter derivation of the same expression.

Let us assume $\mu=0$ so that the extended flavor symmetry $\SU(2N_f)$
is intact in the absence of diquark sources.  In computing $\Pf(W)$
let us separate the contributions from zero modes and those from
nonzero modes.  First, in the space of nonzero modes, we compute the
square of the Pfaffian in \eqref{eq:W}, which is equal to its
determinant,
\begin{align}
  \Pf'(W)^2
  = {\det}'\bep
    C\tau_2(J_RP_R+J_LP_L) & C\tau_2D^\dagger C\tau_2
    \\
    D & -C\tau_2(J_R^\dagger P_L+J_L^\dagger P_R)
  \eep,
\end{align}
where the prime on both sides indicates the omission of zero modes
and we used $-D^T = C\tau_2 D^\dagger C\tau_2$ in the $(1, 2)$ block.
If we interchange the first and the second column, there will be a
factor $(-1)^{d}$, with $d$ the total dimension of the space spanned
by the nonzero modes (we assume a suitable regularization such that
$d<\infty$). Since all nonzero modes are paired by $\gamma_5$, $d$ is
an even integer and $(-1)^d=1$. Thus
\begin{align}
  \Pf'(W)^2 & = {\det}'\bep
    C\tau_2D^\dagger C\tau_2 & C\tau_2(J_RP_R+J_LP_L)
    \\
     -C\tau_2(J_R^\dagger P_L+J_L^\dagger P_R) & D
  \eep
  \notag \\
  & = {\det}' \bep C\tau_2 & 0 \\ 0 & C\tau_2 \eep
  \bep
    D^\dagger C\tau_2 & J_RP_R+J_LP_L
    \\
     - J_R^\dagger P_L - J_L^\dagger P_R & C\tau_2 D
  \eep
  \notag \\
  & = {\det}'  
  \bep
    D^\dagger C\tau_2 & J_RP_R+J_LP_L
    \\
     - J_R^\dagger P_L - J_L^\dagger P_R & C\tau_2 D
  \eep .
  \intertext{Using the formula $\displaystyle \det \bep A &
    B\\C&D\eep=\det(AD-ACA^{-1}B)$ that holds when the blocks are
    square matrices of the same dimension and $A$ is invertible, we
    obtain} 
  \Pf'(W)^2& = {\det}'\left( D^\dagger D + D^\dagger C\tau_2 
  (J_R^\dagger P_L + J_L^\dagger P_R)(D^\dagger C\tau_2)^{-1}(J_RP_R+J_LP_L)  \right)
  \notag \\
  & = {\det}' ( D^\dagger D + J_R^\dagger J_R P_R + J_L^\dagger J_LP_L )\,.
  \label{eq:det_nonzeromodes}
\end{align}
We stress that the above formula cannot be used in the full eigenspace
of $D$ because $D^{-1}$ does not exist if there are zero modes. Since
\eqref{eq:det_nonzeromodes} is positive definite we can take its
square root naively to obtain a formula for $\Pf'(W)$.

The next task is to incorporate zero modes (where from now on we
ignore accidental zero modes). We do this by switching from quark
masses to diquark sources.  Let us recall that the mass term in
\eqref{eq:Lf} can be written as \cite{Kogut:2000ek}
\begin{align}
  \label{eq:mass_Psi_}
  \bar\psi (MP_L+M^\dagger P_R)\psi = \frac{1}{2}\Psi^\dagger
  \sigma_2\tau_2 \bep 0 & M \\ -M^T & 0 \eep\Psi^* + \text{h.c.} 
\end{align}
with the ``extended'' spinor $\displaystyle \Psi \equiv \bep \psi_R \\
\sigma_2\tau_2 \psi^*_L \eep$.\footnote{The symbol $\Psi$ used here is
  not to be confused with $\Psi$ defined in \eqref{eq:Psi}.}  As is
well known \cite{Leutwyler:1992yt}, the contribution of zero modes to
the fermion determinant in the absence of diquark sources is $(\det
M^\dagger)^\nu$ for $\nu\geq 0$ and $(\det M)^{-\nu}$ for $\nu<0$.%
\footnote{\label{ft:conv_diff}This differs from
  \cite{Leutwyler:1992yt,Smilga:1994tb}.  The reason seems to be that
  our $\gamma_5$ differs from their $\gamma_5$ by a minus sign.  (The
  definitions of the mass term and of $\nu$ in \eqref{eq:nu_definition}
  are identical.)  Thus $\nu=n_R-n_L$ in this paper, whereas
  $\nu=n_L-n_R$ in \cite{Leutwyler:1992yt,Smilga:1994tb}.  One should
  be careful when comparing results in this paper with those in
  \cite{Leutwyler:1992yt,Smilga:1994tb}.}  This expression is not
desirable, as it does not manifestly show the spurious invariance
under $\SU(2N_f)$.%
\footnote{For a fixed quark mass, $\SU(2N_f)$ is broken to
  $\U(1)_B\times\SU(N_f)_R\times\SU(N_f)_L$, but the initial symmetry
  is kept intact if we formally transform the quark mass as a spurion
  field, see \eqref{eq:M_pf_sp}.}  Instead, we write $\det M$ as a
Pfaffian \cite[Eq.~(4.13)]{Smilga:1994tb},
\begin{align}
  \det M = (-1)^{N_f/2}\Pf \bep 0 & M \\ -M^T & 0 \eep,
\end{align}
which makes the invariance under $U\in \SU(2N_f)$ explicit, i.e., 
\begin{align}
  \label{eq:M_pf_sp}
  \Pf\left[U \bep 0 & M \\ -M^T & 0 \eep U^T \right] 
  = \det U \cdot \Pf \bep 0 & M \\ -M^T & 0 \eep 
  = \Pf \bep 0 & M \\ -M^T & 0 \eep\,.
\end{align}
Similarly the diquark source in \eqref{eq:Lf} can be cast in the form
\begin{align}
  \frac{1}{2}\psi^T C\tau_2(J_RP_R+J_LP_L)\psi+\text{h.c.} 
  = \frac{1}{2}\Psi^\dagger \sigma_2\tau_2 \bep -J_R^\dagger & 0 \\ 0
  & J_L \eep\Psi^* + \text{h.c.} 
\end{align}
Comparing this with \eqref{eq:mass_Psi_} we notice that the diquark
source is obtained if we simply replace $\displaystyle \bep 0 & M \\
-M^T & 0 \eep$ by $\displaystyle \bep -J_R^\dagger & 0 \\ 0 & J_L
\eep$.  Therefore the contribution of zero modes (for $\nu<0$) in the
chiral limit and with nonzero diquark sources can be found by the
replacement%
\footnote{A precise account of this procedure goes as follows. We
  start with a certain diquark source in the chiral limit. By a
  suitable $\SU(2N_f)$ transformation we can rotate it into the form
  of the mass term, for which we know that the zero modes contribute
  $(\det M)^{-\nu}$ to the fermion determinant. Then we rotate
  inversely, bringing the mass term back to our original diquark
  source.  This is how we get \eqref{eq:J_replace} so quickly.}
\begin{align}
  &(\det M)^{-\nu} = 
  (-1)^{\nu N_f/2}\Pf \bep 0 & M \\ -M^T & 0 \eep^{-\nu} \notag\\
  &\longrightarrow\quad (-1)^{\nu N_f/2}\Pf \bep -J_R^\dagger & 0 \\ 0
  & J_L \eep^{-\nu}  
  =  \big[ \Pf(J_R^\dagger)\Pf(J_L)\big]^{-\nu}\,,
  \label{eq:J_replace}
\end{align}
and similarly for $\nu\geq 0$.  Combining the square root of
\eqref{eq:det_nonzeromodes} with \eqref{eq:J_replace} we finally
obtain for the partition function
\begin{align}
  \label{eq:quick_Z_beta_1}
  Z(J_L, J_R) = 
  \left\langle
    \begin{Bmatrix}
      \big[ \Pf(J_R)\Pf(J_L^\dagger) \big]^{\nu}  \\
      \big[ \Pf(J_R^\dagger)\Pf(J_L) \big]^{-\nu}  
    \end{Bmatrix}
    \sqrt{ {\det}'  ( 
    D^\dagger D + J_R^\dagger J_R P_R + J_L^\dagger J_LP_L ) }
    \right\rangle_\text{YM},
\end{align}
where the first (second) line in curly braces applies to $\nu\ge0$
($\nu<0$).  Assuming that $\mu\ne 0$ does not change the form of this
expression, we arrive at \eqref{eq:app_main_result_beta=1} in the
previous subsection.  Note that this argument hinges on the assumption
that the contribution of exact zero modes at $\mu\ne 0$ is the same as
at $\mu=0$. We also point out that, since $D=-D^\dagger$ at $\mu=0$,
we could also have written $D^\dagger D$ appearing in
\eqref{eq:quick_Z_beta_1} as $-D^2$, but this would not have
generalized to $\mu\ne0$.

\subsection[QCD with isospin chemical potential $(\beta=2)$]
{\boldmath QCD with isospin chemical potential $(\beta=2)$}
\label{app:iso}

Here we only sketch the shorter derivation of the singular value
representation of the partition function, i.e., the extension of the
argument in section~\ref{sc:beta=1_sloppy} to $\beta=2$, and omit the
lengthy and rigorous version.

First, working along similar lines as in
section~\ref{sc:beta=1_sloppy}, the contribution from nonzero modes to
the partition function can be shown to take the form
\begin{align}
  {\det}'  (D^\dagger D + \rho \rho^*P_R + \lambda \lambda^*P_L )\,.
\end{align}

Next we consider the zero modes. In order to figure out the mapping
between quark masses and pionic sources, we go to a new basis
\begin{align}
  \Psi_R \equiv P_R\bep u\\d \eep = \bep u_R \\ d_R \eep
  \qquad\text{and}\qquad 
  \Psi_L \equiv P_L\bep u\\d \eep = \bep u_L \\ d_L \eep\,.
\end{align}
These are the $\beta=2$ counterparts of the extended spinor $\Psi$
introduced after \eqref{eq:mass_Psi_} for $\beta=1$. Then the mass
term for $M=\diag(m_u,m_d)$ is
\begin{align}
  \bar u(m_uP_L+m_u^*P_R)u + \bar d(m_dP_L+m_d^*P_R)d 
  = \bar\Psi_R M \Psi_L + \bar\Psi_LM^\dagger \Psi_R\,.
\end{align}
The zero-mode contribution to the fermion determinant is given by
$(\det M)^{-\nu}$ for $\nu<0$ and $(\det M^\dagger)^{\nu}$ for $\nu
\geq 0$. This form is manifestly invariant under the spurious
$\SU(2)_R\times\SU(2)_L$ rotation of $M$.  On the other hand, the
pionic source term reads
\begin{align}
  \bar u(\lambda^* P_R+\rho^*P_L)d + \bar d(\rho P_R+\lambda P_L)u
  = 
  \bar\Psi_R \bep 0 & \rho^* \\ \lambda & 0 \eep \Psi_L 
  + \bar\Psi_L \bep 0 & \lambda^* \\ \rho & 0 \eep \Psi_R\,.
\end{align}
Comparing with the mass term, we obtain the correspondence
\begin{align}
  M \longleftrightarrow \bep 0 & \rho^* \\ \lambda & 0 \eep
  \qquad\text{and}\qquad
  M^\dagger \longleftrightarrow 
 \bep 0 & \lambda^* \\ \rho & 0 \eep,
\end{align}
i.e., there is a mass term that can be rotated to a given pionic
source term under the action of $\SU(2)_R\times \SU(2)_L$.  Thus the
contribution of zero modes is given by
\begin{subequations}
  \begin{alignat}{2}
    \det\bep 0 & \lambda^* \\ \rho & 0 \eep^\nu 
    &= (-\rho\lambda^*)^\nu\quad & \text{for } \nu\geq 0\,,\\
    \det\bep 0 & \rho^* \\ \lambda & 0 \eep^{-\nu} 
    &= (-\rho^*\lambda)^{-\nu}\quad & \text{for } \nu<0 \,.
  \end{alignat}
\end{subequations}
Summarizing, we find for the partition function
\begin{align}
  \label{eq:app_main_result_beta=2;}
  Z(\rho,\lambda) = \left\langle
  \begin{Bmatrix}
    (-\rho\lambda^*)^{\nu}\\
    (-\rho^*\lambda)^{-\nu}
  \end{Bmatrix}
  {\det}'  (D^\dagger D + \rho \rho^*P_R + \lambda \lambda^*P_L )
  \right\rangle_\text{YM}
  \quad\text{for}\quad\begin{Bmatrix}
    \nu\ge0 \\ \nu<0
  \end{Bmatrix}\,.
\end{align}
In deriving \eqref{eq:app_main_result_beta=2;} we ignored accidental
zero modes.  The result can straightforwardly be extended to include
them, and we then obtain \eqref{eq:app_main_result_beta=2}.

\subsection[QCD with adjoint fermions $(\beta=4)$]{\boldmath QCD with
  adjoint fermions $(\beta=4)$}
\label{app:real}

Again we only sketch the shorter and less rigorous derivation.  First,
we consider the contribution of nonzero modes.  Along similar lines as
for $\beta=1$ and $2$ we obtain
\begin{align}
  \label{eq:det'_beta=4}
  \sqrt{{\det}' (D^\dagger D+J_R^\dagger J_RP_R + J_L^\dagger J_LP_L ) }\,,
\end{align}
but this is not the end of the story. Using \eqref{eq:symm_adj_i} we
can show $(D^\dagger D)^*=CD^\dagger DC$, from which it follows that a
state $C\phi_n^*$ associated with an eigenstate $\phi_n$ of $D^\dagger
D$ is also an eigenstate of $D^\dagger D$ with the same eigenvalue. As
$\phi_n$ and $C\phi_n^*$ are linearly independent because of
$(CK)^2=-\1$, all eigenvalues of $D^\dagger D$ are doubly degenerate
(including zero modes). This allows us to take the square of
\eqref{eq:det'_beta=4},
\begin{align}
  {\det}'' 
  (D^\dagger D+J_R^\dagger J_RP_R + J_L^\dagger J_LP_L )\,,
\end{align}
where in ${\det}''$ each degenerate eigenvalue of $D^\dagger D$ is
counted only once.

Next we consider the zero-mode contributions. With the extended
spinor \smash{$\displaystyle \Psi \!\equiv\! \bep \psi_R \\
  i\sigma_2\psi^*_L \eep$} of length $2N_f$,\footnote{Again, the
  symbol $\Psi$ used here should not be confused with the $\Psi$
  defined in \eqref{eq:Psi} or after \eqref{eq:mass_Psi_}.}  the mass
term becomes
\begin{align}
  \label{eq:4_mass}
  \bar\psi (MP_L+M^\dagger P_R)\psi = -
  \frac{1}{2}\Psi^\dagger \sigma_2 \bep 0 & iM \\ iM^T & 0 \eep\Psi^* - 
  \frac{1}{2}\Psi^T\sigma_2\bep0&-iM^*\\-iM^\dagger &0\eep\Psi \,.
\end{align}
The contribution of the zero modes to the fermion determinant in the
absence of diquark sources is given by
\cite{Leutwyler:1992yt,Smilga:1994tb}
\begin{align}
  \det\bep 0 & iM \\ iM^T & 0 \eep ^{-\bar\nu}
  \quad\text{for }\bar\nu<0\quad\text{and}\quad
  \det\bep 0 & -iM^* \\ -iM^\dagger & 0 \eep^{\bar\nu}
  \quad\text{for }\bar\nu\geq 0\,,
\end{align} 
where $\bar\nu \equiv (n_R-n_L)/2$.%
\footnote{Note that $\bar\nu$ in \cite{Leutwyler:1992yt,Smilga:1994tb}
  is $-\bar\nu$ in our notation, see footnote \ref{ft:conv_diff}.}
Because of the degeneracy mentioned above the numbers $n_R$, $n_L$ of
zero modes are even integers so that $\bar\nu\in\Z$.  The number
$\bar\nu$ is proportional to $\nu$, the winding number of the gauge
field. The proportionality constant depends on the color
representation of the fermions. For example, $\bar\nu=N_c\nu$ for
fermions in the adjoint representation of $\SU(N_c)$
\cite{Leutwyler:1992yt}.  More general cases are considered, e.g., in
\cite{Creutz:2007yr}.

The diquark source term can be written as
\begin{align}
  \frac{1}{2}\psi^T C(J_RP_R+J_LP_L)\psi+\text{h.c.} 
  = - \frac{1}{2}\Psi^\dagger \sigma_2\bep J_R^\dagger & 0 \\ 0 & -J_L 
  \eep\Psi^*
  - \frac{1}{2}\Psi^T \sigma_2\bep J_R& 0 \\ 0 & -J_L^\dagger  
  \eep\Psi \,.
\end{align}
Comparison with \eqref{eq:4_mass} again suggests a correspondence
between diquark matrix and mass matrix, from which the contribution of
the zero modes in the chiral limit and with nonzero diquark sources is
given by
\begin{subequations}
  \begin{alignat}{2}
    \det\bep J_R & 0 \\ 0 & -J_L^\dagger \eep^{\bar\nu}
    &= \det (-J_RJ_L^\dagger )^{\bar\nu} \quad
    & \text{for }\bar\nu \geq 0\,,\\
    \det\bep J_R^\dagger & 0 \\ 0 & -J_L \eep^{-\bar\nu}
    &= \det (-J_R^\dagger J_L )^{-\bar\nu}
    & \text{for }\bar\nu<0\,.
  \end{alignat}
\end{subequations}
Collecting everything, we find for the partition function
\begin{align}
  \label{eq:app_main_result_beta=4;}
  Z(J_L, J_R) = \left\langle 
  \begin{Bmatrix}
    \det(-J_RJ_L^\dagger )^{\bar\nu}\\
    \det(-J_R^\dagger J_L )^{-\bar\nu}
  \end{Bmatrix}
  {\det}'' (D^\dagger D+J_R^\dagger J_RP_R + J_L^\dagger J_LP_L )
  \right\rangle_\text{YM}
  \;\;\text{for}\;\;
  \begin{Bmatrix}\bar\nu\ge0 \\ \bar\nu<0\end{Bmatrix}\,.
\end{align}
Again, in deriving \eqref{eq:app_main_result_beta=4;} we ignored
accidental zero modes.  The result can straightforwardly be extended to
include them, and we then obtain \eqref{eq:app_main_result_beta=4}.

\section{Consequences of a (non-) positive measure}
\label{app:complex}
\subsection{Diquark sources and positivity}
\label{app:complex-1}

From the results of appendix~\ref{app:Z-pfaffian} we can immediately
read off conditions for the positivity of the measure.  Here we
consider some cases of particular physical interest (assuming
$\theta=0$).

\begin{itemize}
\item $\beta=1$: From \eqref{eq:quick_Z_beta_1} it follows that
\begin{enumerate}
\item The choice $J_R=-J_L$ (source for the $0^+$ diquark condensate)
  implies
  \begin{align}
    \Pf(J_R) \Pf(J_L^\dagger)=\Pf(J_R)\Pf(-J_R^\dagger)
    =(-1)^{N_f}|\Pf(J_R)|^2=|\Pf(J_R)|^2>0\,. 
  \end{align}
  Hence the measure is positive definite and there is no sign problem.
\item The choice $J_R=J_L$ (source for the $0^-$ diquark condensate)
  implies
  \begin{align}
    \Pf(J_R) \Pf(J_L^\dagger)=\Pf(J_R) \Pf(J_R^\dagger) 
    =(-1)^{N_f/2}|\Pf(J_R)|^2\,. 
  \end{align}
  This is positive (negative) if $N_f=4n$ ($N_f=4n+2$) with $n\in\N$.
  Thus the sign problem arises for $N_f=4n+2$ due to the topological
  sectors with odd $\nu$.
\end{enumerate}
\item $\beta=2$: From \eqref{eq:app_main_result_beta=2;} it follows
  that
  \begin{enumerate}
  \item The choice $\rho = -\lambda$ (source for the $0^+$ pion
    condensate) implies
    \begin{align}
      -\rho\lambda^* = \rho \rho^* >0
    \end{align}
    so that the measure is positive definite and there is no sign
    problem.
  \item The choice $\rho = \lambda$ (source for the $0^-$ pion
    condensate) implies
    \begin{align}
      -\rho\lambda^* = - \rho \rho^* <0\,.
    \end{align}
    Thus the sign problem arises due to the topological sectors with
    odd $\nu$.
  \end{enumerate}
\item $\beta=4$: From \eqref{eq:app_main_result_beta=4;} it follows
  that
  \begin{enumerate}
  \item The choice $J_R=-J_L$ (source for the $0^+$ diquark
    condensate) implies
    \begin{align}
      \det(-J_R J_L^\dagger)=\det(J_R J_R^\dagger) >0 
    \end{align}
    so that the measure is positive definite and there is no sign
    problem.
  \item The choice $J_R=J_L$ (source for the $0^-$ diquark condensate)
    implies
    \begin{align}
      \det(-J_R J_L^\dagger)=(-1)^{N_f}\det(J_R J_R^\dagger)\,.
    \end{align}
    This is positive (negative) for even (odd) $N_f$,%
    \footnote{Note that $N_f$ is the number of Dirac fermions
      (not Majorana fermions).}  so the sign problem arises for odd
    $N_f$ due to the topological sectors with odd $\bar \nu$.
  \end{enumerate}
\end{itemize}

\subsection{Sign problem and Dirac spectrum}
\label{app:sign}

\subsubsection{Discussion by Leutwyler and Smilga}

A long time ago it was pointed out by Leutwyler and Smilga
\cite[Sec.~VII]{Leutwyler:1992yt} that the positivity of the measure
is a necessary condition for the Banks-Casher relation
\cite{Banks:1979yr} to hold.  They used $N_f=1$ QCD at $\mu=0$ as an
example.  Here we review their discussion briefly and motivate our
microscopic analysis in the next subsection.

A standard derivation of the Banks-Casher relation at $\mu=0$ goes as
follows.  Assuming $\theta=0$ for simplicity and taking $m$ to be real
(but not necessarily positive), we have
\begin{align}
  \langle\bar\psi\psi\rangle_{N_f=1}^{}
  & = \frac{1}{V_4}\frac{\del}{\del m}\log 
  \Big\langle m^{|\nu|}{\prod_n}' (\lambda_n^2+m^2) \Big\rangle
  _\text{YM}
  \notag\\
  & = \frac{1}{V_4m}\langle|\nu|\rangle_{N_f=1} 
  + \frac{1}{V_4}\Big\langle{\sum_n}'\frac{2m}{\lambda_n^2+m^2}
  \Big\rangle_{N_f=1}
  \notag\\
  & = \frac{1}{V_4m}\langle|\nu|\rangle_{N_f=1} 
  + \int_0^{\infty} 
  d\lambda\,\rho(\lambda)\,\frac{2m}{\lambda^2+m^2}\,,
  \label{eq:BC_1-flavor_}
\end{align}
where the $i\lambda_n$ are the Dirac eigenvalues, the primed sum
stands for the summation over nonzero modes, and
\begin{align}
  \rho(\lambda) = \frac{1}{V_4}
  \Big\langle {\sum_n}' \delta(\lambda-\lambda_n) \Big\rangle_{N_f=1}\,.
\end{align}
Note that $\rho(\lambda)$ depends on $m$ through the fermion
determinant in the measure.  The cases $m>0$ and $m<0$ differ
qualitatively, as we shall see now.
\begin{itemize}
\item $m>0$: Since $\langle \nu^2 \rangle_{N_f=1} =mV_4\Sigma$ with
  $\Sigma$ a low-energy constant, the first term in
  \eqref{eq:BC_1-flavor_} is suppressed as $O(1/\sqrt{V_4})$ and
  becomes negligible at large volume.  The second term in
  \eqref{eq:BC_1-flavor_} will reduce to $\pi \rho(0)$ as $m\to 0^+$
  after the thermodynamic limit. Therefore
  \begin{align}
    \langle \bar\psi\psi\rangle^{}_{N_f=1} = \pi \rho(0)\,.
    \label{eq:BC_Nf=1:}
  \end{align}
\item $m<0$: Since a probabilistic interpretation is no longer
  possible one cannot drop the first term of \eqref{eq:BC_1-flavor_}.
  Indeed we have \cite{Leutwyler:1992yt}
  \begin{align}
    \frac{1}{V_4m}\langle|\nu|\rangle_{N_f=1} 
    \sim 
    \frac{\Sigma}{2\sqrt{2\pi}}\frac{\ee^{2|x|}}{|x|^{3/2}}
    \to \infty 
    \quad \text{as } x=mV_4\Sigma\to\infty\,.
    \label{eq:+inf}
  \end{align}
  Similarly \smash{$\displaystyle \int_0^\infty
    d\lambda\,\rho(\lambda)\, \frac{2m}{\lambda^2+m^2}$} diverges to
  $-\infty$, but these divergences cancel neatly in
  \eqref{eq:BC_1-flavor_}.
\end{itemize}
Now it is clear why the Banks-Casher relation 
\eqref{eq:BC_Nf=1:} fails for $m<0$:
\begin{itemize}
  \item 
  Zero-mode contributions are not suppressed at all 
  in the thermodynamic limit. Nontrivial topologies cannot be ignored.
  \item 
  The macroscopic spectral density $\rho(\lambda)$ is ill-defined 
  in the thermodynamic limit at least near the origin.
\end{itemize}
It seems natural to expect that these findings also apply to dense
QCD-like theories, where the diquark (or pionic) sources can spoil the
positivity of the measure (see section~\ref{app:complex-1}). In a
measure with indefinite sign, $\rsv(\xi)$ would be singular in the
thermodynamic limit, and we cannot derive Banks-Casher-type relations
for the diquark and pion condensate.  This is the reason why in the
main text we chose only those sources that respect the positivity of
the measure.  An exceptional case is the limit $\mu=\infty$, where
instantons are completely suppressed and the exact zero modes will
disappear. For this case we can in principle choose any diquark/pionic
sources.

We conclude this subsection with a few important remarks.
\begin{itemize}
\item Sometimes there are multiple external fields that couple to
  different condensates and can be inserted without spoiling the
  positivity.  For instance, in $N_f=4$ QCD at $\mu=0$, the
  Banks-Casher relation can be derived for both
  $\ev{\bar\psi_f\psi_f}$ and $\ev{\bar\psi_f i\gamma_5\psi_f}$
  because the degenerate \emph{purely imaginary} mass term also
  respects positivity.%
  \footnote{For $N_f=2$ this is not the case. This fact does not seem
    to have been appreciated in \cite{Bitar:1997ic}.}  However,
  $\rho(0)$ itself can be measured on the lattice without using any
  external sources. To which condensate $\rho(0)$ is related to is
  determined by the infinitesimal external field we
  select.\footnote{More precisely, we have
    $|\ev{\bar\psi\psi}|=|\ev{\bar\psi i\gamma_5\psi}|=\pi\rho(0)$,
    and the orientation of the condensate is determined by the
    external field.  See also footnote~\ref{fn:ising}.}
\item The failure of a Banks-Casher relation does not prove the
  absence of the condensate%
  \footnote{Indeed $\langle\bar qq\rangle_{N_f=1}$ is the same for
    $m>0$ and $m<0$.  }  because it is possible that the condensate
  exists but is not determined by the value of $\rho(\lambda)$ at the
  origin.  An instructive example is three-color QCD at small $\mu$.
  It supports a nonzero $\langle\bar \psi\psi\rangle$, but there is no
  Banks-Casher relation.  The microscopic spectral density of the
  complex Dirac eigenvalues exhibits a drastic oscillation with a
  period $\sim O(1/V_4)$ and an amplitude $\sim \exp(V_4)$
  \cite{Osborn:2005ss,Osborn:2008jp}.  This suggests that the
  \emph{macroscopic} spectral density in the limit $V_4\to\infty$ is
  singular at least near the origin.
\end{itemize}

In the discussion by Leutwyler and Smilga, a sick behavior of
$\rho(\lambda)$ for $N_f=1$ with $m<0$ was alleged in an indirect
way. In the next subsection we study the spectral density in the
microscopic limit ($\epsilon$-regime) and explicitly show how it
behaves in an indefinite measure.

\subsubsection{The microscopic limit}

The effect of the topologically nontrivial sectors $(\nu\ne 0)$ on the
Dirac spectrum was studied in detail in \cite{Damgaard:1999ij}, where
the microscopic spectral density and the chiral condensate in full QCD
(including all topologies) were analyzed.  All explicit examples in
that reference were worked out for positive quark mass and vanishing
CP-breaking angle $\theta=0$. Here we demonstrate the effect of
negative mass, or equivalently nonzero $\theta$, on the microscopic
Dirac spectrum.  As in \cite{Damgaard:1999ij}, we use a formula
\cite{Damgaard:1997tr} that expresses the microscopic correlation
functions in terms of finite-volume partition functions.  While in
\cite{Damgaard:1999ij} the contributions from different topological
sectors were simply summed up numerically, we here use a closed
expression for the full partition function summed over all topologies
that became available later \cite{Lenaghan:2001ur}.  By doing this we
avoid the errors coming from the truncation of an infinite to a finite
sum.

\begin{figure}[t]
  \centering
  \includegraphics{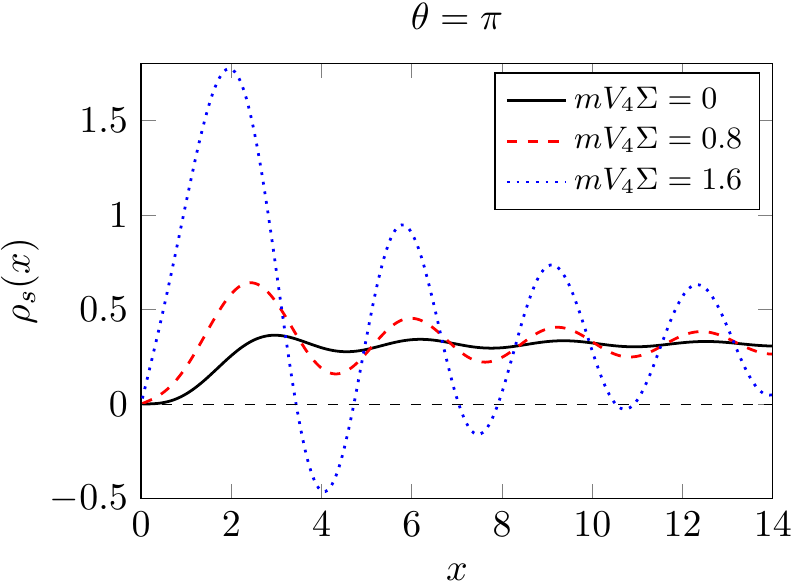}
  \caption{ \label{fig:Nf=1_oscillation} Microscopic spectral density
    of the Dirac eigenvalues in $N_c\geq 3$ QCD ($\beta=2$) for
    $N_f=1$ at $\theta=\pi$ for different quark masses.  Here
    $\rho_s(x)$ is not necessarily positive definite and the
    probabilistic meaning as a ``density'' of eigenvalues is lost,
    reflecting the underlying sign problem.  }
\end{figure}

Since all relevant formulas can be found in the literature
\cite{Damgaard:1997tr,Damgaard:1999ij,Lenaghan:2001ur} we omit the
technical details and only show the final plot.  As the simplest case
we considered $N_f=1$ QCD at $\mu=0$ with $N_c\geq 3$ ($\beta=2$).  We
have chosen $\theta=\pi$ and $m>0$ so that the measure is not positive
definite. (This is physically equivalent to $\theta=0$ with $m<0$.)
In figure \ref{fig:Nf=1_oscillation} we show the plot of the
microscopic spectral density with several values of the quark mass.
In the chiral limit $(m=0)$ the contribution from the $\nu=0$ sector
is dominant and there is no effect of $\theta$. For nonzero quark
mass, the microscopic spectral density exhibits a strong oscillation
whose magnitude grows rapidly with $mV_4\Sigma$, while the period of
oscillation looks roughly independent of $mV_4\Sigma$. This is exactly
the opposite of what we observe at $\theta=0$, where the dynamical
quark decouples for a sufficiently large mass and $\rho_s(x)$ becomes
smoother. This oscillatory behavior and the failure of decoupling is
reminiscent of the complex Dirac eigenvalue spectrum in QCD at $\mu\ne
0$ \cite{Osborn:2005ss}. We also checked that these characteristics
are not limited to $\theta=\pi$: Even for a small $\theta\ne 0$ the
violent oscillation shows up eventually at a sufficiently large value
of $mV_4\Sigma $.

Now we have understood microscopically why the standard Banks-Casher
relation can fail: The spectral density varies rapidly over the scale
of the quark mass and $\rho(0)$ is certainly ill-defined in the
thermodynamic limit. It is quite tempting to conjecture that this
phenomenon is universal and will occur also in other theories with
indefinite measure. In particular, we expect this for dense QCD-like
theories with positivity-breaking external sources.%
\footnote{The $\theta$-dependence of QCD-like theories at $\mu\ne 0$
  was studied in \cite{Metlitski:2005db,Metlitski:2005di}, but the
  singular value spectrum of the Dirac operator was not considered
  there.}  Indeed, in $N_f=2$ two-color QCD, $j_R=j_L$ at $\theta=0$
is equivalent to $j_R=-j_L$ at $\theta=\pi$, and the similarity to the
case considered above is evident.  It would be interesting to study
the impact of the sign problem on the singular value spectrum in more
detail in analogy to the analysis of
\cite{Osborn:2005ss,Osborn:2008jp}.

\subsection{QCD inequalities}
\label{app:qcd_ineq}

QCD inequalities are a powerful theoretical tool to impose strong
constraints on the dynamics when the Euclidean functional measure is
positive definite
\cite{Weingarten:1983uj,Witten:1983ut,Nussinov:1983hb,Espriu:1984mq}.
Although the sign problem hampers the application of QCD inequalities
to three-color QCD at $\mu\ne 0$, they are still applicable to a class
of dense QCD-like theories with positive definite measures
\cite{Kogut:1999iv,Son:2000xc,Son:2000by,Splittorff:2000mm}.  In
particular it was shown for two-color QCD (with nonzero quark masses
and no diquark sources) that the symmetry breaking is driven by the
$0^+$ diquark condensate \cite{Kogut:1999iv}, thus leaving parity
unbroken. In this appendix we reexamine the application of QCD
inequalities to two-color QCD at $\mu\ne 0$ and extend the analysis to
the case of nonzero diquark sources.  For the sake of simplicity we
will limit ourselves to either nonzero quark masses or nonzero diquark
sources, but not both. We always work in Euclidean space.  

\subsubsection{Two-color QCD with quark masses}

We review the original discussion in \cite{Kogut:1999iv} and assume
even $N_f$ for positivity of the measure.  In this subsection, the
transpose ($T$) and the adjoint ($\dagger$) act only on color, spinor,
and flavor indices, but not on space-time indices.  The quark chemical
potential $\mu$ can be arbitrary in the following.  When the diquark
sources are absent and the (degenerate) quark mass is real and
positive, the propagator in a fixed background gauge field is given by
\begin{equation} 
  S_{\psi\bar\psi}(x,y) \equiv \ev{\psi(x)\bar\psi(y)}_\psi 
  = \Big\langle x\Big| \frac{1}{D+m} \Big|y\Big\rangle\,, 
\end{equation} 
where $\langle~\rangle_\psi$ denotes the average only over the fermion
fields.  Below we also use $\langle~\rangle_{\psi, A}$ and
$\langle~\rangle_{A}$ for the full average and the average only over
the gauge fields, respectively.
 
Let us consider a diquark operator $M(x)\equiv \psi^T \Gamma\psi$ with
an antisymmetric matrix $\Gamma$ satisfying $\Gamma^\dagger
\Gamma=\1$.  The correlation function of this field is then given by
\begin{align}
  \ev{M(x)M^\dagger (y)}_{\psi, A} & = \ev{
  \psi^T(x)\Gamma\psi(x) \bar\psi(y)\bar\Gamma\,\bar\psi^T\!(y)
  }_{\psi, A} \qquad (\bar\Gamma \equiv \gamma_4\Gamma^\dagger
  \gamma_4)
  \notag \\
  & = \ev{\tr\big[\Gamma S_{\psi\bar\psi}(x,y)\bar\Gamma
  S_{\psi\bar\psi}(x,y)^T\big]}_{A} 
  \notag \\
  & \leq \ev{\tr\big[
  S_{\psi\bar\psi}(x,y)S_{\psi\bar\psi}(x,y)^\dagger\big]}_{A}\,, 
  \label{eq:CS}
\end{align}
where the last line follows from the Cauchy-Schwarz inequality and the
positivity of the measure.  The same upper bound can be proven for
those correlators of mesonic fields $\bar\psi\Gamma\psi$ that have no
disconnected piece.%
\footnote{By ``disconnected piece'' we mean a diagram whose quark
  lines are not connected directly, although they could be connected
  by gluons, see, e.g., figure~1 in \cite{Witten:1983ut}. The absence
  of disconnected pieces is essential for QCD inequalities. The
  correlators of $\sigma$ and $\eta'$ contain disconnected pieces and
  our arguments do not apply. Indeed, if the disconnected piece could
  be dropped, the $\eta'$ would be as light as the pions. This is of
  course invalid.  In reality the contribution of the gluonic
  intermediate state $|F\tilde{F}\rangle$ is essential in the
  generation of the large mass of the $\eta'$.}  On the other hand,
$(D+m)^T=C\tau_2\gamma_5(D^\dagger+m)C\tau_2\gamma_5$ allows us to
write
\begin{equation}
  \ev{ M(x)M^\dagger (y) }_{\psi, A}
  = \ev{\tr\big[\Gamma S_{\psi\bar\psi}(x,y)\bar\Gamma C\tau_2\gamma_5 
  S_{\psi\bar\psi}(x,y)^\dagger C\tau_2\gamma_5\big]}_{A}\,.
\end{equation}
If we assume $\Gamma=C\tau_2\gamma_5\Gamma_f$ for an antisymmetric
flavor matrix $\Gamma_f$, it follows that
\begin{equation}
  \ev{ M(x)M^\dagger (y) }_{\psi, A} = \tr[\Gamma_f \Gamma_f^\dagger]
  \ev{\tr\big[S_{\psi\bar\psi}(x,y)S_{\psi\bar\psi}(x,y)^\dagger\big]}_{A}\,.
\end{equation}
This is equal to the upper bound in \eqref{eq:CS} up to an irrelevant
multiplicative constant,\footnote{The mass is given by the exponential
  decay of the correlation function, for which the multiplicative
  constant is irrelevant.} implying that the diquarks in this channel
are the lightest of all mesons and diquarks.%
\footnote{Strictly speaking, this argument does not exclude the
  possibility that there are other mesons or diquarks of the same
  mass.  For example, at $\mu=0$ there are mesonic states that have
  the same mass as the scalar diquarks due to the extended flavor
  symmetry.  This degeneracy is lifted at $\mu\ne 0$.}  For
$\Gamma=C\tau_2\Gamma_f$ the upper bound is not saturated, implying
that if a diquark condensate forms, it does so in the $0^+$ channel
and not in the $0^-$ channel.  This is consistent with the observation
that the instanton-induced interaction for $N_f=2$ is attractive in
the $0^+$ diquark channel and repulsive in the $0^-$ diquark channel
\cite{Alford:1997zt,Rapp:1997zu,Rapp:1999qa}.

It is a general rule that the preferred direction of a condensate
depends on the direction of the external symmetry-breaking field. In
the case above, we considered real positive quark masses, but it is
also instructive to consider other cases. Let us take degenerate
purely imaginary masses for instance. For the sake of positivity
$(\det M>0)$ we require $N_f=4n$ with $n\in\N$. The propagator now
reads
\begin{equation}
  S_{\psi\bar\psi}(x,y) = \ev{ \psi(x)\bar\psi(y) }_\psi 
  = \Big\langle x\Big| \frac{1}{D+im\gamma_5} \Big|y \Big\rangle
  \quad\text{with }m\in\R\,,
\end{equation}
and because of $(D+im\gamma_5)^T=-C\tau_2(D+im\gamma_5)^\dagger
C\tau_2$ the propagator satisfies $S_{\psi\bar\psi}(x,y)^T = -C\tau_2
S_{\psi\bar\psi}(x,y)^\dagger C\tau_2$.  Thus we find from the second
line of \eqref{eq:CS} that
\begin{align}
  \ev{ M(x)M^\dagger (y) }_{\psi, A}= 
  - \ev{\tr\big[\Gamma S_{\psi\bar\psi}(x,y)\bar\Gamma 
  C\tau_2 S_{\psi\bar\psi}(x,y)^\dagger C\tau_2\big]}_{A}\,.
\end{align}
When $\Gamma=C\tau_2\Gamma_f$, the r.h.s.\ reduces to
$\tr[\Gamma_f\Gamma_f^\dagger]
\big\langle\tr\big[S_{\psi\bar\psi}(x,y)
S_{\psi\bar\psi}(x,y)^\dagger\big]\big\rangle_{A}$, saturating the
inequality \eqref{eq:CS}.  Therefore this time the diquark
condensation occurs in the $0^-$ channel, breaking parity.  The
conclusion is that the quantum numbers of the condensate are not
entirely determined by the internal dynamics of the system, but are
sensitive to the external symmetry-breaking fields.

The reader may worry that this is in conflict with the Vafa-Witten
theorem \cite{Vafa:1984xg} stating that parity is not spontaneously
broken in vector-like theories.  Let us make two remarks on this
point. First, various loopholes in the original ``proof'' have been
discussed in the literature
\cite{Sharpe:1998xm,Azcoiti:1999rq,Cohen:2001hf,Ji:2001sa,Einhorn:2002rm,Crompton:2005ap,Azcoiti:2008nq},
and therefore it cannot be regarded as an established mathematical
theorem. Second, our result does not contradict recent work
\cite{Azcoiti:2008nq} on a Vafa-Witten-type theorem for fermion
bilinears, since the positivity of the probability distribution
function of the observable, which plays an essential role in
\cite{Azcoiti:2008nq}, is not ensured for purely imaginary masses.

We end with a brief remark on the rotation of the condensate.  It is
known from chiral perturbation theory that the chiral condensate at
$\mu=0$ rotates into the diquark condensate for $\mu>m_{\pi}/2$. For
real positive mass, $\langle \bar\psi\psi\rangle$ at $\mu=0$ rotates
into the $0^+$ diquark condensate $\langle\psi C\gamma_5\psi\rangle$
\cite{Kogut:2000ek}.  For imaginary mass (and $N_f=4n$), $\langle
\bar\psi i\gamma_5\psi\rangle$ at $\mu=0$ rotates into the $0^-$
diquark condensate $\langle\psi C\psi\rangle$.%
\footnote{More general cases, not necessarily preserving the
  positivity of the measure, were addressed in \cite{Metlitski:2005db}
  based on chiral perturbation theory with a nonzero $\theta$-angle.}

\subsubsection{Two-color QCD with diquark sources}

We now turn to two-color QCD in the chiral limit with diquark sources.
In this subsection, the transpose and the adjoint again act only on
color, spinor, and flavor indices, with one exception: as in the other
parts of this paper the adjoint in $D^\dagger$ acts on all indices.

In calculations of correlators, it is important to know whether there
is a disconnected piece or not.  If there is, the contribution of the
gluonic intermediate states will invalidate the derivation of QCD
inequalities.  In the previous subsection, the annihilation of the
diquark into gluons was trivially prohibited by the $\U(1)_B$ charge
conservation, but the latter is explicitly violated once we insert
diquark sources.  Indeed, in the cases of $J_R=\pm J_L=jI$ considered
below, the correlators of $\psi^TC\tau_2I\psi$ and
$\psi^TC\tau_2\gamma_5I\psi$ have disconnected pieces, and QCD
inequalities do not apply. We must limit ourselves to those
correlators that have no disconnected piece.

The propagators in a fixed gauge field are given by the inverse of $W$
defined in \eqref{eq:W},
\begin{align}
  \bep
    S_{\psi\psi}(x,y) & S_{\psi\bar\psi}(x,y) 
    \\ 
    S_{\bar\psi\psi}(x,y) & S_{\bar\psi\bar\psi}(x,y)
  \eep
  & \equiv 
  \bep  
    \ev{\psi(x)\psi^T(y)}_\psi & \ev{\psi(x)\bar\psi(y)}_\psi
    \\
    \big\langle\bar\psi^T\!(x)\psi^T(y)\big\rangle_\psi 
    & 
    \big\langle \bar\psi^T\!(x)\bar\psi(y)\big\rangle_\psi
  \eep
  = \frac{1}{2}\langle x|W^{-1}|y\rangle\,. 
\end{align}
The operator inverse $W^{-1}$ depends on the choice of the diquark
sources.  We first discuss the diquark sources with positive parity
and then those with negative parity.  As stated above we work in the
chiral limit.

For the $0^+$ diquark source, i.e., $J_R=-J_L=jI$ with real $j$, we
find%
\footnote{Note that the denominator in $(W^{-1})_{22}$ contains
  $DD^\dagger$ and not $D^\dagger D$. Using $W$ in \eqref{eq:W} one
  can explicitly confirm $WW^{-1}=\1$.}
\begin{equation}
  \label{eq:Winv}
  W^{-1} = 
  \bep 
    \displaystyle 
    -\frac{1}{j^2+D^\dagger D}jC\tau_2\gamma_5 I & 
    \displaystyle 
    \frac{1}{j^2+D^\dagger D}D^\dagger
    \\
    \displaystyle  
    C\tau_2 D\frac{1}{j^2+D^\dagger D} C\tau_2 & 
    \displaystyle  
    jC\tau_2\gamma_5 I \frac{1}{j^2+DD^\dagger}
  \eep.
\end{equation}
The positivity of the measure is ensured for even $N_f$ (see section
\ref{app:complex-1}).  We now consider the correlator of the diquark
field $M(x)=\psi^T\Gamma\psi$, assuming that there is no disconnected
piece. A rerun of the arguments leading to \eqref{eq:CS} then shows
that
\begin{align}
  \ev{ M(x)M^\dagger (y) }_{\psi, A}
  & = \ev{ \psi^T\!(x)\Gamma\psi(x) \bar\psi(y)\bar\Gamma\,\bar\psi^T\!(y) }_{\psi, A}
  \qquad (\bar\Gamma \equiv \gamma_4\Gamma^\dagger \gamma_4)
  \notag \\
  & = -\ev{\tr\big[\Gamma S_{\psi\bar\psi}(x,y)\bar\Gamma S_{\bar\psi\psi}(y,x)\big]}_{A}
  \notag \\
  & \leq \ev{\tr\big[ S_{\psi\bar\psi}(x,y)S_{\psi\bar\psi}(x,y)^\dagger\big]}_{A}\,.
  \label{eq:SS_up}
\end{align}
We now use \eqref{eq:Winv} to show that\footnote{Note that the
  interchange of $x$ in $y$ in the second line of \eqref{eq:C20} is
  correct even though the outer adjoint does not act on space-time
  indices.}
\begin{align}
  S_{\psi\bar\psi}(x,y)^\dagger 
  & = \left(\frac{1}{2}\Big\langle x\Big|\frac{1}{j^2+D^\dagger D}D^\dagger \Big|y\Big\rangle\right)^\dagger
  \notag \\
  & = \frac{1}{2}\Big\langle y\Big|D\frac{1}{j^2+D^\dagger D}\Big|x\Big\rangle
  \notag \\
  & = C\tau_2S_{\bar\psi\psi}(y,x)C\tau_2\,.
  \label{eq:C20}
\end{align}
Using this relation in the second line of \eqref{eq:SS_up} we obtain
\begin{align}
  \ev{ M(x)M^\dagger (y) }_{\psi, A} = 
  -\ev{\tr\big[\Gamma S_{\psi\bar\psi}(x,y)\bar\Gamma 
  C\tau_2S_{\psi\bar\psi}(x,y)^\dagger C\tau_2\big]}_{A}\,.
\end{align}
It is easily shown that for both $\Gamma=C\tau_2\Gamma_f$ and
$\Gamma=C\tau_2\gamma_5\Gamma_f$, we find
\begin{equation}
  \ev{ M(x)M^\dagger (y) }_{\psi, A} = \tr[\Gamma_f\Gamma_f^\dagger]
  \ev{\tr\big[S_{\psi\bar\psi}(x,y)S_{\psi\bar\psi}(x,y)^\dagger\big]}_{A}\,,
\end{equation}
which saturates the inequality \eqref{eq:SS_up} up to a multiplicative
constant.  Thus $\psi^T C\tau_2\Gamma_f\psi$ and $\psi^T
C\tau_2\gamma_5\Gamma_f\psi$ are the lightest diquarks (with
degenerate masses).

From $S_{\psi\psi}\propto I$ we see that the nonexistence of
disconnected pieces is ensured if $\tr[I\Gamma_f]=0$. Recalling from
section~\ref{sec:ss_high} that an $N_f\times N_f$ antisymmetric matrix
can be expanded in the generators $t^A$ of $\U(N_f)/\Sp(N_f)$, we find
that this condition is satisfied for all $A\ne 0$ but not for
$A=0$. Thus we conclude that $\{\psi^T C\tau_2 t^A\psi\}_{A\ne 0}$ and
$\{\psi^T C\tau_2\gamma_5 t^A\psi\}_{A\ne 0}$ are the lightest
diquarks (with degenerate masses), whereas no information on $\psi^T
C\tau_2 I\psi$ and $\psi^T C\tau_2\gamma_5I\psi$ can be gained from
QCD inequalities.

The fields $\{\psi^T C\tau_2 t^A\psi\}_{A\ne 0}$ and $\{\psi^T
C\tau_2\gamma_5 t^A\psi\}_{A\ne 0}$ can be interpreted as the NG modes
associated with the spontaneous symmetry breaking $\SU(N_f)_R\times
\SU(N_f)_L\to\Sp(N_f)_R\times \Sp(N_f)_L$. However, this breaking is
caused by \emph{both} $\langle\psi^TC\tau_2I\psi\rangle$ and
$\langle\psi^TC\tau_2\gamma_5I\psi\rangle$, which means that our
result does not yield any constraint on the parity of the
condensate. (This is to be contrasted with QCD with quark masses and
no diquark sources, where we could determine the parity of the
condensate even though we could not apply QCD inequalities to the
correlators of $\sigma$ and $\eta'$.)  

We now turn to the $0^-$ diquark source.  For $J_R=J_L=jI$ with real
$j$ we find
\begin{equation}
  W^{-1} = 
  \bep 
    \displaystyle 
    -\frac{1}{j^2+D^\dagger D}jC\tau_2  I & 
    \displaystyle 
    \frac{1}{j^2+D^\dagger D}D^\dagger
    \\
    \displaystyle  
    C\tau_2 D\frac{1}{j^2+D^\dagger D} C\tau_2 & \;
    \displaystyle  
    - jC\tau_2  I \frac{1}{j^2+DD^\dagger}
  \eep.
\end{equation}
For the positivity of the measure we require $N_f=4n$ (see section
\ref{app:complex-1}).  Since the off-diagonal entries are not changed
from \eqref{eq:Winv}, our previous discussion goes through without
modifications, and therefore we reach the same conclusion that
$\{\psi^T C\tau_2 t^A\psi\}_{A\ne 0}$ and $\{\psi^T C\tau_2\gamma_5
t^A\psi\}_{A\ne 0}$ are the lightest diquarks. Again, the parity of
the condensate cannot be determined from this argument alone.

A further discussion of the parity of the diquark condensate is given
in the last paragraph of section~\ref{sec:three}.

\section{Anomaly and index theorem for non-Hermitian Dirac operators}
\label{app:index}

In this appendix we prove the following properties of the Euclidean
Dirac operator%
\footnote{The color generators are assumed to be in the fundamental
  representation of $\SU(N_c)$ for $N_c\geq 2$.}
\begin{align}
  \label{eq:mu_D_def}
  D(\mu) \equiv \gamma_\nu D_\nu+\mu \gamma_4
  \equiv\bep 0&D_L\\D_R&0\eep
\end{align} 
at zero temperature and nonzero quark chemical potential $\mu$ in a
given gauge field:
\begin{enumerate}
\item The chiral anomaly equation in the chiral limit,
  \begin{equation}
    \label{eq:A}
    \del_\nu J_{5\nu}= \frac{i N_f}{16\pi^2} F \tilde F
  \end{equation}
  with $J_{5\nu} = i \bar\psi \gamma_\nu \gamma_5 \psi$,  is unchanged.
\item We have
  \begin{subequations}
    \label{eq:BC}
    \begin{align}
      \label{eq:BCa}
      \dim\ker D_R-\dim \ker D_R^\dagger&=\nu\,,\\
      \label{eq:BCb}
      \dim\ker D_L-\dim \ker D_L^\dagger&= -\nu\,,
    \end{align}
  \end{subequations}
  where $\nu$ is the winding number of the gauge field, defined in
  \eqref{eq:nu_definition}.
\item For $\mu=0$, \eqref{eq:BC} reduces to the ordinary index
  theorem%
  \footnote{\label{ft:conv_diff2}Whether the r.h.s.\ of this equation
    is $\nu$ or $-\nu$ is determined by the convention for $\gamma_5$.
    Here we use $\gamma_5=\gamma_1\gamma_2\gamma_3\gamma_4$ (see
    appendix \ref{app:conv}), which leads to $+\nu$.  Some authors use
    the convention $\gamma_5=-\gamma_1\gamma_2\gamma_3\gamma_4$, which
    leads to $-\nu$. As mentioned in footnote~\ref{ft:conv_diff}, it
    seems that in \cite{Leutwyler:1992yt} (although not stated
    explicitly) the second convention is used, which differs from ours
    by a sign.}
  \begin{equation}
    \label{eq:D}
    \dim\ker D_R-\dim \ker D_L = \nu 
  \end{equation}
  since $D_R=-D_L^\dagger$ for $\mu=0$.
\end{enumerate}

The in-medium chiral anomaly and its impact on the phenomenology have
been discussed in numerous studies
\cite{Abrikosov:1980nx,Abrikosov:1981qb,AragaodeCarvalho:1980de,Shuryak:1982hk,GomezNicola:1992gg,GomezNicola:1994vq,Qian:1994pp,Alford:1997zt,Rapp:1997zu,Schafer:1998up,Carter:1998ji,Rapp:1999qa,Sannino:2000kg,Hsu:2000by,Son:2004tq,Hatsuda:2006ps,Chen:2009gv,Gavai:2009vb,Abuki:2010jq,Basler:2010xy,Zhang:2011xi}.
So far the observation \eqref{eq:A} has been made via perturbative
calculations of the triangle diagram
\cite{GomezNicola:1992gg,GomezNicola:1994vq,Qian:1994pp,Hsu:2000by,Gavai:2009vb}
as well as by Fujikawa's path integral method
\cite{GomezNicola:1994vq,Hsu:2000by,Gavai:2009vb}.  The latter
encounters some subtleties at $\mu\ne 0$ that were not stressed in the
preceding works.  Also, in contrast to the chiral anomaly equation,
the modification of the index theorem due to $\mu\ne 0$ was rarely
considered so far.

In this appendix we take a closer look at these issues.  After some
preliminaries in section~\ref{app:prelim}, we present in
section~\ref{sc:path_anom} two derivations of the anomaly equation
\eqref{eq:A} via the path integral method in a style that differs from
\cite{GomezNicola:1994vq,Hsu:2000by,Gavai:2009vb}.  In
section~\ref{sc:path_ind} we then study the index of the Dirac
operator at $\mu\ne 0$.  We will show that \eqref{eq:BC} can be proven
for generic gauge fields and $\mu\ne0$, while \eqref{eq:D} requires an
additional input.  In section~\ref{sc:path_discre} we will discuss
where the difference between \eqref{eq:BC} and \eqref{eq:D} comes
from.

Three caveats are in order. First, we note that Fujikawa's analysis in
its original form is equivalent to one-loop perturbation theory
\cite{Shizuya:1986am,Shizuya:1986hn,Joglekar:1986nj,Joglekar:1987rf}
although it is sometimes incorrectly said that it offers a
nonperturbative derivation of the anomaly.  Our analysis should be
seen as the extension of this equivalence to $\mu\ne 0$.  Second, our
analysis does not extend to lattice fermions directly. It is beyond
the scope of this paper to analyze the effects of $\mu$ in the lattice
formulation (see e.g.,
\cite{Bloch:2006cd,Bloch:2007xi,Narayanan:2011ff}).  Third, the entire
discussion in this appendix will be given for a fixed gauge field
background.

\subsection{Preliminaries}
\label{app:prelim}

\subsubsection{Path integral measure}
To clarify our stance on the subject and fix the notation, we begin by
reviewing Fujikawa's method at $\mu=0$
\cite{Bertlmann:1996xk,Fujikawa:2004cx}.  It starts with the expansion
of the fermion fields $\psi$ and $\bar\psi$ in terms of the
eigenstates $\{\psi_n\}$ of $D(0)$,
\begin{subequations}
  \begin{align}
    \psi(x)&=\sum_n a_n\psi_n(x)\equiv \sum_n a_n\langle x|n\rangle\,,\\
    \bar\psi(x)&=\sum_n \bar{b}_n\psi^\dagger_n(x) \equiv \sum_n
    \bar{b}_n\langle n|x\rangle\,, 
  \end{align}
\end{subequations}
where the $a_n$ and $\bar{b}_n$ are Grassmann variables.  The
transformations from $\psi$ and $\bar\psi$ to $\{a_n\}$ and
$\{\bar{b}_n\}$ are unitary by virtue of the orthonormality and
completeness of $\{\psi_n\}$, which are guaranteed since $D(0)$ is
anti-Hermitian.  The functional measure then transforms as
\begin{align}
  \prod_{x} \dd\psi(x)\,\dd\bar\psi(x)
  &=(\det\langle n|x\rangle)^{-1}(\det\langle x|m\rangle)^{-1}
  \prod_n da_n d\bar{b}_n
  \label{eq:measure}
  \nonumber \\
  &=(\det\langle n|m\rangle)^{-1}\prod_n da_n d\bar{b}_n
  =\prod_n da_n d\bar{b}_n\,,
\end{align}
where $\det \langle n|x\rangle$ and $\det \langle x|m\rangle$ each
represent the determinant of an infinite-dimensional matrix specified
by the indices $n$, $m$, and $x$.  By standard calculations the
evaluation of the Jacobian for an infinitesimal chiral transformation
leads to an infinite sum, $\tr \gamma_5 = \sum_n \psi_n(x)^\dagger
\gamma_5\psi_n(x)$, which is to be regularized by, e.g.,
$\ee^{D^2/\Lambda^2}$ with $\Lambda\to\infty$.

Two remarks are in order. First, for the change of variables
\eqref{eq:measure} to be unitary, the operator whose eigenstates are
used to expand the fermion fields must be (anti-) Hermitian. Second,
the final result does depend on the choice of the operator. Had we
used the plane-wave basis, the result would simply have vanished:
$\tr(\gamma_5 \ee^{\del^2/\Lambda^2})=0$.  This is sometimes ascribed
to the lack of gauge invariance, but the gauge invariance is merely a
necessary condition.  Indeed there is an instructive example,
$\tr(\gamma_5 \ee^{D_\nu D_\nu/\Lambda^2})=0$, which reveals that a
seemingly unitary transformation from the gauge invariant eigenspace
of $D_\nu D_\nu$ to that of $D$ is actually non-unitary in the
presence of regularization \cite{Fujikawa:1979ay}.

The bottom line is that the operator used to define the functional
space must be chosen so as to diagonalize the fermion action
\cite{Fujikawa:2004cx,Srednicki:2007qs}.  The operators $D_\nu D_\nu$
and $\partial^2$ are not eligible, because their eigenbases do not
diagonalize the action.  In the following, we will work with this
criterion as a guiding principle.  This completes our preliminary
comments on the chiral anomaly at $\mu=0$.

\subsubsection{Remarks on the literature}

$D(\mu)$ is no longer anti-Hermitian at $\mu\ne 0$ and the eigenstates
lose orthogonality, spoiling the unitary transformation
\eqref{eq:measure}.  This requires us to choose a more appropriate
definition of the path integral measure.
\begin{itemize}
\item In \cite{GomezNicola:1994vq} the fermion fields were expanded in
  eigenstates of $D(i\mu)$, the Dirac operator with imaginary
  chemical potential. With $D(i\mu)$ being Hermitian (in their
  convention), the eigenstates are orthonormal and the anomaly equation
  follows straightforwardly.
\item In \cite{Hsu:2000by} the fermion fields were expanded in
  eigenstates of $D(0)$. Again the derivation of the anomaly equation
  is straightforward.
\item In \cite{Gavai:2009vb} the fermion fields were expanded in right
  and left eigenstates of $D(\mu)$, but the subtlety is that the
  ``eigenstates'' were simply defined by multiplying the eigenstates
  of $D(0)$ by $\ee^{\pm \mu x_4}$. They either blow up exponentially
  at infinity for $T=0$, or spoil the anti-periodic boundary condition
  in the $x_4$ direction for $T>0$. Note also that this multiplication
  does not change the eigenvalues from purely imaginary ones at
  $\mu=0$.  We suspect that the eigenstates thus obtained are not the
  legitimate eigenstates of $D(\mu)$ (see, e.g.,
  \cite{AragaodeCarvalho:1980de,Schafer:1998up}) since the
  contribution from occupied states below the Fermi surface is not
  taken into account \cite{Schafer:1998up}.
\end{itemize}
The schemes based on $D(i\mu)$ and $D(0)$
\cite{GomezNicola:1994vq,Hsu:2000by} do not meet our criterion, as the
eigenstates of $D(i\mu)$ and $D(0)$ do not diagonalize the action.
While the final result for the anomaly equation obtained in
\cite{GomezNicola:1994vq,Hsu:2000by,Gavai:2009vb} is correct, the
rigorousness of these approaches is not completely obvious to
us.\footnote{Note also that the index theorem derived in
  \cite{Gavai:2009vb} corresponds to \eqref{eq:D}, while the correct
  result at $\mu\ne0$ is \eqref{eq:BC} if there are accidental zero
  modes.}

\subsection[Proofs of (\ref{eq:A}): Path integral derivation of the
anomaly at $\mu\ne 0$]{\boldmath Proofs of (\ref{eq:A}): Path integral
  derivation of the anomaly at $\mu\ne 0$}
\label{sc:path_anom}

In this subsection we will present new derivations of the anomaly
equation \eqref{eq:A} by the path integral method, which we believe
place the results in the literature on a firmer footing.  For
simplicity we assume $N_f=1$. The extension to $N_f>1$ should be
straightforward.

\subsubsection[Derivation based on $D^\dagger D$]{\boldmath Derivation
  based on $D^\dagger D$}
\label{sc:anomaly-Derivation-1}

A non-Hermitian Dirac operator occurs not only in QCD at $\mu \ne 0$
but also in chiral gauge theories where the gauge interaction involves
a $\gamma_5$-coupling.  We shall apply one of the methods devised for
this situation \cite{Fujikawa:1983bg}%
\footnote{In chiral gauge theories this regularization is known to
  give the so-called covariant anomaly \cite{Bertlmann:1996xk}.  } to
QCD at $\mu \ne 0$.  We use the same orthonormal bases $\{\phi_n\}$
and $\{\tilde\phi_n\}$ as in \eqref{eq:def_phi}.  Note again that the
$\xi_n$ are not the eigenvalues but rather the singular values of $D$,
with $\xi_n\geq 0$ for all $n$.  The orthogonality and completeness of
these bases follow from the Hermiticity of $D^\dagger D$ and
$DD^\dagger$.  The fields can be expanded as
\begin{align}
  \psi(x)=\sum_n a_n\phi_n(x)\,,
  \qquad 
  \bar\psi(x)=\sum_n\bar{b}_n\tilde\phi_n^\dagger(x)\,,
  \label{eq:psi_expand_}
\end{align}
and the fermionic part of the action is then diagonalized,
\begin{align}
  S=\int d^4x~\bar\psi(x) D \psi(x)=\sum_n \xi_n\bar{b}_na_n\,.
\end{align}
Therefore this scheme meets the criterion mentioned above. Note that,
had we expanded both $\psi$ and $\bar\psi$ in only one of the bases,
the action would not have been diagonalized.

Under an infinitesimal chiral transformation $\psi\to
\ee^{i\alpha\gamma_5}\psi$,
$\bar\psi\to\bar\psi\,\ee^{i\alpha\gamma_5}$, the action changes as
\begin{align}
  \ee^{-S}\to \ee^{-S'}=\exp\Big[
  -S + \int d^4x~\alpha(x)\del_\nu J_{5\nu}(x)
  \Big]\,,
  \label{eq:change_action_}
\end{align}
while the functional measure also changes as
\begin{align}
  \dd\psi\,\dd\bar\psi   &\to   
  \dd\psi\,\dd\bar\psi\,
  \exp\Big\{-i\int d^4x\,\alpha(x) \sum_n 
  \Big[\varphi_n^\dagger(x) \gamma_5 \varphi_n(x)
  +\tilde\phi_n^\dagger(x) \gamma_5 \tilde\phi_n(x)\Big] \Big\}
  \nonumber \\
  & \equiv
  \lim_{\Lambda\to\infty}
  \dd\psi\,\dd\bar\psi\,
  \exp\Big\{-i\int d^4x\,\alpha(x)  
  \sum_n \Big[\phi_n^\dagger (x)\gamma_5 
  \ee^{-D^\dagger D/\Lambda^2}\phi_n(x)
  \notag
  \\
  &\qquad \qquad \qquad \qquad \quad\;+\tilde\phi_n^\dagger(x) 
  \gamma_5 \ee^{-DD^\dagger/\Lambda^2}
  \tilde\phi_n(x)\Big] \Big\}\,,
  \label{eq:measu_sum}
\end{align}
where a Gaussian cutoff has been employed.  Because the phase factor
$\alpha(x)$ is arbitrary, we have from \eqref{eq:change_action_} and
\eqref{eq:measu_sum}
\begin{align}
  \del_\nu J_{5\nu} & 
  = i \lim_{\Lambda\to\infty}
  \sum_n \Big[\phi_n^\dagger (x)\gamma_5 
  \ee^{-D^\dagger D/\Lambda^2}\phi_n(x)
  +\tilde\phi_n^\dagger(x) 
  \gamma_5 \ee^{-DD^\dagger/\Lambda^2}
  \tilde\phi_n(x)\Big] \notag\\
  \label{eq:del_j5-2}
  &= i \frac{1}{32\pi^2}\varepsilon_{\alpha\beta\gamma\delta}
  F^a_{\alpha\beta}F^a_{\gamma\delta}\,,
\end{align}
where the last line is obtained using standard algebra involving the
plane-wave basis.  The $\mu$-dependent terms appear in the calculation
but disappear from the final result.  We have thus obtained the
conventional anomaly equation \eqref{eq:A}.

\subsubsection[Derivation based on $D$]{\boldmath Derivation based on
  $D$}
\label{sc:anomaly-Derivation-2}

In the second derivation, we apply the method of
\cite{AlvarezGaume:1983cs,AlvarezGaume:1984dr} to QCD at $\mu \ne 0$.%
\footnote{In chiral gauge theories this regularization is known to
  give the so-called consistent anomaly satisfying the Wess-Zumino
  consistency condition \cite{Bertlmann:1996xk}.  }  Let us denote the
right and left eigenvectors of $D(\mu)$ by $\{\psi_n\}$ and
$\{\chi_n\}$,
\begin{align}
  D(\mu)\psi_n = \lambda_n \psi_n\,, 
  \qquad 
  \chi_n^\dagger D(\mu) = \lambda_n \chi_n^\dagger \,. 
\end{align}
The eigenvalues $\lambda_n$ are no longer purely imaginary but complex
in general.  We assume\footnote{This assumption can be violated on a
  gauge field set of measure zero, see section~\ref{sc:path_discre},
  and then the derivation no longer works.}  that they satisfy the
following orthonormality and completeness relations,%
\footnote{The conditions in \eqref{eq:ornorm} say nothing about inner
  products such as $\int d^4x\ \psi_m^\dagger(x)\psi_n(x)$.}
\begin{align}
  \int d^4x~\chi_m^\dagger (x) \psi_n(x) = \delta_{mn}\,,\qquad 
  \sum_{n} \psi_n(x)\chi_n^\dagger (y) = \delta^4(x-y)\,.
  \label{eq:ornorm}
\end{align}
The orthonormalization is possible because the left and right
eigenstates corresponding to different eigenvalues are orthogonal,
as can be seen from
\begin{align}
  \lambda_m \int d^4x\,\chi_m^\dagger \psi_n
   = \int d^4x\,(\chi_m^\dagger D) \psi_n 
  = \int d^4x\,\chi_m^\dagger (D \psi_n)
   = \lambda_n \int d^4x\,\chi_m^\dagger \psi_n
  \,.
\end{align}
Let us expand the fields $\psi$ and $\bar\psi$ in the sets
$\{\psi_n\}$ and $\{\chi_n^\dagger\}$, respectively,
\begin{align}
  \label{eq:psi_expand_2}
  \psi(x) = \sum_n a_n\psi_n(x) ,  \qquad 
  \bar{\psi}(x) = \sum_n \bar{b}_n \chi_n^\dagger(x) \,,
\end{align} 
by which the fermion action is diagonalized as desired,
\begin{align}
  S = \int d^4x\,\bar\psi D \psi 
     = \sum_k \sum_\ell \int d^4x\, \bar b_k\chi_k^\dagger(x)
         D a_\ell \psi_\ell(x)
     = \sum_k \lambda_k \bar b_k a_k\,.
\end{align}
Thus this choice of bases also satisfies our criterion.  After some
algebra, the change of the measure under an infinitesimal chiral
rotation is found to be
\begin{align}
  \dd\psi \dd\bar\psi
  & \to \dd\psi \dd\bar\psi 
  \exp\Big[-i\int d^4x\, 2\alpha(x)\sum_n 
  \chi_n^\dagger(x) \gamma_5 \psi_n(x)\Big]
  \label{eq:jacob}
  \notag \\
  & = \lim_{\Lambda\to\infty}
  \dd\psi \dd\bar\psi 
  \exp\Big[-i\int d^4x\, 2\alpha(x)\sum_n 
  \chi_n^\dagger(x) \gamma_5 \frac{1}{1-D^2/\Lambda^2} \psi_n(x)\Big]\,.
\end{align}
We avoided a Gaussian cutoff because complex eigenvalues are not
suppressed by a Gaussian factor.  Using the assumed completeness
\eqref{eq:ornorm} we can move to the plane-wave basis. After some
standard algebra, we find that the $\mu$-dependent terms disappear
from the final result and recover the anomaly equation \eqref{eq:A}.

\subsection[Proofs of (\ref{eq:BC}): Index theorem at $\mu \ne
0$]{\boldmath Proofs of (\ref{eq:BC}): Index theorem at $\mu \ne 0$}
\label{sc:path_ind}

In the following we present two derivations of the index theorem at
$\mu \ne 0$.  The first derivation yields \eqref{eq:BC}, while the
second one yields \eqref{eq:D}.  We analyze the origin of this
discrepancy and relate it to the (in-) completeness of the bases.

\subsubsection[Derivation based on $D^\dagger D$]{\boldmath Derivation
  based on $D^\dagger D$}
\label{sc:index-Derivation-1}

Let us reexamine the analysis in
section~\ref{sc:anomaly-Derivation-1}.  From \eqref{eq:del_j5-2} it
follows that
\begin{align}
  \lim_{\Lambda\to\infty} 
  \sum_n \int d^4x\ 
  \Big[\phi_n^\dagger (x)\gamma_5 
  \ee^{-D^\dagger D/\Lambda^2}\phi_n(x)
  +\tilde\phi_n^\dagger(x) 
  \gamma_5 \ee^{-DD^\dagger/\Lambda^2}
  \tilde\phi_n(x)\Big]
  =  2\nu\,.
\end{align}
Using \eqref{eq:rel_bases} one can show that the entire contribution
to the l.h.s.\ comes solely from the eigenstates with $\xi_n=0$, for
any finite $\Lambda$.  This observation, combined with
\begin{align}
  \sum_{n:\,\xi_n=0} 
  \int d^4x\, \phi_n^\dagger (x)\gamma_5 
  \ee^{-D^\dagger D/\Lambda^2}\phi_n(x)
  = 
  \dim \ker D_R - \dim \ker D_L
\end{align}
and 
\begin{align}
  \sum_{n:\,\xi_n=0} 
  \int d^4x\, \tilde\phi_n^\dagger (x)\gamma_5 
  \ee^{-DD^\dagger/\Lambda^2}\tilde\phi_n(x)
  = 
  \dim \ker D^\dagger_L - \dim \ker D^\dagger_R \,,
\end{align}
proves the equality 
\begin{align}
  \dim \ker D_R - \dim \ker D_L + 
  \dim \ker D^\dagger_L - \dim \ker D^\dagger_R = 2\nu\,.
  \label{eq:dim_eq_1}
\end{align}
On the other hand, using the identity
\begin{align}
  \label{eq:trace}
  0 & = \tr\big( \ee^{-DD^\dagger/\Lambda^2}
  - \ee^{-D^\dagger D/\Lambda^2} \big)
\end{align}
that follows from the cyclic invariance of the trace\footnote{It was
  pointed out in \cite{Weinberg:1979ma} that the cyclic invariance of
  the trace could potentially break down.  This does not happen in the
  instanton background, see \cite{Weinberg:1979ma} or section~4.2 of
  \cite{Vandoren:2008xg}.  This possibility does not invalidate our
  result \eqref{eq:BC} since we have also constructed an alternative
  proof of \eqref{eq:BC}, using the eigenbases of $D_5(\mu)$ and
  $D_5(-\mu)$, which does not make use of \eqref{eq:trace}.  For
  brevity we do not show this proof here.}  and the fact that all
nonzero eigenvalues of $D^\dagger D$ and $DD^\dagger$ coincide, we
have
\begin{align}
  \dim \ker D & = \dim \ker D^\dagger\,,
\end{align}
or equivalently
\begin{align}
  \dim \ker D_R + \dim \ker D_L 
  & =  \dim \ker D^\dagger_R +  \dim \ker D^\dagger_L\,. 
  \label{eq:dim_eq_2}
\end{align}
Then \eqref{eq:dim_eq_1} and \eqref{eq:dim_eq_2} prove 
\eqref{eq:BC}.

For $\nu\ge0$, $\dim\ker D_R\geq \nu$ follows from \eqref{eq:BCa}.
For $\nu<0$, $\dim\ker D_L\geq |\nu|=-\nu$ follows from
\eqref{eq:BCb}.  Thus our index theorem at $\mu\ne 0$ shows that
$D(\mu)$ must possess at least $|\nu|$ zero modes for any $\mu$.  This
is consistent with the existence of exact zero modes in the instanton
background at $\mu\ne 0$
\cite{Abrikosov:1981qb,AragaodeCarvalho:1980de}.%
\footnote{%
  Note that the fermionic zero modes in the instanton background at
  $\mu\ne 0$ have a norm that diverges logarithmically in the spatial
  volume \cite{AragaodeCarvalho:1980de}.  }

\subsubsection[Derivation based on $D$]{\boldmath Derivation based on
  $D$}
\label{sc:index-Derivation-2}

Let us return to the second derivation of the anomaly in
section~\ref{sc:anomaly-Derivation-2}.  From \eqref{eq:jacob} and the
anomaly equation \eqref{eq:A} it follows that
\begin{align}
  \lim_{\Lambda\to\infty}
  \int d^4x\, \sum_n 
  \chi_n^\dagger(x) \gamma_5 \frac{1}{1-D^2/\Lambda^2} \psi_n(x) = \nu\,.
  \label{eq:sum_:}
\end{align}
If $\lambda_n\ne 0$, we have
\begin{align}
  \lambda_n \chi_n^\dagger(x) \gamma_5
  \frac{1}{1-\lambda_n^2/\Lambda^2} \psi_n(x) 
  & = [\chi_n^\dagger(x) D] \gamma_5 \frac{1}{1-\lambda_n^2/\Lambda^2}
  \psi_n(x) 
  \nonumber \\
  & = - \chi_n^\dagger(x)\gamma_5 \frac{1}{1-\lambda_n^2/\Lambda^2} D
  \psi_n(x) 
  \nonumber \\
  & = - \lambda_n \chi_n^\dagger(x)\gamma_5
  \frac{1}{1-\lambda_n^2/\Lambda^2} \psi_n(x)\,, 
\end{align}
and hence all contributions to the l.h.s.\ of \eqref{eq:sum_:} from
nonzero modes vanish.  Restricting the sum in \eqref{eq:sum_:} to zero
modes, we arrive at
\begin{align}
  \dim\ker D_R - \dim\ker D_L \overset?= \nu\,.
\end{align}
This ``proves'' \eqref{eq:D}, i.e., the index theorem in the same form
as at $\mu=0$.  It does not imply \eqref{eq:BC} proved in
section~\ref{sc:index-Derivation-1}.  The origin of this discrepancy
is analyzed below.

\subsubsection{Analysis of the discrepancy}
\label{sc:path_discre}

The arguments in section~\ref{sc:index-Derivation-1} proved
\eqref{eq:BC}, whereas in section~\ref{sc:index-Derivation-2} we were
led to \eqref{eq:D}.  Which is the true index theorem at $\mu\ne 0$?

To answer this question, let us begin by recalling a standard argument
at $\mu=0$ about the stability of the index of the Dirac operator
against small perturbations.  Suppose there are $\nu$ right-handed
zero modes and no left-handed zero modes in a given background gauge
field. Now we deform the gauge field smoothly so that two of the
eigenvalues $\lambda$ and $-\lambda$ (with eigenstates $\psi$ and
$\gamma_5\psi$, respectively) approach zero. When $\lambda=0$ is
achieved, we have two accidental zero modes, which we can arrange to
be eigenstates of $\gamma_5$ by choosing the linear combinations
$(\1+\gamma_5)\psi$ and $(\1-\gamma_5)\psi$.  Thus the number of
right-handed (left-handed) zero modes becomes $\nu+1$ (1), but their
difference never changes: $(\nu+1)-1=\nu$. This is why accidental zero
modes cannot change the index.

This argument is correct at $\mu=0$. If it could be extended to
$\mu\ne 0$, the accidental zero modes could not change $\ind
D=\dim\ker D_R-\dim\ker D_L$ and the ordinary index theorem
\eqref{eq:D} would continue to hold.  However, the argument does not
work for $\mu\ne 0$.  The pitfall is that, when $D(\mu)$ is not
anti-Hermitian, $\psi$ and $\gamma_5 \psi$ may become linearly
dependent in the limit $\lambda \to 0$.  Let us exemplify this by a
toy model which mimics $D(\mu)$:
\begin{align}
  D_\text{toy}(\alpha) = \bep 0&0&2&-4 \\ 0&0&1&2 \\ 1&2&0&0 \\
  0&\alpha &0&0 \eep \quad \text{with eigenvalues}\quad \{ \pm 2,\,
  \pm\sqrt{2\alpha} \}\,.
\end{align}
The right eigenvectors of $D_\text{toy}(\alpha)$ associated with the
eigenvalues $\sqrt{2\alpha}$ and $-\sqrt{2\alpha}$ are
$(4,-2,0,-\sqrt{2\alpha})^T$ and $(4,-2,0,\sqrt{2\alpha})^T$,
respectively. In the limit $\alpha\to 0$ they coincide, implying that
$\dim \ker D_R$ increases by 1 whereas $\dim \ker D_L$ remains zero.
The accidental zero modes changed the index.

At $\alpha=0$ the dimension of the eigenspace of
$D_\text{toy}(\alpha)$ is $3\ (<4)$.  This means that the eigenvectors
fail to form a complete basis, a phenomenon that occurs only for
non-Hermitian matrices. The interpretation of the discrepancy between
the last two subsections is now straightforward: Our derivation in
section~\ref{sc:index-Derivation-2} led to \eqref{eq:D} because we
assumed the completeness of the basis in \eqref{eq:ornorm}.  If the
completeness is not ensured, it is not possible to expand an arbitrary
fermion field in this basis as in \eqref{eq:psi_expand_}, and then our
discussions in section~\ref{sc:anomaly-Derivation-2} and
\ref{sc:index-Derivation-2} break down.%
\footnote{Although our discussion here is based on finite-dimensional
  matrices, we believe that the essential part of the argument carries
  over to the actual Dirac operator as an infinite-dimensional
  matrix.}  The index theorem that holds for generic gauge fields and
$\mu\ne0$ is not \eqref{eq:D}, but \eqref{eq:BC}.

However, the incompleteness of the basis requires a fine-tuning of the
gauge field, and so we expect that this occurs only on a gauge field
set of measure zero. For practical calculations it is justified to
neglect the possibility of incompleteness, in the same sense as we can
neglect accidental (non-topological) zero modes at $\mu=0$. Thus both
\eqref{eq:BC} and \eqref{eq:D} will be valid for almost all gauge
fields and $\mu\ne0$. Summarizing, we have
\begin{subequations}
  \begin{alignat}{3}
    &\text{for }\nu\geq 0:\;\; && \dim \ker D_R \overset{\text{a.s.}}{=}
    \dim \ker D_L^\dagger \overset{\text{a.s.}}{=} \nu\,,
    & \dim \ker D_R^\dagger \overset{\text{a.s.}}{=} \dim \ker D_L
    \overset{\text{a.s.}}{=} 0\,, 
    \\
    &\text{for }\nu<0: && \dim \ker D_R^\dagger \overset{\text{a.s.}}{=}
    \dim \ker D_L \overset{\text{a.s.}}{=} -\nu\,,\;\;\;
    & \dim \ker D_R \overset{\text{a.s.}}{=} \dim \ker D_L^\dagger
    \overset{\text{a.s.}}{=} 0\,, 
  \end{alignat}
\end{subequations}
where $\overset{\text{a.s.}}{=}$ denotes an equality that holds
``almost surely''. This completes our discussion of the index theorem
at $\mu\ne 0$.

\section{Derivation of (\ref{eq:rho_pert})} 
\label{app:rho_pert}

In this appendix we outline the derivation of the singular value
density in the free limit in \eqref{eq:rho_pert}. The basic relation
we use is
\begin{align}
  \rsv(\xi) = -\frac{2\xi}{\pi} \im \ev{
  \sum_{n}\frac{1}{\xi_n^2-\xi^2+i\epsilon} } \,,
\end{align}
where $\epsilon\to 0^+$ is tacitly assumed.  To obtain the resolvent
on the r.h.s., we express it in terms of the free Dirac operator
$D(\mu)= (\gamma_\nu \del_\nu+\mu\gamma_4)\otimes \1_{N_c}~(\mu>0)$
and move to momentum space,
\begin{align}
   \ev{ \sum_n \frac{1}{\xi_n^2 - \xi^2+i\epsilon} }
   & = \ev{ \tr \frac{1}{D(\mu)^\dagger D(\mu) - \xi^2+i\epsilon} }
   \\
   & = V_4 N_c \int \frac{d^4p}{(2\pi)^4} 
   \tr \frac{1}{(-i\slashed{p}+\mu\gamma_4)(i\slashed{p}+\mu\gamma_4)-
       \xi^2+i\epsilon } \,.
\end{align}
Then
\begin{align}
  \im \ev{ \sum_n \frac{1}{\xi_n^2 - \xi^2+i\epsilon} } 
  = - 4\pi V_4 N_c  \int \frac{d^4p}{(2\pi)^4}  \left|p^2+\mu^2-\xi^2\right|
  \, \delta\left( (p^2+\mu^2-\xi^2)^2-4\mu^2\mathbf{p}^2 \right),
\end{align}
where $\mathbf{p}=(p_1,p_2,p_3)$ denotes the three-momentum.  After
elementary but tedious algebra this integral can be performed
analytically and \eqref{eq:rho_pert} follows.

\section[Random matrix theory for QCD with isospin chemical potential
\boldmath ($\beta=2$)]{Random matrix theory for QCD with isospin
  chemical potential\\ \boldmath ($\beta=2$)}
\label{app:isospin}

In this appendix, we comment on the random matrix theory describing
QCD with three colors, two flavors, and isospin chemical potential
$\mu_I=2\mu$, see sections~\ref{sec:isospin} and \ref{sec:iso}.
Starting from the two-matrix model of \cite{Osborn:2004rf} the random
matrix theory for this case in the chiral limit is
\begin{align}
  Z_{\nu,N_f=2}^\text{iso}(\hat\mu, \hat\rho_1,\hat\rho_2) 
  &=\int dA\,dB\,\ee^{-N\tr(A^\dagger A+B^\dagger B)}
  \det\begin{pmatrix}
    D(\hat\mu) & X\\
    Y & D(-\hat\mu)
  \end{pmatrix}
\end{align}
with
\begin{align}    
  D(\hat\mu)&=\bep
    0 & iA+\hat\mu B\\
    iA^\dagger+\hat\mu B^\dagger & 0
  \eep,
\end{align}
where $A$ and $B$ are complex $N\times(N+\nu)$ matrices and
\begin{align}
  X=\bep
    \hat\rho_1\1_N & 0\\
    0 & \hat\rho_2\1_{N+\nu}
  \eep
  \quad\text{and}\quad
  Y=\bep
    \hat\rho_2\1_N & 0\\
    0 & \hat\rho_1\1_{N+\nu}
  \eep
\end{align}
are source terms for the pion condensate.  This model can be extended
to $N_f$ pairs of quarks and conjugate quarks.  Introducing
$P=iA+\hat\mu B$ and $Q=iA^\dagger+\hat\mu B^\dagger$, we obtain
\begin{align}
  \label{eq:Ziso}
  Z_{\nu, N_f=2}^\text{iso}(\hat\mu,\hat\rho_1,\hat\rho_2) = \int dP\,dQ\,
  &\exp\left\{-\frac{N(1+\hat\mu^2)}{4\hat\mu^2}
  \left[  \tr(P^\dagger P+Q^\dagger Q) + 
    \tr(PQ+Q^\dagger P^\dagger)
  \right]\right\}
  \notag
  \\
  &\times\hat\rho_1^{-\nu}\hat\rho_2^\nu
  \det(P^\dagger P+\hat\rho_1^2\1_{N+\nu})
  \det(Q^\dagger Q+\hat\rho_2^2\1_N)\,.
\end{align}
To study the singular values of $D(\hat\mu)$ we need to express the
partition function in terms of the eigenvalues of $P^\dagger P$ and
$Q^\dagger Q$.  Now, we notice that \eqref{eq:Ziso} is identical to
eqs.~(2.5)--(2.8) of \cite{Akemann:2006ru}.\footnote{It is interesting
  that the model of \cite{Akemann:2006ru} is useful here since it
  applies to QCD with imaginary chemical potential, which is unrelated
  to QCD with real isospin chemical potential.}  The representation of
\eqref{eq:Ziso} in terms of the eigenvalues of $P^\dagger P$ and
$Q^\dagger Q$ is given in eq.~(2.14) of that reference, and the
microscopic spectral correlations were also computed in
\cite{Akemann:2006ru}.  Hence, the microscopic correlations of the
singular values of $D(\hat\mu)$ can be obtained immediately from the
results of \cite{Akemann:2006ru}.  The mapping of RMT parameters to
physical quantities can be worked out similarly to the $\beta=1$ case
in section~\ref{sec:ls}.

\section{Derivation of (\ref{eq:sumlow0})}
\label{app:sumlow0}

From \eqref{eq:sumlownu} we have for $\nu=0$
\begin{align}
  \label{eq:df}
  \bigg\langle {\sum_n}'\frac1{\xi_n^2}\bigg\rangle_0
  =2(V_4\Phi_\L)^2\frac\partial{\partial z}\ln f(z)
  \quad\text{with}\quad
  f(z)=\int_{S^5}d\vec n\,\ee^{z(n_2^2+n_4^2)}\,.
\end{align}
Performing a Hubbard-Stratonovich transformation gives
\begin{align}
  f(z)&=\frac1\pi\int_{-\infty}^\infty  dp\,dq\,\ee^{-p^2+q^2}
  \int_{S^5}d\vec n\,\ee^{2\sqrt z(pn_2+qn_4)}\,.
\end{align}
We now define $\vec a=(0,2\sqrt zp,0,2\sqrt zq,0,0)$ and then rotate
the coordinate system such that $\vec a$ is parallel to $\hat n_1$ to
obtain
\begin{align}
  \int_{S^5} d\vec n\,\ee^{2\sqrt z(pn_2+qn_4)}
  &=\int_{S^5}d\vec n\,\ee^{an_1}\quad\text{with}\quad a=|\vec a|\notag\\
  &=\frac{8\pi^2}3\int_0^\pi d\phi\,{\sin}^4\phi\,\ee^{a\cos\phi}
  =\frac{8\pi^3}{a^2}I_2(a)\,.
\end{align}
After introducing polar coordinates for $p$ and $q$ we end up with
\begin{align}
  f(z)&=\frac{4\pi^3}z\int_0^\infty  dR\,
  \frac{\ee^{-R^2}}R\,I_2(2\sqrt zR)
  =\pi^3 M(1,3,z)\,,
\end{align}
where $M$ is Kummer's function (a.k.a.\ the confluent hypergeometric
function).  Using the recurrence relations for $M$ we obtain
\begin{align}
  M(1,3,z)=\frac2z\left(\frac{\ee^z-1}{z}-1\right),
\end{align}
and after performing the differentiation according to \eqref{eq:df} we
obtain \eqref{eq:sumlow0}.

\bibliography{Ref_for_paper}
\bibliographystyle{JBJHEP}
\end{document}